\documentclass[11pt]{article}
\usepackage{graphicx}
\usepackage{epstopdf}
\usepackage{amssymb,bm}
\usepackage{bbm}
\usepackage{amsmath}
\usepackage{amsfonts}
\usepackage[margin=1.0in]{geometry}
\usepackage{cite}
\usepackage[shortlabels]{enumitem}
\usepackage{comment}

\usepackage{color}

\DeclareMathOperator{\E}{\mathbbmss{E}}

\DeclareMathOperator{\Hrm}{\mathsf{H}}

\allowdisplaybreaks

\usepackage{amsthm}
\newtheorem{theorem}{Theorem}
\newtheorem*{theorem*}{Theorem}
\newtheorem{lemma}{Lemma}

\theoremstyle{definition}  
\newtheorem{definition}{Definition}
\newtheorem{remark}{Remark}
\newtheorem{corollary}{Corollary}

\usepackage{chngcntr}
\counterwithout{theorem}{section}
\counterwithout{definition}{section}
\counterwithout{lemma}{section}
\counterwithout{remark}{section}
\counterwithout{assumption}{section}
\counterwithout{proposition}{section}
\counterwithout{corollary}{section}
\counterwithout{claim}{section}

\begin{document}
\title{On the Optimality of Treating Inter-Cell Interference as Noise: Downlink Cellular Networks and Uplink-Downlink Duality
\footnotetext{Hamdi Joudeh and Giuseppe Caire are with the Faculty of Electrical Engineering and Computer Science, Technische Universität Berlin, 10587 Berlin, Germany (e-mail: h.joudeh@tu-berlin.de; caire@tu-berlin.de).
Bruno Clerckx is with the
Department of Electrical and Electronic Engineering, Imperial College London, London SW7 2AZ, U.K.
(e-mail: b.clerckx@imperial.ac.uk).
Xinping Yi is with Department of Electrical Engineering and Electronics,
University of Liverpool, Liverpool L69 3BX, U.K. (e-mail: xinping.yi@liverpool.ac.uk).
\\
This work was partially supported by the U.K. Engineering and Physical Sciences Research Council (EPSRC)
under grants EP/N015312/1 and EP/R511547/1, and the European Research Council (ERC) under the ERC Advanced Grant N. 789190, CARENET. Parts of this paper were presented at the 2019 IEEE ISIT \cite{Joudeh2019}.}}
\author{Hamdi~Joudeh, Xinping~Yi, Bruno~Clerckx and Giuseppe~Caire}
\date{}
\maketitle
\begin{abstract}
We consider the information-theoretic optimality of treating inter-cell interference as noise (multi-cell TIN) in downlink cellular networks.
We focus on scenarios modeled by the Gaussian interfering broadcast channel (IBC), comprising $K$ mutually interfering 
Gaussian broadcast channels (BCs), each formed by a base station communicating independent messages to an arbitrary number of users. 
We establish a new power allocation duality between the IBC and its dual interfering multiple access channel (IMAC),
which entails that the corresponding generalized degrees-of-freedom regions achieved through multi-cell TIN and power control 
(TINA regions) for both networks are identical.
As by-products of this duality, we obtain an explicit characterization of the IBC TINA region from a previously established 
characterization of the IMAC TINA region; and identify a multi-cell convex-TIN regime in which the IBC TINA region is a polyhedron (hence convex) without the need for time-sharing.
We then identify a smaller multi-cell TIN regime in which the IBC TINA region is optimal and multi-cell TIN achieves the entire 
capacity region of the IBC, up to a constant gap.
This is accomplished by deriving a new genie-aided outer bound for the IBC, that reveals a novel BC-type order that 
holds amongst users in each constituent BC (or cell) under inter-cell interference, which in turn is not implied by previously known BC-type orders (i.e. degraded, less noisy and more capable orders).
The multi-cell TIN regime that we identify  for the IBC  coincides with a corresponding multi-cell TIN regime previously identified for the IMAC, 
hence establishing a comprehensive uplink-downlink duality of multi-cell TIN in the GDoF (and approximate capacity) sense.
\end{abstract}
\newpage
\section{Introduction}
\label{sec:introduction}
Cellular networks have become an indispensable part of infrastructure for modern day societies. 
Against extensive research and progress, however, information-theoretic capacity limits of such networks remain largely elusive.
At the heart of any cellular network lies a Gaussian interference channel (IC), where each transmitter is paired with a single receiver. 
The capacity region of the simplest instance of this canonical network, i.e. the 2-user IC, remains unknown in general, constituting a long-standing open problem in network information theory \cite{ElGamal2011}.
Nevertheless, despite this lack of exact capacity results, significant progress has been made over the recent years in the understanding of the fundamental limits of interference networks, and specifically of cellular networks. Such progress was made possible by taking few steps away from exact capacity results, and instead pursuing approximate characterizations, mainly through studying degrees-of-freedom (DoF), generalized degrees-of-freedom (GDoF), and related measures and models \cite{Cadambe2008,Etkin2008,Bresler2010,Jafar2010,Jafar2011,Suh2011}.
For instance, going back to the 2-user IC, Etkin, Tse and Wang \cite{Etkin2008}  introduced the notion of GDoF, and characterized the capacity region of this network to within a constant gap of $1$ bit per channel use.

Most relevant to this paper are GDoF studies that focus on regimes of channel strength parameters where interfering links, i.e. links between unpaired transmitters and receivers, are sufficiently weak, such that simple schemes based on power control and treating interference as Gaussian noise (in short, TIN) are optimal \cite{Geng2015,Geng2015a,Geng2016,Sun2016,Yi2016,Gherekhloo2016,Gherekhloo2017,
Geng2018,Yi2019,Chan2019,Joudeh2019a}.
A  breakthrough in this direction is due to Geng et al.  \cite{Geng2015}, who identified a wide regime in which TIN achieves the entire GDoF region of the $K$-user IC,\footnote{This result leads to a characterization of the entire capacity region to within a constant gap in the TIN regime.} known as the TIN regime.
The practical significance of this result is considerable --- TIN is both simple to implement and robust against inaccuracies in channel state 
information at the transmitters (CSIT); moreover, inspired by this TIN result,
a number of high-performing, practical power control and link scheduling algorithms were proposed for device-to-device (D2D) networks \cite{Naderializadeh2014,Geng2016,Yi2016}.
The TIN optimality result of \cite{Geng2015}  is also very interesting from a theoretical perspective --- while the GDoF region of the $K$-user IC remains unknown in its full generality,  this TIN result offers a way forward, serving  as an intermediate step towards a comprehensive solution for general parameter regimes. 

Building upon the result in \cite{Geng2015}, a number of extensions and generalizations followed.
In \cite{Yi2016}, a broader regime for the $K$-user IC, called the convex-TIN (CTIN) regime, is identified, where the TIN achievable (TINA) GDoF region is shown to be convex without the need for time-sharing.\footnote{We recall that the GDoF region achieved through TIN, i.e. the TINA region, which implicitly incorporates power control, is non-convex in general when time-sharing is not used \cite{Geng2015}.}
This CTIN regime is especially relevant when CSIT is limited to finite precision,
a case where TIN turns out to be GDoF optimal for the  $K$-user IC in the CTIN regime, as recently shown in  \cite{Chan2019}.
Beyond the $K$-user IC, TIN GDoF results were derived for other settings, including: channels with general message sets (i.e. X channels) \cite{Geng2015a,Gherekhloo2017}, parallel channels \cite{Sun2016}, multi-state (compound) channels \cite{Geng2016}, multi-state channels with opportunistic decoding capabilities \cite{Yi2019},  and uplink cellular settings modelled by the  Gaussian interfering multiple access channel (IMAC) \cite{Gherekhloo2017,Joudeh2019a}.

In this paper, we are primarily interested in the multi-cell TIN framework of \cite{Joudeh2019a}. In the context of uplink cellular settings modelled by the IMAC, 
multi-cell TIN  is defined in \cite{Joudeh2019a} as the employment of a power-controlled, single-cell-type strategy in each cell, while treating inter-cell interference as Gaussian noise.\footnote{Taking this definition of multi-cell TIN to the extreme of having a single-user in each cell, we reduce to TIN in the  $K$-user IC, e.g. \cite{Geng2015}, with power control and point-to-point Gaussian codes employed in each cell.}
In light of this multi-cell TIN framework, the contributions of \cite{Joudeh2019a} are three-fold: 1) the IMAC TINA region is explicitly characterized as a finite union of polyhedra, each described in terms of channel strength parameters,  2) a multi-cell CTIN regime is identified for the IMAC, in which the TINA region is a polyhedron (hence convex) without the need for time-sharing, and 3) a smaller 
multi-cell TIN regime is identified for the IMAC, in which the TINA region is optimal, and hence multi-cell TIN achieves the entire GDoF region of the IMAC.
Given these multi-cell TIN results for uplink cellular settings, a natural question arises as to whether such results extend to counterpart downlink settings.
We answer this question with the affirmative in this paper.

We focus on downlink cellular settings modelled by the Gaussian interfering broadcast channel (IBC) \cite{Suh2011}, comprising  $K$ mutually interfering Gaussian broadcast channels (BCs), each with an arbitrary number of users.
For the IBC, multi-cell TIN as defined in \cite{Joudeh2019a} translates to the employment of power-controlled superposition coding and successive decoding in each cell, while treating inter-cell interference as Gaussian noise.
Power control naturally complements superposition coding in single-cell settings, i.e. the Gaussian BC, to achieve different  trade-offs among users. The employment of power control in multi-cell TIN settings, however, serves an additional purpose of managing inter-cell interference, hence achieving various trade-offs among users across cells. 
Similar to the definition of the IMAC TINA region in \cite{Joudeh2019a}, the set of GDoF tuples achieved through all feasible power allocations and successive decoding orders constitute  the IBC TINA region. 

In light of the above-described extension of the multi-cell TIN framework in \cite{Joudeh2019a} to the IBC, the contributions of this paper are three-fold, constituting downlink counterparts of the uplink results in \cite{Joudeh2019a}.
These counterparts, nevertheless, require new proof techniques compared to the ones in \cite{Joudeh2019a}, and give rise to fresh insights into the multi-cell TIN problem.
These are summarized as follows:
\begin{enumerate}
\item \emph{IBC TINA region characterization via uplink-downlink duality:} We show that the IBC TINA region is identical to the TINA region of its dual IMAC, obtained by reversing the roles of transmitters and receivers. 
We establish this GDoF uplink-downlink duality of multi-cell TIN by deriving an explicit relationship between power control variables in the IBC and their counterparts in the dual IMAC (see Lemma \ref{lemma:power_allocation_duality}).
This duality enables us to leverage the IMAC TINA region characterization in \cite{Joudeh2019a} to obtain an explicit characterization of the IBC TINA region as a finite union of polyhedra, each described in terms of channel strength parameters.
While uplink-downlink duality results are of interest in their own right, a key advantage of this duality approach is avoiding the potential graph approach of \cite{Geng2015}, which when used directly in multi-cell settings, requires a cumbersome and lengthy procedure of eliminating power control variables and redundant GDoF inequalities (see \cite{Joudeh2019a}). 
\item  \emph{Multi-cell CTIN Regime:} 
The GDoF uplink-downlink duality of multi-cell TIN implies that, similar to the IMAC TINA region, the IBC TINA region is a polyhedron (hence convex) in the multi-cell CTIN regime identified in \cite{Joudeh2019a}. 
Given this observation, we provide an interpretation of the multi-cell CTIN conditions of \cite{Joudeh2019a} in the context of the IBC, and further show that such conditions are insufficient for multi-cell TIN optimality in the IBC (see Section \ref{subsubsec:CTIN_insights}). In particular, we demonstrate that in a sub-regime of the CTIN regime, a scheme based on interference alignment (IA) achieves strict GDoF gains over TIN.
\item  \emph{Multi-cell TIN Regime:}  We establish the optimality of the IBC TINA region in the multi-cell TIN regime identified in \cite{Joudeh2019a}, which is contained in the multi-cell CTIN regime.
This is accomplished by deriving a new genie-aided outer bound for the IBC, which is tight in the GDoF sense in the multi-cell TIN regime.\footnote{This outer bound is also within a constant gap from the entire capacity region in the multi-cell TIN regime.}
In deriving this outer bound, we reveal a new order amongst users in each cell (or BC) of the IBC, which in the multi-cell TIN regime, remains unchanged regardless of the presence or absence of inter-cell interference. 
This new order, which we call the redundancy order, is less restrictive than known orders for the BC, as the degraded order and the less noisy order, which are not necessarily preserved under inter-cell interference in the multi-cell TIN regime (see Section \ref{subsubsec:TIN_insights}). 
Interestingly, in the sub-regime of the CTIN regime which does not overlap with the TIN regime, the redundancy order of users is not necessarily preserved in every cell under 
inter-cell interference.  
This collapse of order in some (or all) cells opens the door for GDoF gains through IA, as alluded to in the above point.
\end{enumerate}
The remainder of this paper is organized as follows:
In Section \ref{sec:system model}, we describe the system model for the IBC and its dual IMAC.
Multi-cell TIN, related definitions and a summary of prior results for the IMAC are presented in Section \ref{sec:TIN}.
The main results of this paper and key insights are given in Section \ref{sec:main_results}. The outer bound for the IBC, used to establish the TIN optimality result in Section \ref{sec:main_results}, is presented alongside its proof in Section   \ref{sec:proof_of_outerbound}.
Section \ref{sec:conclusion} concludes this paper.
In addition to the main sections, some technical details and proofs are relegated to the appendices. 
\subsection{Notation}
For positive integers $z_{1}$ and $z_{2}$, with $z_{1} \leq z_{2}$, the sets $\{1,2,\ldots,z_{1}\}$ and $\{z_{1},z_{1}+1,\ldots,z_{2}\}$ are denoted by
$\langle z_{1} \rangle$ and $\langle z_{1}:z_{2}\rangle$, respectively.
For a real number $a$, $(a)^{+} = \max\{0,a\}$.
The cardinality of set
$\mathcal{A}$ is denoted by $|\mathcal{A}|$.
For sets $\mathcal{A}$ and $\mathcal{B}$, $\mathcal{A} \setminus \mathcal{B}$ is the set of elements in $\mathcal{A}$ and not in $\mathcal{B}$.
The indicator function with condition $A$ is denoted by $\mathbbm{1}(A)$,
which is $1$ when $A$ holds and  $0$ otherwise. 
\section{System Model}
\label{sec:system model}
Consider a $K$-cell cellular network in which each cell $k$ comprises a base station 
BS-$k$ and $L_{k}$ user equipments each denoted by UE-$(l_{k},k)$, where $l_{k} \in \langle L_{k} \rangle$ and $k \in \langle K \rangle$.
The set of tuples corresponding to all UEs in the network is denoted by
$\mathcal{K} \triangleq \left\{(l_{k},k) : l_{k} \in \langle L_{k} \rangle, k \in \langle K \rangle  \right\}$.
\subsection{IBC and Dual IMAC}
When operating in the downlink mode, we assume that the above cellular network is modeled by a Gaussian IBC with private (or unicast) messages only, 
e.g. Fig. \ref{fig:IBC_IMAC}(left).
The input-output relationship at the $t$-th use of this channel, where $t \in \mathbb{N}$, is described as
\begin{equation}
\label{eq:IBC_system model}
Y_{k}^{[l_{k}]}(t)
 = \sum_{i = 1}^{K} \tilde{h}_{ki}^{[l_{k}]} \tilde{X}_{i}(t)
+ Z_{k}^{[l_{k}]}(t).
\end{equation}
In the above, $Y_{k}^{[l_{k}]}(t)$ is the signal received by UE-$(l_{k},k)$, $\tilde{h}_{ki}^{[l_{k}]}$ is the channel coefficient from BS-$i$ to UE-$(l_{k},k)$, $\tilde{X}_{i}(t)$ is the transmitted symbol of BS-$i$ and $Z_{k}^{[l_{k}]}(t) \sim \mathcal{N}_{\mathbb{C}}(0,1)$
is the additive white Gaussian noise (AWGN) at UE-$(l_{k},k)$.
All symbols are complex and the signal transmitted from each BS-$i$ is subject to the average power constraint
given by
\begin{equation}
\label{eq:power_constraint_IBC}
\frac{1}{n}\sum_{t=1}^{n}\E \big[|\tilde{X}_{i}(t)|^{2}\big] \leq P_{i},
\end{equation}
where $n$ is the duration of the communication in channel uses.

The dual IMAC for the above IBC is obtained by reversing the roles of transmitters and receivers in the IBC, e.g. Fig. \ref{fig:IBC_IMAC}(right).
The input-output relationship for the dual IMAC is given by
\begin{equation}
\label{eq:IMAC_system model}
\bar{Y}_{i}(t)  = \sum_{k = 1}^{K} \sum_{l_{k} = 1}^{L_{k}}  \tilde{h}_{ki}^{[l_{k}]} \tilde{\bar{X}}_{k}^{[l_{k}]}(t)
+ \bar{Z}_{i}(t),
\end{equation}
where $\bar{Y}_{i}(t)$ and $\bar{Z}_{i}(t) \sim \mathcal{N}_{\mathbb{C}}(0,1)$ are the received signal and the AWGN at BS-$i$,
respectively, and $\tilde{\bar{X}}_{k}^{[l_{k}]}(t)$ is the transmitted symbol of UE-$(l_{k},k)$.
Signals transmitted by the UEs of cell $k$ are subject to the average sum-power constraint
given by
\begin{equation}
\label{eq:power_constraint_IMAC}
\frac{1}{n}\sum_{t=1}^{n} \E \big[|\tilde{\bar{X}}_{k}^{[l_{k}]}(t)|^{2}\big] \leq \frac{1}{L_{k}} \cdot P_{k}.
\end{equation}
\begin{remark}
In BC-MAC dualities, a rather artificial sum-power constraint across non-cooperating UEs
in the dual MAC is commonly adopted so that that the BC and MAC capacity regions coincide, see for example \cite{Jindal2004}.
Following the same approach, the $L_{k}$ power constraints associated with the UEs in cell $k$ in \eqref{eq:power_constraint_IMAC}
should be replaced with a per-cell sum-power constraint given by 
\begin{equation}
\frac{1}{n}\sum_{t=1}^{n} \sum_{l_{k}=1}^{L_{k}} \E \big[|\tilde{\bar{X}}_{k}^{[l_{k}]}(t)|^{2}\big] \leq P_{k}.
\end{equation}
Nevertheless, for GDoF purposes, it suffices to consider
per-user power constraints as in \eqref{eq:power_constraint_IMAC}.
\hfill $\lozenge$
\end{remark}
\begin{figure}
\centering
\includegraphics[width = 0.7\textwidth,trim={0cm 0cm 0cm 0cm},clip]{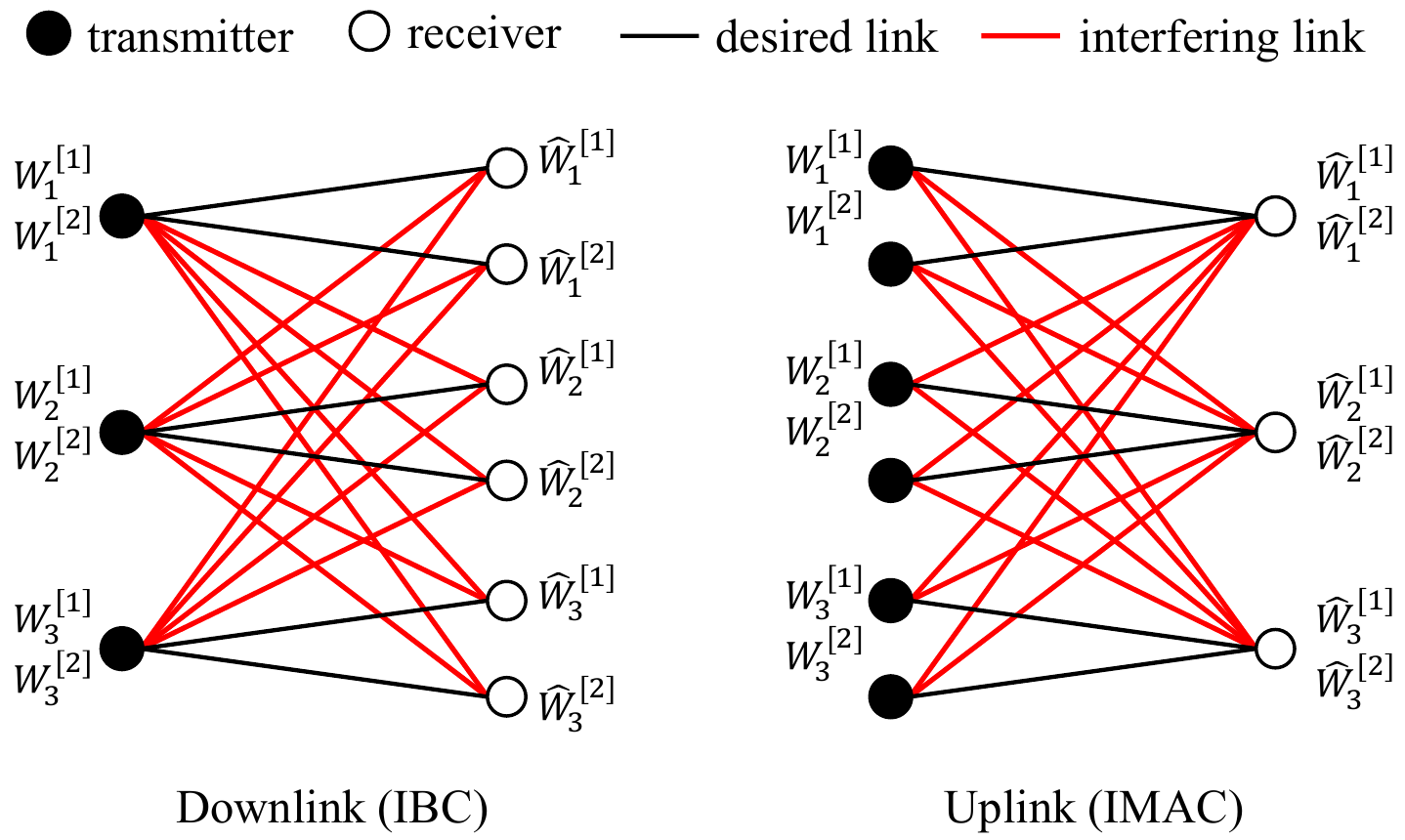}
\caption{\small
Cellular network examples: (left) 3-cell interfering broadcast channel (donwlink) with 2 users in each cell, and (right) the corresponding 
dual interfering multiple access channel (uplink).}
\label{fig:IBC_IMAC}
\end{figure}
\subsection{GDoF Framework}
\label{subsec:system_model_GDoF}
Following \cite{Geng2015}, the above channel models are translated into GDoF-friendly normalized models,
to facilitate GDoF and constant gap capacity studies.
To this end, we define the channel strength level of the link between BS-$i$ and UE-$(l_{k},k)$ as 
\begin{equation}
\alpha_{ki}^{[l_{k}]} \triangleq \frac{ \max \left\{ 0, \log \left( | \tilde{h}_{ki}^{[l_{k}]} |^{2} P_{i} \right) \right\} }{\log P}
, \ \forall (l_{k},k) \in \mathcal{K}, \
i\in \langle K \rangle,
\end{equation}
where $P >0$ is a nominal power value. 
The IBC input-output relationship in \eqref{eq:IBC_system model}  translates into
\begin{equation}
\label{eq:IBC_system model GDoF}
Y_{k}^{[l_{k}]}(t)  =  \sum_{i = 1}^{K} \sqrt{ P^{\alpha_{ki}^{[l_{k}]}} } e^{j \theta_{ki}^{[l_{k}]}} X_{i}(t)
+ Z_{k}^{[l_{k}]}(t),
\end{equation}
where $X_{i}(t) \triangleq {\tilde{X}_{i}(t)}/{\sqrt{P_{i}}}$ is the normalized transmit symbol of BS-$i$ with power constraint
\begin{equation}
\label{eq:power_constraint_IBC_GDoF}
\frac{1}{n}\sum_{t=1}^{n}\E \Big[\big|X_{i}(t)\big|^{2}\Big] \leq 1.
\end{equation}
In  \eqref{eq:IBC_system model GDoF},  $\sqrt{ P^{\alpha_{ki}^{[l_{k}]}} }$ and $\theta_{ki}^{[l_{k}]}$ are the magnitude and phase of the link between BS-$i$ and UE-$(l_{k},k)$, respectively, from which we define the corresponding coefficient as  $h_{ki}^{[l_{k}]} \triangleq \sqrt{ P^{\alpha_{ki}^{[l_{k}]}} } e^{j \theta_{ki}^{[l_{k}]}} $.
As shown in \cite{Geng2015}, avoiding negative channel strength levels has no impact on GDoF or constant gap results.
Therefore, we focus on the equivalent channel model in \eqref{eq:IBC_system model GDoF}
henceforth.

Moving on to the dual IMAC, the model in \eqref{eq:IMAC_system model} translates into
\begin{equation}
\label{eq:IMAC_system model_GDoF}
\bar{Y}_{i}(t)  = \sum_{k = 1}^{K} \sum_{l_{k} = 1}^{L_{k}}
\sqrt{ P^{\alpha_{ki}^{[l_{k}]}} } e^{j \theta_{ki}^{[l_{k}]}} \bar{X}_{k}^{[l_{k}]}(t)
+ \bar{Z}_{i}(t),
\end{equation}
where $\bar{X}_{k}^{[l_{k}]}(t) \triangleq \tilde{\bar{X}}_{k}^{[l_{k}]}(t)/\sqrt{P_{k}}$
is the normalized transmit symbol of UE-$(l_{k},k)$ with power constraint
\begin{equation}
\label{eq:power_constraint_IMAC_GDoF}
\frac{1}{n}\sum_{t=1}^{n}\E \Big[\big|\bar{X}_{k}^{[l_{k}]}(t)\big|^{2}\Big] \leq \frac{1}{L_{k}}.
\end{equation}
\begin{remark}
\label{remark:strength_order}
Without loss of generality, we assume that users in each cell are in an ascending order with respect to their direct link strength levels (or SNRs).
That is:
\begin{equation}
\label{eq:strength_order}
\alpha_{kk}^{[1]}  \leq  \alpha_{kk}^{[2]} \leq \cdots \leq \alpha_{kk}^{[L_{k}]}, \ \forall k \in \langle K \rangle.
\end{equation}
Moreover, for any pair of users UE-$(l_{k},k)$ and UE-$(l_{k}',k)$ in the same cell $k$ with $l_{k} > l_{k}'$,
and hence $\alpha_{kk}^{[l_{k}]}  \geq  \alpha_{kk}^{[l_{k}']}$,
we refer to the former as the \emph{stronger} BC user and to the latter as the \emph{weaker} BC user,
where the notions of stronger and weaker are defined with respect to SNRs.
\hfill $\lozenge$
\end{remark}
\subsection{Messages, Rates, Capacity and GDoF}
Each BS-$k$ in the IBC, where $k \in \langle K \rangle$,
has the independent messages $W_{k}^{[1]},\ldots,W_{k}^{[L_{k}]}$
intended to UE-$(1,k)$,$\ldots$,UE-$(L_{k},k)$, respectively.
Codes, error probabilities, achievable rates and the capacity region are all defined in the standard Shannon theoretic sense.
For fixed $P$, an achievable rate tuple is denoted by
$\mathbf{R}(P) = \big(R_{k}^{[l_{k}]}(P) : (l_{k},k) \in \mathcal{K} \big)$ and the capacity region is denoted by $\mathcal{C}^{\mathrm{IBC}}(P)$.
On the other hand, an achievable GDoF tuple is denoted by $\mathbf{d} = \big(d_{k}^{[l_{k}]} : (l_{k},k) \in \mathcal{K} \big)$
and the GDoF region is denoted by $\mathcal{D}^{\mathrm{IBC}}$, where the latter is defined as
\begin{equation}
\mathcal{D}^{\mathrm{IBC}} \triangleq \left\{ \mathbf{d} : \  d_{k}^{[l_{k}]} = \lim_{P \rightarrow \infty} \frac{R_{k}^{[l_{k}]}(P)}{\log(P)}, \
\forall (l_{k},k) \in \mathcal{K}, \ \mathbf{R}(P) \in \mathcal{C}^{\mathrm{IBC}}(P)  \right\}.
\end{equation}

The above definitions translate to the dual IMAC by reversing the roles of transmitters and receivers.
An achievable rate tuple is denoted by
$\bar{\mathbf{R}}(P) = \big(\bar{R}_{k}^{[l_{k}]}(P) : (l_{k},k) \in \mathcal{K} \big)$, the capacity region is denoted by $\mathcal{C}^{\mathrm{IMAC}}(P)$, an achievable GDoF tuple is denoted by $\bar{\mathbf{d}} =
\big(\bar{d}_{k}^{[l_{k}]} : (l_{k},k) \in \mathcal{K} \big)$,
and the GDoF region is denoted by $\mathcal{D}^{\mathrm{IMAC}}$, which in turn is defined as
\begin{equation}
\mathcal{D}^{\mathrm{IMAC}} \triangleq \left\{ \bar{\mathbf{d}} : \  \bar{d}_{k}^{[l_{k}]} = \lim_{P \rightarrow \infty} \frac{\bar{R}_{k}^{[l_{k}]}(P)}{\log(P)}, \
\forall (l_{k},k) \in \mathcal{K}, \ \bar{\mathbf{R}}(P) \in \mathcal{C}^{\mathrm{IMAC}}(P)  \right\}.
\end{equation}
\section{Treating (Inter-Cell) Interference as Noise}
\label{sec:TIN}
We consider TIN in the cellular sense as in \cite{Joudeh2019a}, i.e. multi-cell TIN,
where  a single-cell, capacity-achieving-type strategy is employed by each cell, while treating all inter-cell interference as noise.
This is done in tandem with power control to manage inter-cell interference.
When there is no confusion, we may drop the multi-cell attribute, and briefly refer to 
multi-cell TIN as TIN. 
\subsection{Multi-Cell TIN in the IBC}
\label{subsec:TIN_IBC}
Each BS-$k$, where $k \in \langle K \rangle$, in the IBC employs superposition coding, whilst UEs in cell $k$
employ successive decoding according to a decoding order $\pi_{k}(\cdot)$, while treating all inter-cell interference as noise.
In particular, the transmitted signal of BS-$k$  is composed as
\begin{equation}
X_{k}(t) =  \sum_{l_{k} = 1}^{L_{k}} X_{k}^{[l_{k}]}(t),
\end{equation}
where each message $W_{k}^{[l_{k}]}$ is encoded into a codeword
$X_{k}^{[l_{k}]n} \triangleq  \big( X_{k}^{[l_{k}]}(1),\ldots,X_{k}^{[l_{k}]}(n) \big)$,
drawn from a Gaussian codebook
with average power $\frac{1}{n}\sum_{t=1}^{n}\E \big[ | X_{k}^{[l_{k}]}(t) |^{2} \big] = q_{k}^{[l_{k}]}$.
The powers  $q_{k}^{[1]},\ldots,q_{k}^{[L_{k}]}$ allocated to different codewords satisfy the constraint in \eqref{eq:power_constraint_IBC_GDoF},
which translates into
\begin{equation}
\label{eq:power_constraint_SC_SIC}
\sum_{l_{k}=1}^{L_{k}}q_{k}^{[l_{k}]} \leq 1.
\end{equation}
At the other end of the channel, each UE-$\big(\pi_{k}(l_{k}), k\big)$ starts by successively decoding and cancelling
$X_{k}^{[\pi_{k}(1)]n}, X_{k}^{[\pi_{k}(2)]n},\ldots,X_{k}^{[\pi_{k}(l_{k}-1)]n}$, in this order,
before decoding its own signal $X_{k}^{[\pi_{k}(l_{k})]n}$, while treating all other signals (i.e. both intra-cell and inter-cell interference) as noise.

Using the above scheme, the signal $X_{k}^{[\pi_{k}(l_{k})]n}$, intended to UE-$\big(\pi_{k}(l_{k}), k\big)$,
is decoded by all UEs indexed by $\big(\pi_{k}(m_{k}), k\big)$, where $m_{k} \geq l_{k}$.
Therefore, the corresponding effective signal-to-interference-plus-noise ratio (SINR), denoted by $\mathrm{SINR}_{k}^{[\pi_{k}(l_{k})]}$,
is given by
\begin{equation}
\label{eq:SINR_IBC}
\mathrm{SINR}_{k}^{[\pi_{k}(l_{k})]} = \min_{m_{k}:m_{k} \geq l_{k}} \left\{
\frac{ P^{\alpha_{kk}^{[\pi_{k}(m_{k})]}} q_{k}^{[\pi_{k}(l_{k})]}  }
{1 + \sum_{l_{k}'' = l_{k} + 1}^{L_{k}} P^{\alpha_{kk}^{[\pi_{k}(m_{k})]}} q_{k}^{[\pi_{k}(l_{k}'')]} + \sum_{j:j\neq k} P^{\alpha_{kj}^{[\pi_{k}(m_{k})]}} \sum_{l_{j} = 1}^{L_{j}} q_{j}^{[l_{j}]} } \right\}.
\end{equation}
It follows that message $W_{k}^{[\pi_{k}(l_{k})]}$ is reliably communicated to
UE-$\big(\pi_{k}(l_{k}),k\big)$ at any rate $R_{k}^{[\pi_{k}(l_{k})]} \geq 0$ satisfying
\begin{equation}
\label{eq:IBC_rate per user}
R_{k}^{[\pi_{k}(l_{k})]} \leq   \log
\Big(  1+  \mathrm{SINR}_{k}^{[\pi_{k}(l_{k})]}   \Big),
\end{equation}
where the right-hand-side in \eqref{eq:IBC_rate per user} is the achievable rate when treating inter-cell, and remaining intra-cell, interference as Gaussian noise. 

Next, we translate the above into the GDoF framework.
For GDoF purposes, we may further restrict the power allocation in each cell $k$ such that
$q_{k}^{[l_{k}]} \leq {1}/{L_{k}}$, for all $l_{k} \in \langle L_{k} \rangle$, which clearly does not violate the power constraint
in \eqref{eq:power_constraint_SC_SIC}.
More importantly, this allows us to write
\begin{equation}
\label{eq:qk_IBC_GDoF}
q_{k}^{[l_{k}]} = \frac{1}{L_{k}} \cdot P^{r_{k}^{[l_{k}]}}, \ \text{for some} \ r_{k}^{[l_{k}]} \leq 0,
\end{equation}
where $r_{k}^{[l_{k}]}$ is the corresponding transmit power exponent, or power control variable.
Using the power allocation in \eqref{eq:qk_IBC_GDoF}, it follows that
UE-$\big(\pi_{k}(l_{k}),k\big)$  achieves any GDoF $d_{k}^{[\pi_{k}(l_{k})]} \geq 0$ that satisfies
\begin{multline}
\label{eq:IBC_GDoF per user}
d_{k}^{[\pi_{k}(l_{k})]}  \leq
\biggl( \min_{m_{k}:m_{k} \geq l_{k} }    \Bigl\{  \alpha_{kk}^{[\pi_{k}(m_{k})]}   +   r_{k}^{[\pi_{k}(l_{k})]}
\\
 -   \max \bigl\{0, \alpha_{kk}^{[\pi_{k}(m_{k})]}  +  \max_{l_{k}'':l_{k}''>l_{k}} \{ r_{k}^{[\pi_{k}(l_{k}'')]} \} ,
 \max_{(l_{j},j):j \neq k} \{ \alpha_{kj}^{[\pi_{k}(m_{k})]}
 +   r_{j}^{[l_{j}]} \} \bigr\}    \Bigr\}  \biggr)^{+}.
\end{multline}

The tuple of power control variables is denoted by
$\mathbf{r}  =  \big(r_{k}^{[l_{k}]} : (l_{k},k) \in \mathcal{K} \big)$, and is also referred to as a
power allocation for the IBC.
On the other hand, a network decoding order tuple is given by $\bm{\pi} \triangleq \left( \pi_{1},\ldots,  \pi_{K} \right) \in \Pi$,
where $\Pi$ is the set of all possible network decoding orders.
For fixed $(\bm{\pi},\mathbf{r})$, the set of all TIN achievable (TINA) GDoF tuples is denoted by $\mathcal{D}_{\mathrm{TINA}}^{\mathrm{IBC}}(\bm{\pi},\mathbf{r})$, which is given by all
GDoF tuples $\mathbf{d} \in \mathbb{R}_{+}^{|\mathcal{K}|} $ with components satisfying \eqref{eq:IBC_GDoF per user}.
By maintaining a fixed network decoding order $\bm{\pi}$ while considering all possible power allocations $\mathbf{r}$,
we obtain the TINA$(\bm{\pi})$ region given by
\begin{equation}
\label{eq:IBC_TINA_pi_GDoF_region}
\mathcal{D}_{\mathrm{TINA}}^{\mathrm{IBC}}(\bm{\pi}) \triangleq  \bigcup_{\mathbf{r} \leq \mathbf{0}}
\mathcal{D}_{\mathrm{TINA}}^{\mathrm{IBC}}(\bm{\pi},\mathbf{r}).
\end{equation}
By further considering all possible network decoding orders, we obtain the TINA region defined as
\begin{equation}
\label{eq:IBC_TINA_GDoF_region}
\mathcal{D}_{\mathrm{TINA}}^{\mathrm{IBC}} \triangleq \bigcup_{\bm{\pi} \in \Pi} \bigcup_{\mathbf{r} \leq \mathbf{0}}
\mathcal{D}_{\mathrm{TINA}}^{\mathrm{IBC}}(\bm{\pi},\mathbf{r}).
\end{equation}
It is worthwhile highlighting that the use of time-sharing is disallowed in the above multi-cell TIN scheme,
hence keeping to a widely adopted tradition in previous GDoF works 
\cite{Etkin2008,Bresler2010,Geng2015,Geng2015a,Geng2016,Sun2016,Yi2016,Geng2018,Yi2019,Chan2019,Joudeh2019a}.
As a result, the TINA region $\mathcal{D}_{\mathrm{TINA}}^{\mathrm{IBC}}$ is non-convex in general, 
and each GDoF tuple $\mathbf{d} \in \mathcal{D}_{\mathrm{TINA}}^{\mathrm{IBC}}$ is achieved through a strategy
identified by a fixed  $(\bm{\pi},\mathbf{r})$.
Remarkably, while prohibiting time-sharing is mainly motivated by tractability, 
this restriction turns out to have no influence on the results in the regimes of interest, as shown further on in Section \ref{sec:main_results}.
\begin{remark}
The IBC multi-cell TIN setting considered here is related to  the compound IC TIN setting in \cite{Geng2016} and the multi-state IC opportunistic TIN setting in \cite{Yi2019}.
In particular, all three settings share the same input-output signal model in \eqref{eq:IBC_system model GDoF}, and may be interpreted as scenarios of downlink multi-cell transmission.\footnote{Note that in both \cite{Geng2016} and \cite{Yi2019}, each cell is interpreted as a single user with multiple states.}
The difference  is in the message sets.
The compound setting in \cite{Geng2016} captures scenarios where each BS has a single degraded (i.e. multicast) message, intended to all its corresponding users.
The opportunistic setting in \cite{Yi2019} is more general than \cite{Geng2016}; 
in addition to the \emph{basic} multicast message, each BS 
transmits additional, lower priority, \emph{opportunistic} messages,
where each such message is opportunistically decoded by a subset of intended users in a degraded message sets fashion. 
In the IBC considered in this paper, each BS transmits unicast messages only, and we do not consider degraded message sets in the model.
Nevertheless, by construction of the multi-cell TIN scheme in Section \ref{subsec:TIN_IBC}, a degraded structure is enforced on messages, which allows us to retrieve the  TINA regions in \cite{Geng2016} and \cite{Yi2019} from the IBC TINA region in \eqref{eq:IBC_TINA_GDoF_region}.

For instance, the compound TINA region in \cite{Geng2016} is retrieved from the IBC TINA  region in \eqref{eq:IBC_TINA_GDoF_region}  by eliminating all messages, e.g. by setting their GDoF to zero, except for message 
 $W_{k}^{[\pi_{k}(1)]}$  in each cell $k$, which is decoded by all users in such cell.
Similarly, the opportunistic TINA region in \cite{Yi2019} is retrieved by fixing a decoding order $\bm{\pi}$ a priori, and (possibly) eliminating some messages from each cell according to the criteria in \cite{Yi2019}.
Looking through this multi-cell TIN lens, it can be seen that the message sets in both the compound IC and the opportunistic IC settings are less restrictive than that in the multi-cell TIN IBC setting considered here. Hence, it is not surprising that the multi-cell TIN conditions for the IBC, presented in the following section, are more restrictive than the compound TIN conditions in \cite{Geng2016} and the opportunistic TIN conditions in  \cite{Yi2019}.
\hfill $\lozenge$
\end{remark}
\subsection{TIN in the Dual IMAC}
\label{subsec:TIN_IMAC}
For the dual IMAC, we adopt the TIN scheme given in \cite{Joudeh2019a}.
In particular, each UE-$(l_{k},k)$, where $(l_{k},k) \in \mathcal{K}$,
employs an independent Gaussian codebook with average power that satisfies the constraint in
\eqref{eq:power_constraint_IMAC_GDoF}.
Similar to the IBC, we may write
\begin{equation}
\label{eq:qk_IMAC_GDoF}
\frac{1}{n}\sum_{t=1}^{n}\E \big[ | \bar{X}_{k}^{[l_{k}]}(t) |^{2} \big]  = \frac{1}{L_{k}} \cdot P^{\bar{r}_{k}^{[l_{k}]}}, \ \text{for some} \ \bar{r}_{k}^{[l_{k}]} \leq 0,
\end{equation}
where $\bar{r}_{k}^{[l_{k}]} \leq 0$ is the corresponding power control variable.
On the other end, each BS-$k$ successively decodes and cancels its in-cell signals
$\bar{X}_{k}^{[\pi_{k}(L_{k})]},\bar{X}_{k}^{[\pi_{k}(L_{k}-1)]},\ldots,\bar{X}_{k}^{[\pi_{k}(1)]}$, in this (descending) order,
while treating all inter-cell interference as noise.

A power control tuple (or power allocation) for the IMAC is denoted by
$\bar{\mathbf{r}} =  \big(\bar{r}_{k}^{[l_{k}]} : (l_{k},k) \in \mathcal{K} \big)$, while a
network decoding order tuple is given by $\bm{\pi} \in \Pi$.
For a fixed strategy given by $(\bm{\pi},\bar{\mathbf{r}})$, UE-$\big(\pi_{k}(l_{k}),k\big)$ achieves any GDoF $\bar{d}_{k}^{[\pi_{k}(l_{k})]} \geq 0$ that
satisfies
\begin{equation}
\label{eq:IMAC_GDoF per user}
\bar{d}_{k}^{[\pi_{k}(l_{k})]}  \leq
\biggl(  \alpha_{kk}^{[\pi_{k}(l_{k})]} + \bar{r}_{k}^{[\pi_{k}(l_{k})]}
  -  \max \bigl\{0, \max_{l_{k}':l_{k}' < l_{k}} \{ \alpha_{kk}^{[\pi_{k}(l_{k}')]} + \bar{r}_{k}^{[\pi_{k}(l_{k}')]} \}  ,\max_{(l_{j},j):j \neq k}
 \{ \alpha_{jk}^{[l_{j}]} +  \bar{r}_{j}^{[l_{j}]} \}   \bigr\} \biggr)^{+}.
\end{equation}
For a fixed strategy $(\bm{\pi},\bar{\mathbf{r}})$, the set of GDoF tuples $\bar{\mathbf{d}} \in \mathbb{R}_{+}^{|\mathcal{K}|} $ with components satisfying
\eqref{eq:IMAC_GDoF per user} is denoted by $\mathcal{D}_{\mathrm{TINA}}^{\mathrm{IMAC}}(\bm{\pi},\bar{\mathbf{r}})$,
while the TINA$(\bm{\pi})$ region for the dual IMAC is defined as
\begin{equation}
\label{eq:IMAC_TINA_pi_GDoF_region}
\mathcal{D}_{\mathrm{TINA}}^{\mathrm{IMAC}}(\bm{\pi}) \triangleq  \bigcup_{\bar{\mathbf{r}} \leq \mathbf{0}}
\mathcal{D}_{\mathrm{TINA}}^{\mathrm{IMAC}}(\bm{\pi},\bar{\mathbf{r}}).
\end{equation}
Finally, the TINA region for the dual IMAC is defined as
\begin{equation}
\label{eq:IMAC_TINA_GDoF_region}
\mathcal{D}_{\mathrm{TINA}}^{\mathrm{IMAC}} \triangleq \bigcup_{\bm{\pi} \in \Pi} \bigcup_{\bar{\mathbf{r}} \leq \mathbf{0}}
\mathcal{D}_{\mathrm{TINA}}^{\mathrm{IMAC}}(\bm{\pi},\bar{\mathbf{r}}).
\end{equation}
\begin{remark}
For any given cell $k$ and permutation $\pi_{k}$, the uplink decoding order
is the reverse of its counterpart downlink decoding order. This reverse relationship in decoding orders is commonly
exhibited in uplink-downlink dualities, see for example  \cite[Ch. 10.3.4]{Tse2005}.
\hfill $\lozenge$
\end{remark}
\subsection{Definitions and Prior Results}
\label{subsec:definitions}
We conclude this section with some definitions and a summary of the IMAC TIN results in \cite{Joudeh2019a},
 which are instrumental to the formulation of the IBC TIN results presented in the following section.
\begin{definition}
\textbf{(Subnetwork).}
A subnetwork is a subset of UEs and their corresponding subset of serving BSs, e.g. $\mathcal{S} \subseteq \mathcal{K}$ and $\mathcal{M} \subseteq \langle K \rangle$,
where $\mathcal{S} = \cup_{i \in \mathcal{M}} \mathcal{S}_{i}$ is the subset of UEs and $\mathcal{S}_{i} \subseteq \mathcal{K}_{i}$ comprises the UEs participating from cell 
$i \in \mathcal{M}$.
For brevity, we refer to $\mathcal{S}$
as a subnetwork, especially that the corresponding $\mathcal{M}$ is automatically identified by $\mathcal{S}$.
\end{definition}
\begin{definition}
\textbf{(Natural Decoding Order).}
The natural (or identity) decoding order is given by $\bm{\pi} = \bm{\mathrm{id}}$, where $\bm{\mathrm{id}} \triangleq (\mathrm{id}_{1},\ldots, \mathrm{id}_{K})$
and $ \mathrm{id}_{i}(l_{i}) = l_{i}$, for all $(l_{i},i) \in \mathcal{K}$.
From a GDoF region viewpoint, this is the optimal decoding order for both the IBC and IMAC in the absence of inter-cell interference, i.e. whenever 
$\alpha^{[l_{i}]}_{ij} = 0$ for all $(l_{i},i)$ and $j$, where $i \neq j$.
\end{definition}
\begin{definition}
\textbf{(Subnetwork Decoding Order).}
For a subnetwork $\mathcal{S} = \cup_{i \in \mathcal{M}} \mathcal{S}_{i}$,
a subnetwork decoding order tuple is given by $\bm{\pi} \triangleq (\pi_{i}: i \in \mathcal{M})$,
where $\pi_{i} : \langle |\mathcal{S}_{i}| \rangle \rightarrow \mathcal{S}_{i}$ maps the decoding
order $s_{i} \in \langle |\mathcal{S}_{i}| \rangle$ to user $\pi_{i}(s_{i}) \in \mathcal{S}_{i}$ in cell $i$.
The set comprising all subnetwork decoding orders over $\mathcal{S}$ is denoted by $\Pi(\mathcal{S})$.
From this definition, it is evident that $\Pi = \Pi(\mathcal{K})$.
\end{definition}
\begin{definition}
\textbf{(Cyclic Sequences).}
For a subset of cells $\mathcal{M} \triangleq \{k_{1},k_{2},\ldots,k_{|\mathcal{M}|}\} \subseteq \langle K \rangle$,
$\Sigma(\mathcal{M})$ denotes the set of all cyclically ordered sequences formed by 
any number of elements in $\mathcal{M}$ without repetitions.
For example, for $\mathcal{M} = \{ k_{1},k_{2},k_{3} \}$, we have
\begin{equation}
\nonumber
\Sigma(\mathcal{M} ) =
\big\{(k_{1}),(k_{2}),(k_{3}),(k_{1},k_{2}),(k_{1},k_{3}),(k_{2},k_{3}),
(k_{1},k_{2},k_{3}),(k_{1},k_{3},k_{2})  \big\}.
\end{equation}
\end{definition}
Next, we define a polyhedral GDoF regions and two regimes of channel parameters, whose operational significance is given in a following theorem, summarizing the main results of \cite{Joudeh2019a}.
\begin{definition}
\textbf{(Polyhedral-TIN Region).}
\label{def:polyhedral_region}
For a subnetwork $\mathcal{S} = \cup_{i \in \mathcal{M}} \mathcal{S}_{i}$ and subnetwork decoding order $\bm{\pi} \in \Pi(\mathcal{S})$,
the corresponding polyhedral-TIN region, denoted by $\mathcal{P}(\bm{\pi},\mathcal{S})$, is given by
all GDoF tuples $\mathbf{d} \in \mathbb{R}_{+}^{|\mathcal{K}|}$ that satisfy
\begin{align}
d_{j}^{[l_{j}]}   &  = 0,   \ \forall (l_{j},j) \in \mathcal{K} \setminus \mathcal{S} \\
\sum_{s_{i} = 1}^{ l_{i}}
d_{i}^{[\pi_{i}(s_{i})]} & \leq \alpha_{ii}^{[\pi_{i}(l_{i})]}, \
\forall l_{i}\in \langle |\mathcal{S}_{i}| \rangle, \; i \in \mathcal{M} \\
\nonumber
\sum_{j =1 }^{m} \sum_{s_{i_{j}} = 1}^{l_{i_{j}} } d_{i_{j}}^{[\pi_{i_{j}}(s_{i_{j}})]} & \leq
\sum_{j =1 }^{m} \alpha_{i_{j}i_{j}}^{[\pi_{i_{j}}(l_{i_{j}})]} - \alpha_{i_{j}i_{j-1}}^{[\pi_{i_{j}}(l_{i_{j}})]}, \\
\label{eq:polyhedral_TINA_subset_2}
 \forall l_{i_{j}} \in \langle |\mathcal{S}_{i}| \rangle ,
 (i_{1} & ,\ldots  ,i_{m}) \in \Sigma\big(\mathcal{M}\big), m \in \langle 2:|\mathcal{M}| \rangle.
\end{align}
In \eqref{eq:polyhedral_TINA_subset_2}, a modulo-$m$ operation is implicitly used on cell indices such that $i_{0} = i_{m}$.
\end{definition}
Polyhedral-TIN  regions, obtained by varying $\bm{\pi}$ and $\mathcal{S}$, are the main building blocks of the the IMAC TINA region \cite{Joudeh2019a}.
In the following section, we see that such polyhedral-TIN regions play a similar role for the IBC.
Next, we define the multi-cell TIN and multi-cell CTIN regimes, which were identified in the IMAC context in \cite{Joudeh2019a}.
Since it clear from the context that we are considering multi-cell scenarios, i.e. the IBC and IMAC, as opposed  to the regular IC scenarios in \cite{Geng2015,Yi2016,Chan2019},
we often drop the multi-cell attribute when referring the TIN and CTIN regimes. 
\begin{definition}
\label{def:CTIN_regime}
\textbf{(CTIN Regime).}
In this regime, channel strengths satisfy
\begin{align}
\label{eq:CTIN_cond_1}
\alpha_{ii}^{[l_{i}]}  & \geq  \alpha_{ij}^{[l_{i}]} +  \alpha_{ii}^{[l_{i}']}  - \alpha_{ij}^{[l_{i}']}, \
\forall l_{i}',l_{i} \in \langle L_{i} \rangle, \; l_{i}' < l_{i} \\
\label{eq:CTIN_cond_2}
\alpha_{ii}^{[1]} & \geq  \alpha_{ij}^{[1]} + \alpha_{ki}^{[l_{k}]} -
\alpha_{kj}^{[l_{k}]} \mathbbm{1}\big( k \neq j\big), \  \forall l_{k} \in \langle L_k \rangle,
\end{align}
for all cells $i,j,k \in \langle K \rangle$, such that  $i \notin \{j,k\}$.
\end{definition}
\begin{definition}
\label{def:TIN_regime}
\textbf{(TIN Regime).}
In this regime, channel strengths satisfy
\begin{align}
\label{eq:TIN_cond_1}
\alpha_{ii}^{[l_{i}]}  & \geq  \alpha_{ij}^{[l_{i}]} +  \alpha_{ii}^{[l_{i}']} 
\ \text{ or } \ 
\alpha_{ii}^{[l_{i}]}   \geq  2\alpha_{ij}^{[l_{i}]} +  \alpha_{ii}^{[l_{i}']}  - \alpha_{ij}^{[l_{i}']}, \
\forall l_{i}',l_{i} \in \langle L_{i} \rangle, \; l_{i}' < l_{i} \\
\label{eq:TIN_cond_2}
\alpha_{ii}^{[1]} & \geq  \alpha_{ij}^{[1]} + \alpha_{ki}^{[l_{k}]}, \  \forall l_{k} \in \langle L_k \rangle,
\end{align}
for all cells $i,j,k \in \langle K \rangle$, such that  $i \notin \{j,k\}$.
\end{definition}
\begin{remark}
\label{remark:CTIN_TIN_other_users}
The CTIN condition in \eqref{eq:CTIN_cond_2} and TIN condition in  \eqref{eq:TIN_cond_2}, which are reminiscent of their $K$-user IC counterparts in \cite[Th. 4]{Yi2016} and \cite[Th. 1]{Geng2015}, respectively,
are made to hold for the \emph{weakest} BC user in each cell $i$, i.e. UE-$(1,i)$, against all 
users in other cells, i.e. UE-$(l_{k},k)$ for all $l_{k} \in \langle L_{k} \rangle$ and $k \in \langle K \rangle \setminus \{i\}$.
It can be verified that in the CTIN and TIN regimes, respectively,  
\eqref{eq:CTIN_cond_2} and  \eqref{eq:TIN_cond_2} implicitly hold for all remaining users in cell $i$, i.e. UE-$(l_{i},i)$ for all $l_{i} > 1$.

That is, for all cells  $i,j,k \in \langle K \rangle$ where $i \notin \{j,k\}$, 
the CTIN conditions \eqref{eq:CTIN_cond_1} and \eqref{eq:CTIN_cond_2} imply that
\begin{align}
\label{eq:CTIN_cond_other_users}
\alpha_{ii}^{[l_{i}]}  \geq  \alpha_{ij}^{[l_{i}]} + \alpha_{ki}^{[l_{k}]} -
\alpha_{kj}^{[l_{k}]} \mathbbm{1}\big( k \neq j\big), \; \forall l_{i} \in \langle L_i \rangle, l_{k} \in \langle L_k \rangle,
\end{align}
while the TIN conditions \eqref{eq:TIN_cond_1} and \eqref{eq:TIN_cond_2} imply that
\begin{align}
\label{eq:TIN_cond_other_users}
\alpha_{ii}^{[l_{i}]}  \geq  \alpha_{ij}^{[l_{i}]} + \alpha_{ki}^{[l_{k}]}, \; \forall l_{i} \in \langle L_i \rangle, l_{k} \in \langle L_k \rangle.
\end{align}
We recall that the redundant inequalities in \eqref{eq:CTIN_cond_other_users} and \eqref{eq:TIN_cond_other_users} 
are expressed explicitly in the descriptions of the CTIN and TIN regimes in \cite[Th. 3]{Joudeh2019a} and \cite[Th. 4]{Joudeh2019a}, respectively.
\hfill $\lozenge$
\end{remark}
By inspection, it can be verified that the TIN regime is included in the CTIN regime. 
The significance of the CTIN and TIN regimes, in addition to the rest of the above definitions,
is epitomized through the following theorem, which summarizes the main results in \cite{Joudeh2019a}.
\begin{theorem}
\label{theorem:TIN_IMAC}
\textnormal{\cite[Th. 2, Th.3, Th. 4]{Joudeh2019a}}.
For the IMAC described in Section \ref{sec:system model}, we have the following:
\begin{enumerate}
\item $\mathcal{D}^{\mathrm{IMAC}}_{{\mathrm{TINA}}}$
is characterized by a union of polyhedral-TIN regions as
\begin{equation}
\label{eq:GDoF_region_IMAC}
\mathcal{D}^{\mathrm{IMAC}}_{{\mathrm{TINA}}} = \bigcup_{\mathcal{S} \subseteq \mathcal{K}} \bigcup_{\bm{\pi} \in \Pi(\mathcal{S})}
\mathcal{P}(\bm{\pi}, \mathcal{S}).
\end{equation}
\item In the CTIN regime (Definition \ref{def:CTIN_regime}), $\mathcal{D}^{\mathrm{IMAC}}_{{\mathrm{TINA}}}$
is a convex polyhedron given by
\begin{equation}
\label{eq:GDoF_region_IMAC_CTIN}
\mathcal{D}^{\mathrm{IMAC}}_{{\mathrm{TINA}}} = \mathcal{P}(\bm{\mathrm{id}}, \mathcal{K}).
\end{equation}
\item In the TIN regime (Definition \ref{def:TIN_regime}),  $\mathcal{D}^{\mathrm{IMAC}}_{{\mathrm{TINA}}}$ is
optimal such that
\begin{equation}
\label{eq:GDoF_region_IMAC_TIN}
\mathcal{D}^{\mathrm{IMAC}} = \mathcal{D}^{\mathrm{IMAC}}_{{\mathrm{TINA}}} = \mathcal{P}(\bm{\mathrm{id}}, \mathcal{K}).
\end{equation}
\end{enumerate}
\end{theorem}
From Theorem \ref{theorem:TIN_IMAC}, $\mathcal{P}(\bm{\mathrm{id}}, \mathcal{K})$ 
includes all other polyhedral-TIN regions in the CTIN and TIN regimes.
We conclude this section with an explicit characterization of $\mathcal{P}(\bm{\mathrm{id}},\mathcal{K})$, obtained by 
specializing Definition \ref{def:polyhedral_region}. This is given by all GDoF tuples
$\mathbf{d} \in \mathbb{R}_{+}^{|\mathcal{K}|}$ that satisfy
\begin{align}
\label{eq:polyhedral_TIN_region_1}
\sum_{s_{i} = 1}^{ l_{i}}
d_{i}^{[s_{i}]} & \leq \alpha_{ii}^{[l_{i}]}, \
\forall l_{i}\in \langle L_{i}\rangle, \; i \in \mathcal{K} \\
\nonumber
\sum_{j =1 }^{m} \sum_{s_{i_{j}} = 1}^{l_{i_{j}} } d_{i_{j}}^{[s_{i_{j}}]} & \leq
\sum_{j =1 }^{m} \alpha_{i_{j}i_{j}}^{[l_{i_{j}}]} - \alpha_{i_{j}i_{j-1}}^{[l_{i_{j}}]},   \\
\label{eq:polyhedral_TIN_region_2}
\forall l_{i_{j}} \in \langle L_{i} \rangle,  \;  (i_{1}  , \ldots & ,   i_{m} )   \in \Sigma\big(\mathcal{K}\big),  \; m \in \langle 2:K \rangle.
\end{align}
\section{Main Results}
\label{sec:main_results}
The primary result of this paper is an IBC counterpart of the IMAC 
result in Theorem \ref{theorem:TIN_IMAC}.
As an auxiliary result, we first establish a new GDoF uplink-downlink duality of multi-cell TIN.
This enables us to obtain a characterization of the IBC TINA region $\mathcal{D}^{\mathrm{IBC}}_{{\mathrm{TINA}}} $ 
in terms of its dual IMAC TINA region $\mathcal{D}^{\mathrm{IMAC}}_{{\mathrm{TINA}}}$, which in turn is characterized in Theorem \ref{theorem:TIN_IMAC}.
\subsection{Uplink-Downlink Duality of Multi-Cell TIN}
To facilitate the exposition of the duality result and its proof, we start by expressing the GDoF inequality of the IBC in \eqref{eq:IBC_GDoF per user}
more briefly as
\begin{equation}
\label{eq:IBC_GDoF_gamma}
d_{k}^{[\pi_{k}(l_{k})]}  \leq
\left(  \alpha_{kk}^{[\pi_{k}(l_{k})]} + r_{k}^{[\pi_{k}(l_{k})]}   - \gamma_{k}^{[\pi_{k}(l_{k})]} \right)^{+},
\end{equation}
where the quantity $\gamma_{k}^{[\pi_{k}(l_{k})]}$ is defined as
\begin{equation}
\label{eq:IBC_gamma}
\gamma_{k}^{[\pi_{k}(l_{k})]}
\! \triangleq \!
\alpha_{kk}^{[\pi_{k}(l_{k})]} \! +
\max \Bigl\{  \max_{l_{k}'':l_{k}'' > l_{k}}\{r_{k}^{[\pi_{k}(l_{k}'')]} \},
\max_{m_{k}:m_{k} \geq l_{k} } \bigl\{
 \bigl( \max_{(l_{j},j):j \neq k} \{ \alpha_{kj}^{[\pi_{k}(m_{k})]}
 +   r_{j}^{[l_{j}]} \} \bigr)^{+}  \! - \alpha_{kk}^{[\pi_{k}(m_{k})]}  \bigr\}     \Bigr\}.
\end{equation}
In \eqref{eq:IBC_GDoF_gamma}, $\alpha_{kk}^{[\pi_{k}(l_{k})]} + r_{k}^{[\pi_{k}(l_{k})]}$ is the
received power level of the desired signal at UE-$\big(\pi_{k}(l_{k}),k\big)$,
while $\gamma_{k}^{[\pi_{k}(l_{k})]}$ is the effective interference level,
which takes into account the fact that $X_{k}^{[\pi_{k}(l_{k})]n}$ is
decoded by all UEs in cell $k$, which follow UE-$\big(\pi_{k}(l_{k}),k\big)$ in the decoding order.
It can be verified that $\gamma_{k}^{[\pi_{k}(l_{k})]} \geq 0$.
In a similar manner, the GDoF inequality for the dual IMAC in \eqref{eq:IMAC_GDoF per user}
is expressed as
\begin{equation}
\label{eq:IMAC_GDoF_gamma}
\bar{d}_{k}^{[\pi_{k}(l_{k})]}  \leq
\left(  \alpha_{kk}^{[\pi_{k}(l_{k})]} + \bar{r}_{k}^{[\pi_{k}(l_{k})]}
  -  \bar{\gamma}_{k}^{[\pi_{k}(l_{k})]} \right)^{+},
\end{equation}
where the interference level $\bar{\gamma}_{k}^{[\pi_{k}(l_{k})]}$, which satisfies $\bar{\gamma}_{k}^{[\pi_{k}(l_{k})]} \geq 0$, is defined as
\begin{equation}
\label{eq:IMAC_gamma}
\bar{\gamma}_{k}^{[\pi_{k}(l_{k})]} \triangleq
\max \bigl\{0, \max_{l_{k}':l_{k}' < l_{k}} \{ \alpha_{kk}^{[\pi_{k}(l_{k}')]} +
\bar{r}_{k}^{[\pi_{k}(l_{k}')]} \}  ,\max_{(l_{j},j):j \neq k}
 \{ \alpha_{jk}^{[l_{j}]} +  \bar{r}_{j}^{[l_{j}]} \}   \bigr\}.
\end{equation}

Next, we present a GDoF-based power allocation duality result, formulated using the above defined quantities.
The proof of this result is relegated to Appendix \ref{appebdix:proof_of_duality_lemma}.
\begin{lemma}
\label{lemma:power_allocation_duality}
The following multi-cell TIN power allocation duality holds:
\begin{enumerate}
\item Consider an IBC TIN strategy identified by $(\bm{\pi},\mathbf{r})$, which achieves the set of GDoF tuples
$\mathcal{D}_{\mathrm{TINA}}^{\mathrm{IBC}}(\bm{\pi},\mathbf{r})$.
The dual IMAC TIN strategy $(\bm{\pi},\bar{\mathbf{r}})$, where each component of  $\bar{\mathbf{r}}$
is given by
\begin{equation}
\label{eq:power_allocation_duality_IMAC}
\bar{r}_{k}^{[\pi_{k}(l_{k})]} = -\gamma_{k}^{[\pi_{k}(l_{k})]},
\end{equation}
achieves the set of GDoF tuples $\mathcal{D}_{\mathrm{TINA}}^{\mathrm{IMAC}}(\bm{\pi},\bar{\mathbf{r}})$,
which includes $\mathcal{D}_{\mathrm{TINA}}^{\mathrm{IBC}}(\bm{\pi},\mathbf{r})$.
\item Consider an IMAC TIN strategy  identified by $(\bm{\pi},\bar{\mathbf{r}})$, which achieves the set of GDoF tuples
$\mathcal{D}_{\mathrm{TINA}}^{\mathrm{IMAC}}(\bm{\pi},\bar{\mathbf{r}})$.
The dual IBC TIN strategy $(\bm{\pi},\mathbf{r})$, where each component of  $\mathbf{r}$
is given by
\begin{equation}
\label{eq:power_allocation_duality_IBC}
r_{k}^{[\pi_{k}(l_{k})]} = -\bar{\gamma}_{k}^{[\pi_{k}(l_{k})]},
\end{equation}
achieves the set of GDoF tuples $\mathcal{D}_{\mathrm{TINA}}^{\mathrm{IBC}}(\bm{\pi},\mathbf{r})$,
which includes  $\mathcal{D}_{\mathrm{TINA}}^{\mathrm{IMAC}}(\bm{\pi},\bar{\mathbf{r}})$.
\end{enumerate}
\end{lemma}
The first statement in  the above lemma  implies 
that $\mathcal{D}^{\mathrm{IBC}}_{{\mathrm{TINA}}} \subseteq \mathcal{D}^{\mathrm{IMAC}}_{{\mathrm{TINA}}}$,
while its second statement implies  that $\mathcal{D}^{\mathrm{IMAC}}_{{\mathrm{TINA}}} \subseteq \mathcal{D}^{\mathrm{IBC}}_{{\mathrm{TINA}}}$.
Therefore, Lemma \ref{lemma:power_allocation_duality} leads directly to the following corollary. 
\begin{corollary}
\label{corollary:TIN_UL_DL_Duality}
The TINA regions for the IBC and the dual IMAC are identical, i.e.
$\mathcal{D}^{\mathrm{IBC}}_{{\mathrm{TINA}}}=\mathcal{D}^{\mathrm{IMAC}}_{{\mathrm{TINA}}}$.
\end{corollary}
Lemma \ref{lemma:power_allocation_duality} is a generalization of \cite[Lem. 1]{Geng2018}, which establishes 
a similar duality  for the regular $K$-user IC.
This explicit power allocation duality is useful
for solving GDoF-based TIN power control problems.
For instance, suppose that we have an algorithm that returns a desired Pareto optimal TINA GDoF tuple
$\mathbf{d}^{\star}$ for the IMAC and
a strategy $(\bm{\pi},\bar{\mathbf{r}})$ which achieves it.
Using the simple transformation in \eqref{eq:power_allocation_duality_IBC},
we obtain a strategy $(\bm{\pi},\mathbf{r})$ for the corresponding dual IBC that achieves the same GDoF tuple
$\mathbf{d}^{\star}$, which is Pareto optimal for the dual IBC as well.

While the algorithmic aspects of TIN and power control are of high interest in their own right, especially for practical purposes
as demonstrated in \cite{Geng2016,Yi2016,Geng2018}, in this work we are primarily
interested in obtaining an explicit characterization of the IBC TINA region, which in turn, enables us to identify regimes of channel strengths in which such region is convex and optimal 
(i.e. CTIN and TIN regimes for the IBC). 
Lemma \ref{lemma:power_allocation_duality} is still very useful in this regard as seen in Corollary \ref{corollary:TIN_UL_DL_Duality}.
\subsection{CTIN and TIN Regimes for the IBC}
\label{subsec:CTIN_TIN_IBC}
Equipped with the duality result in Lemma \ref{lemma:power_allocation_duality}, which enables us to characterize $\mathcal{D}^{\mathrm{IBC}}_{{\mathrm{TINA}}}$
using the characterization of $\mathcal{D}^{\mathrm{IMAC}}_{{\mathrm{TINA}}}$ as seen in Corollary \ref{corollary:TIN_UL_DL_Duality},
we are now ready to present our main result. 
\begin{theorem}
\label{theorem:TIN_IBC}
For the IBC described in Section \ref{sec:system model}, we have the following:
\begin{enumerate}
\item $\mathcal{D}^{\mathrm{IBC}}_{{\mathrm{TINA}}}$ is characterized by a union of polyhedral-TIN regions as
\begin{equation}
\label{eq:GDoF_region_IBC}
\mathcal{D}^{\mathrm{IBC}}_{{\mathrm{TINA}}} = \bigcup_{\mathcal{S} \subseteq \mathcal{K}} \bigcup_{\bm{\pi} \in \Pi(\mathcal{S})}
\mathcal{P}(\bm{\pi}, \mathcal{S}).
\end{equation}
\item In the CTIN regime (Definition \ref{def:CTIN_regime}), $\mathcal{D}^{\mathrm{IBC}}_{{\mathrm{TINA}}}$ 
is a convex polyhedron given by
\begin{equation}
\label{eq:GDoF_region_IBC_CTIN}
\mathcal{D}^{\mathrm{IBC}}_{{\mathrm{TINA}}} = \mathcal{P}(\bm{\mathrm{id}}, \mathcal{K}).
\end{equation}
\item In the TIN regime (Definition \ref{def:TIN_regime}), $\mathcal{D}^{\mathrm{IBC}}_{{\mathrm{TINA}}}$ is
optimal such that
\begin{equation}
\label{eq:GDoF_region_IBC_TIN}
\mathcal{D}^{\mathrm{IBC}}  = \mathcal{D}^{\mathrm{IBC}}_{{\mathrm{TINA}}} = \mathcal{P}(\bm{\mathrm{id}}, \mathcal{K}).
\end{equation}
\end{enumerate}
\end{theorem}
The characterization in \eqref{eq:GDoF_region_IBC} and the convexity result in \eqref{eq:GDoF_region_IBC_CTIN}
follow by combining their IMAC counterparts in Theorem \ref{theorem:TIN_IMAC}, i.e. \eqref{eq:GDoF_region_IMAC} and \eqref{eq:GDoF_region_IMAC_CTIN}, respectively, 
with the TINA region duality in Corollary \ref{corollary:TIN_UL_DL_Duality}.
The TIN optimality result in \eqref{eq:GDoF_region_IBC_TIN}, however, requires a new information-theoretic outer bound for the IBC.
This new outer bound is presented  in the next section. 

The remainder of this section is dedicated to understanding the CTIN and TIN conditions in light of the IBC result in 
Theorem  \ref{theorem:TIN_IBC}.
We refer to the conditions in \eqref{eq:CTIN_cond_1} and \eqref{eq:TIN_cond_1} as BC-type conditions, as their main purpose is to
govern the \emph{order} of users within the same cell (or BC) in the presence of inter-cell interference. 
On the other hand, the conditions in \eqref{eq:CTIN_cond_2} and \eqref{eq:TIN_cond_2}
are referred to as IC-type conditions, as their main purpose is to guarantee sufficiently low levels of (inter-cell) interference,
and they are identical to their IC counterparts in \cite{Yi2016} and \cite{Geng2015}, respectively.
These points are further clarified through the next discussion, which we keep at an intuitive level.
Rigorous proofs of the claims we make  are deferred to the following parts of the paper. 
\subsubsection{2-cell, 3-user Network}
\label{subsubsec:2_cell_3_user}
We find it more instructive to conduct our treatment  of the CTIN and TIN conditions using the 2-cell, 3-user network in Fig. \ref{fig:signal_levels_TIN}(a).
This seemingly simple network captures the most essential aspects of downlink multi-cell TIN, and insights developed here extend to the general IBC.
For ease of exposition, we adopt more brief notation in this part, as shown in Fig. \ref{fig:signal_levels_TIN}(a). 
The $3$ users are labelled by $a$, $b$ and $c$,
and their corresponding received signals are given by
\begin{align}
\label{eq:signal_model_y_a}
Y_{a}(t) & = a_{1}X_{1}(t) +  a_{2}X_{2}(t) + Z_{a}(t) \\
\label{eq:signal_model_y_b}
Y_{b}(t) & = b_{1}X_{1}(t) +  b_{2}X_{2}(t) + Z_{b}(t) \\
\label{eq:signal_model_y_c}
Y_{c}(t) & = c_{1}X_{1}(t) +  c_{2}X_{2}(t) + Z_{c}(t).
\end{align}
For all $i \in \{1,2\}$, we have $|a_{i}|^{2} = P^{\alpha_{i}}$, 
$|b_{i}|^{2} = P^{\beta_{i}}$ and $|c_{i}|^{2} = P^{\gamma_{i}}$, where $\alpha_{i}$, $\beta_{i}$ and $\gamma_{i}$ are the corresponding channel strengths.
Transmitter $1$ is associated with users $a$ and $b$, and together they form cell $1$. 
On the other hand, cell $2$ is formed by transmitter 2 and user $c$.
Without loss of generality, we assume that user $b$ is the stronger BC user in cell $1$,  i.e. $\beta_{1} \geq \alpha_{1}$.
Moreover, corresponding rate and GDoF tuples are denoted by $(R_{a},R_{b},R_{c})$ and $(d_{a},d_{b},d_{c})$, respectively. 

From Definition \ref{def:CTIN_regime}, the CTIN conditions for the above 2-cell, 3-user network are listed as follows:
\begin{align}
\label{eq:CTIN_cond_2_cell_3_user_1}
\beta_{1} - \beta_{2} & \geq \alpha_{1} - \alpha_{2} \\
\label{eq:CTIN_cond_2_cell_3_user_2}
\alpha_{1} & \geq  \alpha_{2} + \gamma_{1} \\
\label{eq:CTIN_cond_2_cell_3_user_3}
\gamma_{2} & \geq  \gamma_{1}  + \max \{  \alpha_{2} , \beta_{2}  \}.
\end{align}
It can be verified that for $2$-cell networks, the CTIN IC-type condition in \eqref{eq:CTIN_cond_2} reduces to the TIN IC-type condition in \eqref{eq:TIN_cond_2}.
Therefore, the TIN conditions of Definition \ref{def:TIN_regime} for the above 2-cell, 3-user network 
are given by \eqref{eq:CTIN_cond_2_cell_3_user_2} and \eqref{eq:CTIN_cond_2_cell_3_user_3}, in addition to the following BC-type condition:
\begin{subequations}
\label{eq:TIN_cond_2_cell_3_user}
\begin{align}
\label{eq:TIN_cond_2_cell_3_user_a}
\beta_{1} - \beta_{2} & \geq  \alpha_{1} \  \  \text{or} \\
\label{eq:TIN_cond_2_cell_3_user_b}
\beta_{1} - \beta_{2}   & \geq \alpha_{1}  - (\alpha_{2}  - \beta_{2}).
\end{align}
\end{subequations}
\subsubsection{Redundancy Order}
\label{subsubsec:redun_order}
To further set the stage for the following treatment, we establish the notion of \emph{redundancy} order, 
which is key to interpreting the CTIN and TIN conditions.
To this end, let us first consider cell $1$ of the above 2-cell, 3-user network in isolation, which is a $2$-user degraded Gaussian BC. 
It is well known that the degradedness of this channel imposes an order amongst users, where user $b$ is stronger than user $a$.
This degraded order implies a less restrictive type of order, 
known as the \emph{less noisy} order \cite[Ch. 5]{ElGamal2011}, where for all Markov chains $W \rightarrow X_{1}^{n} \rightarrow (Y_{a}^{ n},Y_{b}^{ n})$, we have
\begin{equation}
\label{eq:less_noisy_degraded_BC}
I\big(W ; Y_{a}^{ n}   \big)  \leq  I \big(W ;Y_{b}^{n}   \big).
\end{equation}
An implication of \eqref{eq:less_noisy_degraded_BC} is that user $b$ (the less noisy user) can decode whatever user $a$ (the more noisy user)  decodes.
This less noisy order, in turn, implies a further less restrictive type of order, which we call the \emph{redundancy} order.
For cell $1$ in isolation, the redundancy order is specified by
\begin{equation}
\max_{(R_{a},R_{b}) \in \mathcal{C}^{\mathrm{BC}}} R_{a} + R_{b}  \ \ = \max_{(R_{a},R_{b}) \in \mathcal{C}^{\mathrm{BC}}} R_{b},
\end{equation}
where $\mathcal{C}^{\mathrm{BC}}$ denotes the corresponding BC capacity region.
That is, user $a$ is redundant with respect to user $b$ in the sense that the network's sum-capacity is achieved by omitting the former user 
(i.e. setting $R_{a} = 0$).
While it is clear that all three orders hold for users of cell $1$ in isolation, discrepancies start to surface once inter-cell interference from cell $2$  comes into play.

In the presence of cell $2$,  users of cell $1$ may receive different levels of inter-cell interference, which may tamper with their order.
It is easy to perceive that the above-described orders, that hold in the absence of inter-cell interference, 
are generally not preserved under inter-cell interference. 
As it turns out however, some orders are preserved under the TIN conditions in 
\eqref{eq:CTIN_cond_2_cell_3_user_2}--\eqref{eq:TIN_cond_2_cell_3_user}, albeit in a weaker, constant-rate-gap sense.
For instance, in the case where  \eqref{eq:TIN_cond_2_cell_3_user_a} holds, while the degraded order of users in cell $1$ is lost, 
the less noisy order remains approximately intact such that
\begin{equation}
\label{eq:less_noisy_constant}
I\big(W ; Y_{a}^{ n}   \big)  \leq  I \big(W ;Y_{b}^{n}   \big) + n \times \mathrm{constant}.
\end{equation}
This approximate less noisy order, however,  is also lost in the TIN regime when \eqref{eq:TIN_cond_2_cell_3_user_a} is violated.
Nevertheless, in this case \eqref{eq:TIN_cond_2_cell_3_user_b} preserves the redundancy, approximately, such that 
\begin{equation}
\label{eq:redundancy_order}
\max_{(R_{a},R_{b},R_{c}) \in \mathcal{C}^{\mathrm{IBC}}} R_{a} + R_{b} + R_{c}  \  \  \leq  \max_{(R_{a},R_{b},R_{c}) \in \mathcal{C}^{\mathrm{IBC}}} R_{b} + R_{c}  + \mathrm{constant}.
\end{equation}
That is, in the presence of user $b$ (the necessary user), user $a$ (the redundant user) contributes at most a constant to the sum-capacity, which is approximately achieved by omitting user $a$.
In the GDoF sense, the presence of user $a$ does not increase the overall GDoF, as \eqref{eq:redundancy_order} translates to
\begin{equation}
\label{eq:redundancy_order_GDoF}
\max_{(d_{a},d_{b},d_{c}) \in \mathcal{D}^{\mathrm{IBC}}} d_{a} + d_{b} + d_{c}  \  \  = \max_{(d_{a},d_{b},d_{c}) \in \mathcal{D}^{\mathrm{IBC}}} d_{b} + d_{c}.
\end{equation}
Since the less noisy order in \eqref{eq:less_noisy_constant} implies the redundancy order in \eqref{eq:redundancy_order}, 
it turns out that the TIN regime has a special attribute of preserving the redundancy order of user $a$ and user $b$ under inter-cell interference, which is not necessarily the case outside the TIN regime.
This is further elaborated by taking a closer look at  CTIN and TIN regimes in light of the above orders. 
\subsubsection{CTIN Regime}
\label{subsubsec:CTIN_insights}
Starting with the BC-type condition in  \eqref{eq:CTIN_cond_2_cell_3_user_1}, this imposes a signal to interference ratio (SIR) order on the users  of cell $1$.
In particular,  user $b$, which is stronger than user $a$ in the SNR sense, 
remains stronger in the SIR sense (i.e. under inter-cell interference). 
This condition ensures that the natural decoding order $\bm{\mathrm{id}}$,  which is optimal in the absence of inter-cell interference, 
remains optimal  under inter-cell interference, where optimality here is with respect to the TINA region.\footnote{This optimality of the identity order does not necessarily hold when \eqref{eq:CTIN_cond_2_cell_3_user_1} is violated, see \cite[Fig. 2(c)]{Joudeh2019a}.}

While the above BC-type condition ensures that the intra-cell order of SNRs  is inherited by the corresponding SIRs, it does not necessarily guarantee low levels of inter-cell interference. 
This, in turn, is guaranteed by the IC-type conditions in \eqref{eq:CTIN_cond_2_cell_3_user_2} and \eqref{eq:CTIN_cond_2_cell_3_user_3}. 
To see this, we recall from Remark \ref{remark:CTIN_TIN_other_users}  that in addition to the IC-type conditions for users $a$ and $c$, given in \eqref{eq:CTIN_cond_2_cell_3_user_2} and \eqref{eq:CTIN_cond_2_cell_3_user_3} respectively, an implicit IC-type condition for user $b$ holds as well, 
i.e. $\beta_{1}  \geq  \beta_{2} + \gamma_{1}$.
These conditions, by regular IC terms, guarantee that each transmitter-receiver pair is in the ``very weak" interference regime \cite{Etkin2008,Geng2015}.
A key implication is that all GDoF benefits of using time-sharing over subnetworks in tandem with TIN are eliminated, 
hence guaranteeing that the TINA region $\mathcal{D}_{\mathrm{TINA}}^{\mathrm{IBC}}$  is convex. 
In this case, the TINA region is described by all non-negative 
tuples $(d_{a},d_{b},d_{c})$ that satisfy
\begin{align}
\label{eq:polyhedral_TINA_2_cell_1}
d_{a} & \leq \alpha_{1} \\
\label{eq:polyhedral_TINA_2_cell_2}
d_{a} + d_{b} & \leq \beta_{1} \\
\label{eq:polyhedral_TINA_2_cell_3}
d_{a} + d_{c} & \leq (\alpha_{1} - \alpha_{2}) + (\gamma_{2} - \gamma_{1}) \\
\label{eq:polyhedral_TINA_2_cell_4}
d_{a} + d_{b} + d_{c} & \leq (\beta_{1} - \beta_{2}) + (\gamma_{2} - \gamma_{1}).
\end{align}

From the viewpoint of the above-described TINA region, user $a$, which is weaker than user $b$ in the BC sense, 
appears to be redundant with respect to user $b$. 
This is seen through the sum-GDoF inequalities in \eqref{eq:polyhedral_TINA_2_cell_2} and \eqref{eq:polyhedral_TINA_2_cell_4}, which can be achieved by omitting user $a$ (i.e. setting $d_{a} = 0$).
This gives rise to the question of whether the SIR order, imposed by the IBC-type condition in \eqref{eq:CTIN_cond_2_cell_3_user_1}, 
implies a redundancy order as the one in \eqref{eq:redundancy_order_GDoF}.
The answer to this question is: not necessarily.
The redundancy order exhibited by the TINA region in \eqref{eq:polyhedral_TINA_2_cell_1}--\eqref{eq:polyhedral_TINA_2_cell_4} is, in general, merely 
an artefact of employing the TIN scheme, which due to the lack of structure, is only sensitive to the SIRs. 
This is particularly the case in the sub-regime where the CTIN conditions are satisfied but the TIN conditions are not.

To elaborate, we first observe from \eqref{eq:polyhedral_TINA_2_cell_1}--\eqref{eq:polyhedral_TINA_2_cell_4} that when 
considering all three users $a$, $b$ and $c$, or users $b$ and $c$ only, 
the maximum achievable sum-GDoF using TIN 
is given by 
\begin{equation}
\label{eq:sum_GDoF_users_a_b_c_TIN}
d_{\mathrm{TINA}}^{\mathrm{IBC}} = (\beta_{1} - \beta_{2}) + (\gamma_{2} - \gamma_{1}).
\end{equation}
In Appendix \ref{appendix:IA}, we show that a scheme based on IA with structured codes  strictly surpass $d_{\mathrm{TINA}}^{\mathrm{IBC}}$, 
achieving a sum-GDoF of $d_{\mathrm{IA}}^{\mathrm{IBC}} = d_{\mathrm{TINA}}^{\mathrm{IBC}} + \gamma_{\mathrm{IA}}$,
where $\gamma_{\mathrm{IA}} > 0$  given that \eqref{eq:CTIN_cond_2_cell_3_user_1} holds and \eqref{eq:TIN_cond_2_cell_3_user} is strictly 
violated.\footnote{This holds throughout this sub-regime, except for a set of channel coefficients of measure zero.}
User $a$ is essential for achieving this strict IA gain, which is further confirmed by observing that the sum-GDoF of users $b$ and $c$ is,
as a matter of fact, bounded above by
\begin{equation}
\label{eq:sum_GDoF_users_b_c}
\max_{(d_{a},d_{b},d_{c}) \in \mathcal{D}^{\mathrm{IBC}}}  d_{b} + d_{c} \leq (\beta_{1} - \beta_{2}) + (\gamma_{2} - \gamma_{1}),
\end{equation}
which is an information-theoretic bound that holds due to the IC-type TIN conditions \cite{Etkin2008}.
Putting together the above pieces, it follows that in this sub-regime we have
\begin{equation}
\max_{(d_{a},d_{b},d_{c}) \in \mathcal{D}^{\mathrm{IBC}}}  d_{a} + d_{b} + d_{c} \ \ > \max_{(d_{a},d_{b},d_{c}) \in \mathcal{D}^{\mathrm{IBC}}}  d_{b} + d_{c},
\end{equation}
hence showing the necessity, rather than redundancy, of user $a$. More importantly, 
this shows that the CTIN conditions in Definition \ref{def:CTIN_regime} 
are insufficient for TIN optimality in the IBC. 
This is in line with a similar observation made for the dual IMAC in \cite{Gherekhloo2016} (see also \cite[Rem. 7]{Joudeh2019a}).
\subsubsection{TIN Regime}
\label{subsubsec:TIN_insights}
The difference between the CTIN and TIN regimes of the 2-cell, 3-user network is due to the BC-type condition in \eqref{eq:TIN_cond_2_cell_3_user}.
As we further elaborate next, this condition imposes an order on the users of cell $1$ which guarantees, in the information-theoretic sense, that user $a$ is redundant with respect to user $b$, in accordance with \eqref{eq:redundancy_order_GDoF} in Section \ref{subsubsec:redun_order}.

Starting with the branch in \eqref{eq:TIN_cond_2_cell_3_user_a}, this condition guarantees that in cell $1$, the SNR of the weaker BC user (user $a$) 
is no more than the SIR of the stronger BC user (user $b$).  
As we show in  Section \ref{subsec:lemmas}, this condition on channel strengths imposes a less noisy order on the users of cell $1$ as the one in \eqref{eq:less_noisy_constant}, which implies that user $b$ can approximately decode whatever user $a$ decodes.
The condition in \eqref{eq:TIN_cond_2_cell_3_user_a} is illustrated in Fig. \ref{fig:signal_levels_TIN}(b) using signal levels, measured in terms of channel strength parameters (i.e. exponents of $P$).
It is seen that all levels of the desired signal $X_{1}$ observed by  user $a$ are received by user $b$ above inter-cell interference levels, caused by 
the interfering signal $X_{2}$.  This enables user $b$ to retrieve all useful signal levels received by user $a$.

\begin{figure}
\centering
\includegraphics[width = 1.0\textwidth,trim={0cm 0cm 0cm 0cm},clip]{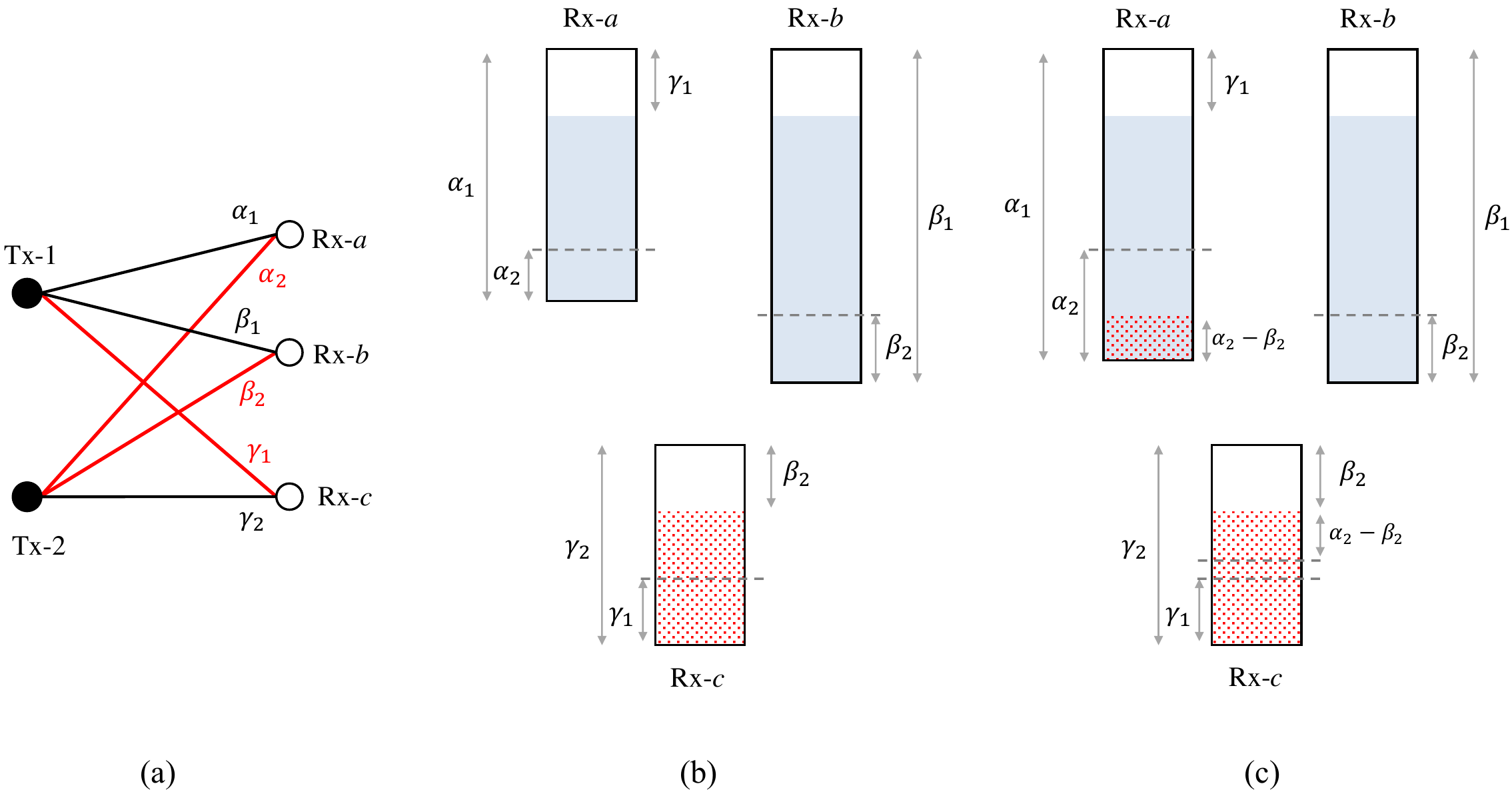}
\caption{\small
(a) 2-cell, 3-user network, and two examples of received signal power levels in the TIN regime: (b) the case where $\beta_{1} - \beta_{2} \geq \alpha_{1}$, and (c) the case where $\beta_{1} - \beta_{2} < \alpha_{1}$ and $\beta_{1} - 2\beta_{2} \geq \alpha_{1} - \alpha_{2}$. Levels in white, blue and dotted red represent empty levels, $X_{1}$ and $X_{2}$, respectively.}
\label{fig:signal_levels_TIN}
\end{figure}
Next, we consider the branch in \eqref{eq:TIN_cond_2_cell_3_user_b}, which  is most relevant whenever \eqref{eq:TIN_cond_2_cell_3_user_a} is strictly violated, i.e.
$\beta_{1} - \beta_{2} \geq \alpha_{1} - (\alpha_{2} - \beta_{2})$ and $\beta_{1} - \beta_{2} < \alpha_{1}$.
Contrary to the previous branch, the less noisy order in cell $1$ is not preserved in this case.
Nevertheless, the redundancy order in cell $1$, i.e. \eqref{eq:redundancy_order_GDoF}, 
is preserved here.
To see this, we first highlight that since \eqref{eq:TIN_cond_2_cell_3_user_b} holds and \eqref{eq:TIN_cond_2_cell_3_user_a} does not, user $b$ is strictly  \emph{less interfered with} compared to user $a$ , i.e. 
$\beta_{2} < \alpha_{2}$.
Hence, as shown through the illustrative example in  Fig. \ref{fig:signal_levels_TIN}(c), the signal levels of user $a$ can be partitioned into two parts: 
\begin{itemize}
\item \emph{The upper $\alpha_{1} - (\alpha_{2} - \beta_{2})$ levels}: All signal levels of $X_{1}$ observed by user $a$ 
through this part are received by user $b$ above inter-cell interference levels. 
This holds since user $a$ has an effective SNR of $\alpha_{1} - (\alpha_{2} - \beta_{2})$  in this upper part, 
which is no more than the SIR of user $b$.
Hence, similar to the first branch \eqref{eq:TIN_cond_2_cell_3_user_a}, user $b$ is less noisy than user $a$ in this upper part.
\item \emph{The lower  $(\alpha_{2} - \beta_{2})$ levels:} In this part, user $a$ receives levels of the interfering signal $X_{2}$ which 
the stronger user $b$ does not receive (see Fig. \ref{fig:signal_levels_TIN}(c)).
This part hence can be exploited to grant user $a$  access to levels of the desired signal $X_{1}$,  which may be corrupted at user $b$ by higher levels of 
the interfering signal $X_{2}$. For instance, this occurs if cell $2$ abstains from transmitting at (some of) these lower  levels  received by user $a$, and transmits at higher levels instead.
The GDoF gain achieved by user $a$ through this lower part of signal levels is at most $(\alpha_{2} - \beta_{2})$. 
Nevertheless, it can be shown that such gain can only be realized at the expense of the GDoF achieved by user $b$ or user $c$, where an equal loss is incurred.
\end{itemize}
The above partition shows that in the presence of user $b$,  user $a$ does not contribute an effective increase to the overall GDoF.
Therefore, condition in \eqref{eq:TIN_cond_2_cell_3_user_b} imposes a redundancy order on the users of cell $1$ in which  user $a$ is redundant with respect to user $b$. 

By combining the above observations, we reach the conclusion that the interference-free redundancy order of users in cell $1$ is preserves in the GDoF sense under inter-cell interference in the TIN regime. 
This redundancy of user $a$ implies that the information-theoretic bound on the sum-GDoF of users $b$ and $c$, given in \eqref{eq:sum_GDoF_users_b_c}, is also valid as a bound on the sum-GDoF $a$, $b$ and $c$.
Combining this with existing results, it follows that TINA region in \eqref{eq:polyhedral_TINA_2_cell_1}--\eqref{eq:polyhedral_TINA_2_cell_4} is also an information-theoretic outer bound,
hence establishing the optimality of TIN for this network in the TIN regime. 
\begin{remark}
It is worthwhile highlighting that the TIN conditions identified for the IBC in this work 
are only sufficient for TIN optimality, and we make no claim of necessity. 
This is also the case for the majority of TIN-optimality results in the literature, as noted in \cite[Rem. 1]{Geng2016} (see also  \cite[Rem. 8]{Joudeh2019a}).
A representative example is the sufficient TIN condition identified for the regular IC in \cite{Geng2015}.
It was conjectured in \cite{Geng2015} that such condition is also necessary for TIN-optimality  in the $K$-user IC, except for a set of channel
gain values of measure zero.
This conjecture remains open. 
\hfill $\lozenge$
\end{remark}
\section{Outer Bound}
\label{sec:proof_of_outerbound}
We start this section by stating the converse result. 
\begin{theorem}
\label{theorem:outer_bound}
In the TIN regime (Definition \ref{def:TIN_regime}),
the capacity region of the IBC described in Section \ref{sec:system model},  denoted by  $\mathcal{C}^{\mathrm{IBC}}(P)$, 
is included in the set of non-negative rate tuples satisfying:
\begin{align}
\label{eq:capacity_outer_1}
\sum_{s_{i}\in \langle l_{i} \rangle}R_{i}^{[s_{i}]} & \leq \alpha_{ii}^{[l_{i}]}\log ( P ) + O(1),
\; l_{i} \in \langle L_{i} \rangle, \forall i \in \langle K \rangle\\
\nonumber
\sum_{j \in \langle m \rangle } \sum_{s_{i_{j}} \in \langle l_{i_{j}} \rangle} R_{i_{j}}^{[s_{i_{j}}]} & \leq
\sum_{j \in \langle m \rangle } \big( \alpha_{i_{j}i_{j}}^{[l_{i_{j}}]}-\alpha_{i_{j}i_{j-1}}^{[l_{i_{j}}]} \big)
\log( P ) + O(1), \\
\label{eq:capacity_outer_2}
\forall l_{i_{j}} \in \langle L_{i_{j}} \rangle, \; & (i_{1},\ldots,i_{m}) \in \Sigma\big(\langle K \rangle\big), m \in \langle 2:K \rangle.
\end{align}
In \eqref{eq:capacity_outer_2}, a modulo-$m$ operation is implicitly used on cell indices such that $i_{0} = i_{m}$.
\end{theorem}
The $O(1)$ terms in \eqref{eq:capacity_outer_1} and \eqref{eq:capacity_outer_2} are constants with respect to $\log(P)$,
yet depend on the size of the network, specified by $K$ and $L_{1},\ldots,L_{K}$.
Therefore, it readily follows that in the GDoF sense, the outer bound in Theorem \ref{theorem:outer_bound} translates to 
$\mathcal{P}(\bm{\mathrm{id}} , \mathcal{K})$, as described in \eqref{eq:polyhedral_TIN_region_1} and \eqref{eq:polyhedral_TIN_region_2}.
It follows that $\mathcal{D}^{\mathrm{IBC}} = \mathcal{P}(\bm{\mathrm{id}} , \mathcal{K}) = \mathcal{D}^{\mathrm{IBC}}_{\mathrm{TINA}}$,
hence proving the last point in Theorem \ref{theorem:TIN_IBC}.

It is worthwhile highlighting that the outer bound in Theorem \ref{theorem:TIN_IBC} lends itself to a constant-gap characterization 
of the entire capacity region $\mathcal{C}^{\mathrm{IBC}}(P)$ in the TIN regime, where the gap is independent of $P$ and may only depend on the size of the network.
Such characterization can be derived through a direct application of the steps used for the $K$-user 
IC in \cite[Th. 4]{Geng2015}  (see also for the $K$-cell IMAC in \cite[Th. 4]{Joudeh2018}). 
We omit the gap characterization from this paper, as it gives no new insights. 
The remainder of this section is dedicated to proving Theorem \ref{theorem:outer_bound}, where 
the TIN conditions in Definition \ref{def:TIN_regime} are assumed to hold throughout the proof. 
\subsection{Auxiliary Lemmas}
\label{subsec:lemmas}
We commence the outer bound proof by presenting two instrumental lemmas. 
For convenience, we adopt the  notation of the 2-cell, 3-user network in Section \ref{subsubsec:2_cell_3_user} while presenting and proving these lemmas. 
We consider a model with transmitters $1$ and $2$ and receivers $a$ and $b$, with input-output relationship given by \eqref{eq:signal_model_y_a} and \eqref{eq:signal_model_y_b}.
Note that receiver  $c$ is not required here.
We  further assume that $\alpha_{1} \geq \alpha_{2} \geq 0$ and $\beta_{1} \geq \beta_{2} \geq 0$.
Moreover, in addition to variables defined in the signal model, we further consider an arbitrary random variable
$W \sim F(w)$, independent of $X_{2}^{n}$, $Z_{a}^{n}$ and $Z_{b}^{n}$,
with cumulative distribution function given by $F(w)$, and which forms a Markov chain as:
\begin{equation}
\label{eq:Markov_chain_lemmas}
W \rightarrow X_{1}^{n} \rightarrow \big(Y_{a}^{n}, Y_{b}^{n} \big).
\end{equation}
We are now ready to state our lemmas.  
\begin{lemma}
\textnormal{\textbf{(Less Noisy under Interference).}}
\label{lemma:more_capable}
Assume that the following condition holds
\begin{equation}
\label{eq:more_capable_condition}
\beta_{1} - \beta_{2}  \geq \alpha_{1}.
\end{equation}
Then for all $X_{1}^{n}$, $X_{2}^{n}$ and $W$ as defined above, we have  
\begin{equation}
\label{eq:more_capable_inequality}
I\big(W ; Y_{a}^{ n}   \big)  \leq  I \big(W ;Y_{b}^{n}   \big) +  n.
\end{equation}
\end{lemma}
\begin{lemma}
\label{lemma:diff_entropies}
Assume that the following conditions hold
\begin{equation}
\label{eq:TIN_BC_condition_lemma}
\beta_{1} - 2\beta_{2}  \geq \alpha_{1} - \alpha_{2} 
\ \ \text{and} \ \
\alpha_{2} \geq \beta_{2}.
\end{equation}
Then for all $X_{1}^{n}$, $X_{2}^{n}$ and $W$ as defined above, we have  
\begin{equation}
\label{eq:diff_entropies}
h\big(Y_{a}^{ n} | W  \big) 
\leq  h \big(Y_{b}^{n} | W  \big) +  n\big(  \alpha_{2} - \beta_{2} \big) \log (P) + n \log(6).
\end{equation}
\end{lemma}
We now proceed with some high level insights
which echo some of the points mentioned earlier in Section \ref{subsubsec:TIN_insights}.
The formal proofs of Lemma \ref{lemma:more_capable} and Lemma \ref{lemma:diff_entropies} are relegated to Appendix \ref{appendix:proof_lemmas_more_capable_diff_entropies}.

Let us first assume that both
receivers $a$ and $b$ are interested in retrieving $W$, communicated through the signal $X_{1}^{n}$, while the signal $X_{2}^{n}$ is seen as interference by both receivers. 
Lemma \ref{lemma:more_capable} gives  a condition under which receiver $b$ is \emph{less noisy} than receiver $a$, up to a constant rate gap, and hence in the GDoF sense.
This condition corresponds to the SIR
at  receiver $b$ being no less than the SNR (or signal power) at receiver $a$.
In this case, all information about $W$ (and, in fact, $X_{1}^{n}$) contained in the observation of receiver $a$ can be 
(approximately) retrieved from the upper, interference-free, signal levels observed by receiver $b$, 
hence making the latter the less noisy receiver. 
Furthermore, as a special case, 
it follows that under the less noisy condition in \eqref{eq:more_capable_condition}, we have
\begin{equation}
I\big(X_{1}^{n} ; Y_{a}^{ n}   \big)  \leq  I \big(X_{1}^{n}  ;Y_{b}^{n}   \big) +  n,
\end{equation}
for which we say that receiver $b$ is \emph{more capable}\footnote{For a detailed exposition of the \emph{less noisy} and \emph{more capable} notions, originally introduced in the context of the classical BC under no interference, readers are referred to \cite[Ch. 5]{ElGamal2011}.} than receiver $a$, up to a constant rate gap and hence  in the GDoF sense, under interference from transmitter $2$. 

Moving on to Lemma \ref{lemma:diff_entropies}, and omitting the conditioning on $W$ for ease of exposition, 
$h\big(Y_{a}^{ n} \big) $ may be roughly decomposed into contributions from the upper 
$\alpha_{1} - (\alpha_{2} - \beta_{2})$ signal levels
and contributions from the lower $(\alpha_{2} - \beta_{2})$ signal levels, observed by receiver  $a$ (see Fig. \ref{fig:signal_levels_TIN}(c)). 
On the GDoF scale, the latter contribute at most $( \alpha_{2} - \beta_{2} )$ to the difference of entropies given by $h\big(Y_{a}^{ n} \big) - h\big(Y_{b}^{ n} \big)$.
On the other hand, by considering only the upper $\alpha_{1} - (\alpha_{2} - \beta_{2})$ signal levels at receiver $a$, this receiver  
becomes more noisy (and hence less capable) than receiver $b$, and both receivers further see similar levels of interference, i.e. $\beta_{2}$.
Therefore, in the GDoF sense, the upper $\alpha_{1} - (\alpha_{2} - \beta_{2})$ signal levels at receiver $a$ 
do not contribute to creating a positive difference of entropies between receivers $a$ and $b$.
Therefore, the total difference of entropies, while considering all signal levels, is bounded above by $(\alpha_{2} - \beta_{2})$ in the GDoF sense.
\subsection{Proof of Outer Bound: 2-Cell, 3-User Example}
Equipped with Lemma \ref{lemma:more_capable} and Lemma \ref{lemma:diff_entropies}, we now proceed to prove
Theorem \ref{theorem:outer_bound}. 
Due to the multi-cell nature of the setting, the general proof tends to be notationally cumbersome. Therefore, it is instructive to start by considering a simpler special case.
For this purpose, we consider the 2-cell 3-user network of Section \ref{subsubsec:2_cell_3_user}.
Nevertheless, at this stage we revert back to the general notation of Section \ref{sec:system model} to emphasize the links with the 
general proof, presented further on. 

Specialized to this network, the TIN conditions in Definition \ref{def:TIN_regime} become:
\begin{align}
\label{eq:TIN_conditions_example_1}
\alpha_{11}^{[2]} & \geq \alpha_{11}^{[1]} + \alpha_{12}^{[2]} \ \text{or} \ 
\alpha_{11}^{[2]} \geq \alpha_{11}^{[1]} - \alpha_{12}^{[1]} + 2\alpha_{12}^{[2]}  \\
\label{eq:TIN_conditions_example_2}
\alpha_{11}^{[1]} & \geq \alpha_{12}^{[1]} + \alpha_{21}^{[1]}  \\
\label{eq:TIN_conditions_example_3}
\alpha_{22}^{[1]} & \geq \alpha_{12}^{[1]} + \alpha_{21}^{[1]} \\
\label{eq:TIN_conditions_example_4}
\alpha_{22}^{[1]} & \geq \alpha_{12}^{[2]} + \alpha_{21}^{[1]}.
\end{align}
Moreover, the outer bound in Theorem \ref{theorem:outer_bound} becomes:
\begin{align}
\label{eq:outer_bound_example_1}
R_{1}^{[1]} & \leq \alpha_{11}^{[1]}\log ( P ) + O(1) \\
\label{eq:outer_bound_example_2}
R_{1}^{[1]} + R_{1}^{[2]} & \leq \alpha_{11}^{[2]}\log ( P ) + O(1)\\
\label{eq:outer_bound_example_3}
R_{2}^{[1]}  & \leq \alpha_{22}^{[1]}\log ( P ) + O(1) \\
\label{eq:outer_bound_example_4}
R_{1}^{[1]} + R_{2}^{[1]} & \leq \big( \alpha_{11}^{[1]}-\alpha_{12}^{[1]} \big)
\log( P ) + \big( \alpha_{22}^{[1]}-\alpha_{21}^{[1]} \big)
\log( P ) + O(1) \\
\label{eq:outer_bound_example_5}
R_{1}^{[1]} + R_{1}^{[2]} + R_{2}^{[1]} & \leq \big( \alpha_{11}^{[2]}-\alpha_{12}^{[2]} \big)
\log( P ) + \big( \alpha_{22}^{[1]}-\alpha_{21}^{[1]} \big)
\log( P ) + O(1).
\end{align}
It is evident that \eqref{eq:outer_bound_example_1}--\eqref{eq:outer_bound_example_3} are all single-cell bounds, 
and hence follow from standard results in information theory, namely the 
capacity of the Gaussian point-to-point channel and the sum-capacity of the degraded Gaussian degraded BC \cite{Cover2012}.
\eqref{eq:outer_bound_example_4} is essentially a $2$-user IC bound, 
which holds due to the TIN conditions in \eqref{eq:TIN_conditions_example_2} and  \eqref{eq:TIN_conditions_example_3}
\cite{Etkin2008,Geng2015}.
Therefore, we focus on proving the $3$-user sum-rate bound in \eqref{eq:outer_bound_example_5}. 
To this end, we consider the two following cases that  constitute \eqref{eq:TIN_conditions_example_1}.
\begin{enumerate}[label=C.{\arabic*}]
\item $\alpha_{11}^{[2]} - \alpha_{12}^{[2]} \geq \alpha_{11}^{[1]}$: In this case, according to Lemma \ref{lemma:more_capable},
 UE-$(2,1)$ is less noisy than  UE-$(1,1)$ under interference from cell $2$.
\label{case1:converse_example}
This is used to bound the sum rate $R_{1}^{[1]} + R_{1}^{[2]}$ using a single mutual information term as we see next. 
Starting from Fano's inequality, we have
\begin{align}
n \big( R_{1}^{[1]} +  R_{1}^{[2]}  -  2\epsilon \big)   & \leq
 I\big(W_{1}^{[1]} ; Y_{1}^{[1]n}  \big)  +   I\big(W_{1}^{[2]} ; Y_{1}^{[2]n}  \big) \\
\label{eq:fano_case1_11}
& \leq   I\big(W_{1}^{[1]} ; Y_{1}^{[2]n} \big) + n 
+ I\big(W_{1}^{[2]} ; Y_{1}^{[2]n}  \big) \\
& \leq   I\big(W_{1}^{[1]} ; Y_{1}^{[2]n} | W_{1}^{[2]} \big) + n 
+ I\big(W_{1}^{[2]} ; Y_{1}^{[2]n}  \big) \\
& =  I\big(W_{1}^{[1]} ,W_{1}^{[2]}  ; Y_{1}^{[2]n} \big) + n \\
\label{eq:fano_case1_12}
& \leq  I\big(X_{1}^{n} ; Y_{1}^{[2]n} \big)  + n.
\end{align}
In the above, the critical step is \eqref{eq:fano_case1_11}, which follows directly from the less noisy result in Lemma \ref{lemma:more_capable}.
Moving on to cell 2,  the rate of UE-$(1,2)$ is bounded as
\begin{align}
\label{eq:fano_case1_2}
n \big(R_{2}^{[1]} - \epsilon \big)  & \leq I\big(X_{2}^{n} ; Y_{2}^{[1]n} \big).
\end{align}
From \eqref{eq:fano_case1_12} and \eqref{eq:fano_case1_2}, it is evident that
the setting is (approximately) reduced to a 2-user IC with BS-$1$ and BS-$2$ as transmitters and UE-$(2,1)$ and UE-$(1,2)$ as the corresponding receivers, respectively. 
Therefore,  the sum-rate bound in \eqref{eq:outer_bound_example_5} follows directly from the $2$-user IC genie-aided outer bound in
\cite[Th. 1]{Etkin2008}.
Next, we present key elements of the  genie-aided approach in \cite{Etkin2008},
which are central to the remainder of our proof. 

We start by defining the side information (or genie) signals:
\begin{align}
 S_{1}(t)  & = h_{21}^{[1]}X_{1}(t) + Z_{2}^{[1]}(t) \\
 S_{2}(t)  & = h_{12}^{[2]}X_{2}(t) + Z_{1}^{[2]}(t).
\end{align}
Proceeding from \eqref{eq:fano_case1_12} and \eqref{eq:fano_case1_2}, The above signals are employed as follows:
\begin{align}
n \big( R_{1}^{[1]} +  R_{1}^{[2]}  -  2\epsilon \big) & \leq  I\big(X_{1}^{n} ; Y_{1}^{[2]n}, S_{1}^{n} \big) + n \\
& = h\big(Y_{1}^{[2]n} | S_{1}^{n} \big)  + h\big( S_{1}^{n} \big)  -  h\big(Y_{1}^{[2]n} | X_{1}^{n}  \big) -  h\big(S_{1}^{n}  | Y_{1}^{[2]n} , X_{1}^{n}  \big) + n \\
\label{eq:Fano_genie_UEs_1}
& = h\big(Y_{1}^{[2]n} | S_{1}^{n} \big)  + h\big( S_{1}^{n} \big)  -  h\big(S_{2}^{n} \big) -  h\big(Z_{2}^{[1]n} \big) + n \\
n \big( R_{2}^{[1]} -  \epsilon \big) & \leq  I\big(X_{2}^{n} ; Y_{2}^{[1]n}, S_{2}^{n} \big) \\
& = h \big( Y_{2}^{[1] n} | S_{2}^{n}  \big)  + h \big( S_{2}^{n}  \big) -  h  \big( Y_{2}^{[1] n} |  X_{2}^n  \big)
- h  \big( S_{2}^{n}   |  Y_{2}^{[1] n}  , X_{2}^n  \big) \\
\label{eq:Fano_genie_UE21}
&  =  h \big( Y_{2}^{[1] n} | S_{2}^{n}  \big)  + h \big( S_{2}^{n}  \big) -  h  \big( S_{1}^{n}  \big)
- h  \big( Z_{1}^{[2] n} \  \big).
\end{align} 
By adding the bounds in \eqref{eq:Fano_genie_UEs_1} and \eqref{eq:Fano_genie_UE21}, we obtain
\begin{equation}
\label{eq:genie_bound_2_user}
  n \big( R_{1}^{[1]} +   R_{1}^{[2]}  +   R_{2}^{[1]} - 3 \epsilon \big)
  \leq  h \big(Y_{1}^{[2] n} | S_{1}^{n}  \big) - h(Z_{1}^{[2] n})
  + h \big(Y_{2}^{[1] n} | S_{2}^{n}  \big) - h(Z_{2}^{[1] n}) + n.
\end{equation}
Next, we bound the first difference of entropies on the right-hand-side of \eqref{eq:genie_bound_2_user} as follows:
\begin{align}
\label{eq:2_user_IC_bounds_11}
h \big(Y_{1}^{[2] n} | S_{1}^{n}  \big) - h(Z_{1}^{[2] n})  & \leq n \log \left( 1 + P^{\alpha_{12}^{[2]}} +
\frac{P^{\alpha_{11}^{[2]}}}{1+ P^{\alpha_{21}^{[1]}}} \right)  \\
\label{eq:2_user_IC_bounds_12}
& \leq n\big(  \alpha_{11}^{[2]} - \alpha_{21}^{[1]} \big) \log (P)  + n\log(3).
\end{align}
In the above, \eqref{eq:2_user_IC_bounds_11} follows by first applying a single-letterization step, which exploits the i.i.d.-ness of the noise,
and then using the fact that Gaussian inputs maximize conditional differential entropies under covariance constraints 
(see similar steps in \eqref{eq:diff_ha_hb_1st_0}--\eqref{eq:diff_ha_hb_1st_2} in Appendix \ref{appendix:subsection_proof_lemma_diff}).
On the other hand, \eqref{eq:2_user_IC_bounds_12} holds due to the TIN conditions \eqref{eq:TIN_conditions_example_1} and \eqref{eq:TIN_conditions_example_2},
which together imply $\alpha_{11}^{[2]} \geq \alpha_{12}^{[2]} + \alpha_{21}^{[1]}$.
In a similar manner, we also obtain the following bound:
\begin{align}
\label{eq:2_user_IC_bounds_22}
h \big(Y_{2}^{[1] n} | S_{2}^{n}  \big) - h(Z_{2}^{[1] n})  \leq  n\big(  \alpha_{22}^{[1]} - \alpha_{12}^{[2]} \big) \log (P) + n\log(3).
\end{align}
It is evident that \eqref{eq:genie_bound_2_user}, \eqref{eq:2_user_IC_bounds_12} 
and \eqref{eq:2_user_IC_bounds_22} yield the desired sum-rate bound in \eqref{eq:outer_bound_example_5}. 
\item $\alpha_{11}^{[2]} - 2\alpha_{12}^{[2]} \geq \alpha_{11}^{[1]} - \alpha_{12}^{[1]}$ and 
$\alpha_{11}^{[2]} - \alpha_{12}^{[2]} < \alpha_{11}^{[1]}$:
\label{case2:converse_example}
For this case, we must have 
\begin{equation}
\label{eq:int_order_case2}
\alpha_{12}^{[1]} > \alpha_{12}^{[2]}.
\end{equation}
This holds since by assuming the contrary, i.e. $\alpha_{12}^{[1]} \leq \alpha_{12}^{[2]}$,
we obtain
\begin{equation}
\alpha_{11}^{[2]} - 2\alpha_{12}^{[2]} \geq \alpha_{11}^{[1]} - \alpha_{12}^{[1]} \Longleftrightarrow
\alpha_{11}^{[2]} - \alpha_{12}^{[2]} \geq \alpha_{11}^{[1]} - \big( \alpha_{12}^{[1]} - \alpha_{12}^{[2]} \big)  
\implies 
\alpha_{11}^{[2]} - \alpha_{12}^{[2]} \geq \alpha_{11}^{[1]},
\end{equation} 
which cannot be true as we have assumed that $\alpha_{11}^{[2]} - \alpha_{12}^{[2]} < \alpha_{11}^{[1]}$.
Therefore, \eqref{eq:int_order_case2} must hold.

Next, we  define the following side information signals
\begin{align}
 S_{1}(t)  & = h_{21}^{[1]}X_{1}(t) + Z_{2}^{[1]}(t) \\
 S_{2}(t)  & = h_{12}^{[1]}X_{2}(t) + Z_{1}^{[1]}(t).
\end{align}
The signal $S_{1}^{n}$, which contains interference caused by cell 1 to cell 2, is given to the \emph{stronger} BC user in cell 1, i.e. UE-$(2,1)$. 
On the other hand, $S_{2}^{n}$, which contains interference caused by cell 2 to cell 1, is given to the (only) user in cell 2.
It is worthwhile noting that the signal $S_{2}^{n}$,  used to  enhance the user in cell 2, contains the interference caused by cell 2 
to the \emph{weaker} BC user in cell 1, as this user experiences greater interference in this case. This is in contrast to the case in \ref{case1:converse_example}, where $S_{2}^{n}$
contained interference cause to the \emph{stronger} BC user.

Next, we bound each of the individual rates as:
\begin{align}
  n \big( R_{1}^{[1]}  - \epsilon \big) & \leq  I \big(W_{1}^{[1]} ; Y_{1}^{[1] n}  | W_{1}^{[2]}   \big)
  \\
  & = h\big(Y_{1}^{[1] n}  | W_{1}^{[2]}   \big) - h \big(Y_{1}^{[1] n}  |  W_{1}^{[1]} , W_{1}^{[2]}   \big) \\
  \label{eq:Fano_UE11}
  & = h\big(Y_{1}^{[1] n}  | W_{1}^{[2]}   \big) - h \big(S_{2}^{n} \big)    \\
  n \big( R_{1}^{[2]}  - \epsilon \big) & \leq  I \big(W_{1}^{[2]} ; Y_{1}^{[2] n}  , S_{1}^{n}  \big)
  \\
  & = h \big(Y_{1}^{[2] n} | S_{1}^{n}  \big) + h \big( S_{1}^{n}  \big)   - h \big(Y_{1}^{[2] n}  |  W_{1}^{[2]} \big)
  - h \big(S_{1}^{n}  |  Y_{1}^{[2] n} , W_{1}^{[2]} \big)  \\
  \label{eq:Fano_UE21}
  & \leq h \big(Y_{1}^{[2] n} | S_{1}^{n}  \big) + h \big( S_{1}^{n}  \big)   - h \big(Y_{1}^{[2] n}  |  W_{1}^{[2]} \big)
  - h \big( Z_{2}^{[1] n} \big) \\
 n \big( R_{2}^{[1]}  - \epsilon \big) & \leq  I \big(W_{2}^{[1]} ; Y_{2}^{[1] n}  , S_{2}^{n}  \big)  \\
&  =  h \big( Y_{2}^{[1] n} | S_{2}^{n}  \big)  + h \big( S_{2}^{n}  \big) -  h  \big( Y_{2}^{[1] n} |  W_{2}^{[1]}   \big)
- h  \big( S_{2}^{n}   |  Y_{2}^{[1] n}  , W_{2}^{[1]}   \big) \\
 \label{eq:Fano_UE12}
&  =  h \big( Y_{2}^{[1] n} | S_{2}^{n}  \big)  + h \big( S_{2}^{n}  \big) -  h  \big( S_{1}^{n}  \big)
- h  \big( Z_{1}^{[1] n} \  \big).
\end{align}
By adding the bounds in \eqref{eq:Fano_UE11}, \eqref{eq:Fano_UE21} and \eqref{eq:Fano_UE12}, we obtain 
\begin{multline}
\label{eq:Fano_sum}
  n \big( R_{1}^{[1]} +   R_{1}^{[2]}  +   R_{2}^{[1]} - 3 \epsilon \big)
  \leq  h \big(Y_{1}^{[2] n} | S_{1}^{n}  \big) - h(Z_{1}^{[1] n})
  + h \big(Y_{2}^{[1] n} | S_{2}^{n}  \big) - h(Z_{2}^{[1] n}) \\
  + h\big(Y_{1}^{[1] n}  | W_{1}^{[2]}   \big)  - h \big(Y_{1}^{[2] n}  |  W_{1}^{[2]}
  \big).
\end{multline}
The first two differences of entropies in \eqref{eq:Fano_sum} are bounded as in
\eqref{eq:2_user_IC_bounds_12}, yielding
\begin{align}
h \big(Y_{1}^{[2] n} | S_{1}^{n}  \big) - h(Z_{1}^{[1] n})  
& \leq n\big(  \alpha_{11}^{[2]} - \alpha_{21}^{[1]} \big) \log (P)  + n\log(3) \\
\label{eq:diff_entropies_genie_2_cell_3_user}
h \big(Y_{2}^{[1] n} | S_{2}^{n}  \big) - h(Z_{2}^{[1] n})  
& \leq  n\big(  \alpha_{22}^{[1]} - \alpha_{12}^{[1]} \big) \log (P)
+ n\log(3).
\end{align}
The third difference of entropies is bounded as 
\begin{equation}
\label{eq:diff_entropies_2_cell_3_user}
h\big(Y_{1}^{[1] n}  | W_{1}^{[2]}   \big)  - h \big(Y_{1}^{[2] n}  |  W_{1}^{[2]}\big)
\leq n\big(  \alpha_{12}^{[1]} - \alpha_{12}^{[2]} \big) \log (P)
+ n \log(6)
\end{equation}
which follows directly from Lemma \ref{lemma:diff_entropies}.
Combining all bounds, we obtain \eqref{eq:outer_bound_example_5}.
\end{enumerate}
Some insights gained from the $2$-cell, $3$-user network are summarized in the following remark.
\begin{remark}
\label{remark:2_cell_converse}
For  \ref{case1:converse_example}, 
the stronger BC user UE-$(2,1)$ is less noisy than the weaker BC user UE-$(1,1)$ under inter-cell interference.
This allows us, with the help of Lemma \ref{lemma:more_capable}, to eliminate UE-$(1,1)$ from the picture and bound  the sum-rate of cell $1$
by the rate of the less noisy receiver UE-$(2,1)$, plus a constant.
The $2$-cell network hence reduces to a $2$-user IC, for which the genie-aided bound in \cite{Etkin2008} applies.  
For \ref{case2:converse_example}, UE-$(2,1)$ is not less noisy than UE-$(1,1)$.
Nevertheless, we observe that in this case, the interference level seen by UE-$(1,1)$ is higher than that 
seen by UE-$(2,1)$ (see \eqref{eq:int_order_case2}).
Hence, as previously elaborated in Section \ref{subsubsec:TIN_insights}, this gives UE-$(1,1)$ the opportunity to achieve a GDoF gain of at most 
$(\alpha_{12}^{[1]} - \alpha_{12}^{[2]})$, shown through Lemma \ref{lemma:diff_entropies}.
This gain, however, is offset by designing the genie signal for cell $2$, i.e. $S_{2}^{n}$, such that it contains the  interference seen by UE-$(1,1)$, 
that is the dominant interference caused to cell $1$ (see \eqref{eq:diff_entropies_genie_2_cell_3_user} and \eqref{eq:diff_entropies_2_cell_3_user}). 
\hfill $\lozenge$
\end{remark}
\subsection{Proof of Outer Bound: General Case}
\begin{figure}
\centering
\includegraphics[width = 0.7\textwidth,trim={0cm 0cm 0cm 0cm},clip]{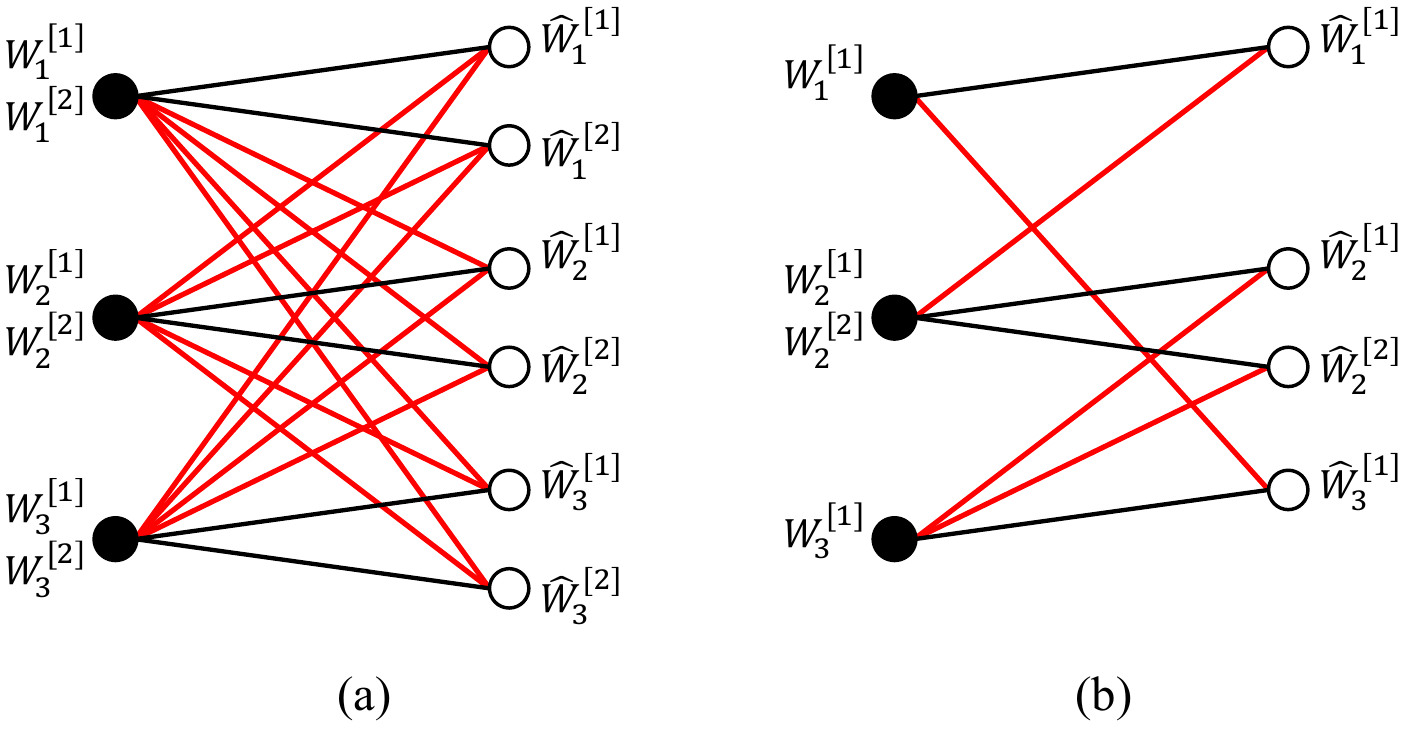}
\caption{\small
(a) 3-cell interfering broadcast channel with 2 users in each cell, and (b) 
the corresponding partially connected cyclic network associated with the sequences $(i_{1},i_{2},i_{3}) = (1,3,2)$ and $(l_{i_{1}},l_{i_{2}},l_{i_{3}}) = (1,1,2) $.}
\label{fig:IBC_cyclic}
\end{figure}
As argued in the previous part, the single-cell bounds in \eqref{eq:capacity_outer_1} follow from standard results.
Therefore, we focus on the multi-cell cyclic bounds in \eqref{eq:capacity_outer_2}.
It is evident that each such bound is identified by two sequences: 
$(i_{1},i_{2},\ldots,i_{m})$, which specifies participating cells and their cyclic order, and 
$(l_{i_{1}},l_{i_{2}},\ldots,l_{i_{m}})$, which specifies the number (and identities) of participating users in each participating cell.
In what follows, we fix such sequences, hence focusing on an arbitrary cyclic bound from  \eqref{eq:capacity_outer_2}.
It is also useful to recall that a modulo-$m$ operation is implicitly used on cell indices such that $i_{0} = i_{m}$ and $i_{m+1} = i_{1}$.
Next, we go through the following steps:
\begin{itemize}
\item Eliminate non-participating cells, non-participating users and their corresponding messages. 
\item For the remaining network, eliminate all interfering links except for links from BS-$i_{j-1}$ to 
UE-$(s_{i_{j}},i_{j})$, for all $j \in \langle m \rangle$ and $s_{i_{j}} \in \langle l_{i_{j}} \rangle$.
\end{itemize}
Applying the above steps yields a partially connected cyclic network, see Fig. \ref{fig:IBC_cyclic} for example.  
This new network is described by the following input-output relationship:
\begin{equation}
\label{eq:cyclic_IBC_system_model}
Y_{i_{j}}^{[s_{i_{j}}]} (t) = h_{i_{j} i_{j}}^{[s_{i_{j}}]} X_{i_{j}} (t) + h_{i_{j} i_{j - 1}}^{[s_{i_{j}}]} X_{i_{j-1}} (t)  + Z_{i_{j}}^{[s_{i_{j}}]}  (t).
\end{equation}
Since the above steps cannot hurt the rates of participating users, we restricted or attention to 
the channel in \eqref{eq:cyclic_IBC_system_model} for the purpose of deriving the corresponding cyclic outer bound.

Let us now focus on a single participating cell $i_{j}$. As there is no ambiguity, we refer to users by their first index only, e.g. UE-$(s_{i_{j}},i_{j})$
is referred to as user $s_{i_j}$.
We partition the set of participating users, i.e. $\langle l_{i_{j}} \rangle  $, into two subsets as follows: 
\begin{itemize}
\item  $\mathcal{L}_{i_{j}}'$: this consists of the strongest BC user $l_{i_{j}}$ and 
users which are more noisy than $l_{i_{j}}$, i.e.
\begin{equation}
\mathcal{L}_{i_{j}}'  \triangleq \left\{ l_{i_{j}}  \right\} \cup \left\{ s_{i_{j}} \in \langle l_{i_{j}} - 1 \rangle :  
 \alpha_{i_{j} i_{j}}^{[l_{i_j} ]} - \alpha_{i_{j} i_{j-1}}^{[l_{i_j} ]}  \geq   \alpha_{i_{j} i_{j}}^{[s_{i_j}]}   \right\}.
\end{equation} 
We use $l_{i_{j}}'$ to denote the cardinality of this subset,  i.e. $l_{i_{j}}' = | \mathcal{L}_{i_{j}}' |$.
Moreover, we label the user indices as $\big\{p_{i_{j}}(1),\ldots, p_{i_{j}}(l_{i_{j}}') \big\} \triangleq \mathcal{L}_{i_{j}}'$,
such that $p_{i_{j}}(1) < p_{i_{j}}(2) < \cdots < p_{i_{j}}(l_{i_{j}}') = l_{i_{j}}$.
\item $\mathcal{L}_{i_{j}}'' $: this consists of users which are not more noisy than $l_{i_{j}}$, i.e.
\begin{equation}
\label{eq:not_more_noisy_subset}
\mathcal{L}_{i_{j}}''  \triangleq \left\{ s_{i_{j}} \in \langle l_{i_{j}} - 1 \rangle :  
 \alpha_{i_{j} i_{j}}^{[l_{i_j} ]} - \alpha_{i_{j} i_{j-1}}^{[l_{i_j} ]} <  \alpha_{i_{j} i_{j}}^{[s_{i_j}]}   \right\}.
\end{equation} 
We use $l_{i_{j}}''$ to denote the cardinality of this subset,  i.e. $l_{i_{j}}'' = |\mathcal{L}_{i_{j}}''|$.
Moreover, we label the user indices as $ \big\{q_{i_{j}}(1),\ldots, q_{i_{j}}(l_{i_{j}}'') \big\} \triangleq \mathcal{L}_{i_{j}}''$,
such that $q_{i_{j}}(1) < q_{i_{j}}(2) < \cdots < q_{i_{j}}(l_{i_{j}}'')$.
\end{itemize}
Next, we highlight an intrinsic order which exists amongst users in the second subset $\mathcal{L}_{i_{j}}''$.
\begin{lemma}
\label{lemma:not_more_noisy_subset}
For all $s \in \{ 1, \ldots, l_{i_{j}}'' \}$, the following holds:
\begin{align}
\label{eq:condition_not_more_noisy_1}
& \alpha_{i_{j} i_{j}}^{[q_{i_j}(s+1)]}  - \alpha_{i_{j} i_{j-1}}^{[q_{i_j}(s+1)]}   < \alpha_{i_{j} i_{j}}^{[q_{i_j}(s)]} \ \ \text{and} \\
\label{eq:condition_not_more_noisy_2}
& \alpha_{i_{j} i_{j}}^{[q_{i_j}(s+1)]}  - 2 \alpha_{i_{j} i_{j-1}}^{[q_{i_j}(s+1)]}   \geq \alpha_{i_{j} i_{j}}^{[q_{i_j}(s)]}  -  \alpha_{i_{j} i_{j-1}}^{[q_{i_j}(s)]}.
\end{align}
In the above, we take $q_{i_j}(l_{i_j}'' + 1)$ to be equal to $l_{i_j}$.
\end{lemma}
\begin{proof}
First, we observe that for all $s \in \{ 1, \ldots, l_{i_{j}}'' \}$, we have
\begin{align}
\label{eq:condition_not_more_noisy_1_2}
& \alpha_{i_{j} i_{j}}^{[l_{i_j}]}  - \alpha_{i_{j} i_{j-1}}^{[l_{i_j}]}   < \alpha_{i_{j} i_{j}}^{[q_{i_j}(s)]} \ \ \text{and} \\
\label{eq:condition_not_more_noisy_2_2}
& \alpha_{i_{j} i_{j}}^{[l_{i_j}]}  - 2 \alpha_{i_{j} i_{j-1}}^{[l_{i_j}]}   \geq \alpha_{i_{j} i_{j}}^{[q_{i_j}(s)]}  -  \alpha_{i_{j} i_{j-1}}^{[q_{i_j}(s)]},
\end{align}
where \eqref{eq:condition_not_more_noisy_1_2} holds by definition of the subset $\mathcal{L}_{i_{j}}'' $ in  \eqref{eq:not_more_noisy_subset},
while \eqref{eq:condition_not_more_noisy_2_2} holds due to the TIN condition in \eqref{eq:TIN_cond_1}.
Therefore, it follows that  \eqref{eq:condition_not_more_noisy_1} and  \eqref{eq:condition_not_more_noisy_2}  must hold for $s = l_{i_{j}}''$.
Next, suppose that \eqref{eq:condition_not_more_noisy_1} does not hold for some $s \in \langle l_{i_{j}}''  - 1\rangle$.
Then we must have the following:
\begin{align}
\alpha_{i_{j} i_{j}}^{[q_{i_j}(s)]} & \leq  \alpha_{i_{j} i_{j}}^{[q_{i_j}(s+1)]}  - \alpha_{i_{j} i_{j-1}}^{[q_{i_j}(s+1)]}   \\
& \leq 
\alpha_{i_{j} i_{j}}^{[l_{i_{j}}]}  - 2 \alpha_{i_{j} i_{j-1}}^{[l_{i_{j}}]}  \\
& \leq 
\alpha_{i_{j} i_{j}}^{[l_{i_{j}}]}  -  \alpha_{i_{j} i_{j-1}}^{[l_{i_{j}}]}  \\
 \implies &  q_{i_j}(s) \in \mathcal{L}_{i_j}'.
\end{align}
The above yields a contradiction, since $q_{i_j}(s) \in \mathcal{L}_{i_j}''$.
Therefore, \eqref{eq:condition_not_more_noisy_1}  must hold for all  $s \in \langle l_{i_{j}}''  - 1\rangle$.
Due to the TIN condition in \eqref{eq:TIN_cond_1}, it follows that \eqref{eq:condition_not_more_noisy_2_2} 
must also hold for all  $s \in \langle l_{i_{j}}''  - 1\rangle$.  
\end{proof}
In words, Lemma \ref{lemma:not_more_noisy_subset} states that any user $q_{i_j}(s+1)$ in $\mathcal{L}_{i_j}''$ is not less noisy than its preceding user $q_{i_j}(s)$,
as assuming the contrary implies that user $l_{i_j}$ is less noisy than user $q_{i_j}(s)$, which  contradicts the definition of $\mathcal{L}_{i_j}''$.
From Lemma \ref{lemma:not_more_noisy_subset}, it follows that for all $q_{i_j}(s)  \in  \mathcal{L}_{i_j}''$, we have
\begin{equation}
\label{eq:int_order_general_case}
\alpha_{i_{j}i_{j-1}}^{[q_{i_j}(s)]} > \alpha_{i_{j}i_{j-1}}^{[q_{i_j}(s+1)]}.
\end{equation}
The above order of interference levels generalizes  \eqref{eq:int_order_case2}, and can be shown using similar steps.

Next, we define the side information signal for cell $i_{j}$ as:
\begin{equation}
\label{eq:side_information_general}
S_{i_{j}}(t)  = h_{i_{j+1}i_j} X_{i_{j}} (t) + Z_{i_{j+1}} (t),
\end{equation}
where $Z_{i_{j+1}}^{n}$ is an i.i.d. AWGN sequence independent of all other signals, while $h_{i_{j+1}i_j} $ is defined as:
\begin{equation}
h_{i_{j+1}i_j} \triangleq
\begin{cases}
 h_{i_{j+1}i_j}^{[q_{i_{j+1}}(1)]}, &  l_{i_{j+1}}'' > 0 \\ 
  h_{i_{j+1}i_j}^{[l_{i_{j+1}}]},  & l_{i_{j+1}}'' = 0.
 \end{cases}
\end{equation}
We use $\alpha_{i_{j+1}i_j} $ to denote the channel strength level associated with $h_{i_{j+1}i_j} $, i.e. 
\begin{equation}
\alpha_{i_{j+1}i_j} \triangleq
\begin{cases}
 \alpha_{i_{j+1}i_j}^{[q_{i_{j+1}}(1)]}, &  l_{i_{j+1}}'' > 0 \\ 
  \alpha_{i_{j+1}i_j}^{[l_{i_{j+1}}]},  & l_{i_{j+1}}'' = 0.
 \end{cases}
\end{equation}
The side information signal $S_{i_{j}}^{n}$ in \eqref{eq:side_information_general}  is constructed by taking guidance from the points raised in Remark \ref{remark:2_cell_converse}.
In particular, when $\mathcal{L}_{i_{j+1}}''$ is empty, user $l_{i_{j+1}}$ of cell  $i_{j+1}$ is less noisy than all other users in the same cell, which in turn, may be eliminated 
from the picture. Therefore, $S_{i_{j}}^{n}$ contains interference caused to user $l_{i_{j+1}}$ of cell  $i_{j+1}$ in this case.
On the other hand, when $\mathcal{L}_{i_{j+1}}''$ is non-empty, user $q_{i_{j+1}}(1)$  of cell $i_{j+1}$ sees the strongest interference signal from cell $i_{j}$  
(see \eqref{eq:int_order_general_case}). Therefore, $S_{i_{j}}^{n}$ contains the dominant interference component, 
caused to user $q_{i_{j+1}}(1)$  of cell $i_{j+1}$  in this case. 

Next, we invoke the fact that user $l_{i_j}$ is less noisy than users in $\mathcal{L}_{i_{j}}'$
and Lemma \ref{lemma:more_capable}  to bound the sum-rate of users in $\mathcal{L}_{i_{j}}'$ through a single mutual information term as follows:
\begin{align}
n\sum_{s_{i_j} \in\mathcal{L}_{i_{j}}' } \big( R_{i_{j}}^{[s_{i_j}]}  - \epsilon \big) & \leq 
\sum_{s_{i_j} \in \mathcal{L}_{i_{j}}'}  I \big(W_{i_{j}}^{[s_{i_j}]} ; Y_{i_{j}}^{[s_{i_{j}}] n}   \big) \\
& = I \big(W_{i_{j}}^{[l_{i_j}]} ; Y_{i_{j}}^{[l_{i_j}] n}   \big)  +  \sum_{s \in \langle l_{i_{j}}' - 1 \rangle  }  I \big(W_{i_{j}}^{[p_{i_j}(s)]} ; Y_{i_{j}}^{[p_{i_j}(s)] n}   \big) \\
& \leq I \big(W_{i_{j}}^{[l_{i_j}]} ; Y_{i_{j}}^{[l_{i_j}] n}   \big)  +  \sum_{s \in \langle l_{i_{j}}' - 1 \rangle  }  I \big(W_{i_{j}}^{[p_{i_j}(s)]} ; Y_{i_{j}}^{[l_{i_j}] n}    \big)  \\
& \leq I \big(W_{i_{j}}^{[l_{i_j}]} ; Y_{i_{j}}^{[l_{i_j}] n}   \big)  +  \sum_{s \in \langle l_{i_{j}}' - 1 \rangle  }  I \big(W_{i_{j}}^{[p_{i_j}(s)]} ; Y_{i_{j}}^{[l_{i_j}] n}  |  W_{i_{j}}^{\{ p_{i_j}(s+1),\ldots,
p_{i_j}(l_{i_j}') \}}  \big)  \\
\label{eq:fano_general_1st_subset}
& \leq I \big(W_{i_{j}}^{\mathcal{L}_{i_{j}}'} ; Y_{i_{j}}^{[l_{i_j}] n}   \big).
\end{align}
Note that in the above, we have used $W_{i_{j}}^{\mathcal{L}}$ to denote the set of messages $\{W_{i_{j}}^{[s]} : s \in \mathcal{L} \}$, for 
some $\mathcal{L} \subseteq \langle L_{i_{j}} \rangle$.
Proceeding from \eqref{eq:fano_general_1st_subset}, we further enhance 
the users in $\mathcal{L}_{i_{j}}'$ by providing the side information signal $S_{i_{j}}^{n}$.
This leads to the following
\begin{align}
n\sum_{s_{i_j} \in \mathcal{L}_{i_{j}}' } \big( R_{i_{j}}^{[s_{i_j}]}  - \epsilon \big) & \leq   I \big(W_{i_{j}}^{\mathcal{L}_{i_{j}}'} ; Y_{i_{j}}^{[l_{i_j}] n} , S_{i_{j}}^{n}  \big) \\
& =   h\big(Y_{i_{j}}^{[l_{i_j}] n} | S_{i_{j}}^{n}  \big) + h\big( S_{i_{j}}^{n}  \big)  - h\big(Y_{i_{j}}^{[l_{i_j}] n} | W_{i_{j}}^{\mathcal{L}_{i_{j}}'}  \big)    - 
h\big(S_{i_{j}}^{n}  | Y_{i_{j}}^{[l_{i_j}] n} , W_{i_{j}}^{\mathcal{L}_{i_{j}}'}  \big)  \\
\label{eq:fano_general_sum_bound_1st_subset}
& \leq   h\big(Y_{i_{j}}^{[l_{i_j}] n} | S_{i_{j}}^{n}  \big) + h\big( S_{i_{j}}^{n}  \big)  - h\big(Y_{i_{j}}^{[l_{i_j}] n} | W_{i_{j}}^{\mathcal{L}_{i_{j}}'}  \big)    - 
h\big(Z_{i_{j+1}}^{n}   \big).
\end{align}
Now we turn to the second  subset of users $\mathcal{L}_{i_{j}}''$.
We use $W_{i_{j}}''^{(s)}$ to briefly denote the set of messages 
$\big\{ W_{i_{j}}^{\{ q_{i_j}(s), \ldots,q_{i_j}(l_{i_j}'') \}}, 
W_{i_{j}}^{\mathcal{L}_{i_{j}}'} \big\}$, for all $s \in \langle l_{i_j}'' \rangle$.
The sum-rate of users in $\mathcal{L}_{i_{j}}''$ is bounded  as
\begin{align}
n\sum_{s_{i_j} \in \mathcal{L}_{i_{j}}''} \big( R_{i_{j}}^{[s_{i_j}]}  - \epsilon \big) & \leq 
\sum_{s_{i_j} \in \mathcal{L}_{i_{j}}'' }  I \big(W_{i_{j}}^{[s_{i_j}]} ; Y_{i_{j}}^{[s_{i_{j}}] n}   \big) \\
& = \sum_{s \in \langle l_{i_{j}}'' \rangle  }  I \big(W_{i_{j}}^{[q_{i_j}(s)]} ; Y_{i_{j}}^{[q_{i_{j}}(s)] n}   \big)  \\
& \leq \sum_{s \in \langle l_{i_{j}}'' \rangle  }  I \big(W_{i_{j}}^{[q_{i_j}(s)]} ; Y_{i_{j}}^{[q_{i_{j}}(s)] n} | W_{i_{j}}''^{(s+1)}   \big)  \\
\label{eq:fano_general_sum_bound_2nd_subset}
& =  \sum_{s \in \langle l_{i_{j}}'' \rangle  }  \bigg[
h \big(Y_{i_{j}}^{[q_{i_{j}}(s)] n} 
| W_{i_{j}}''^{(s+1)}    \big) - 
h \big(Y_{i_{j}}^{[q_{i_{j}}(s)] n} | W_{i_{j}}''^{(s)}    \big) \bigg].
\end{align}
Note that in the above, we take $W_{i_{j}}''^{(l_{i_j}'' + 1)} $ to be  $W_{i_{j}}^{\mathcal{L}_{i_{j}}'} $.
Adding the bounds in \eqref{eq:fano_general_sum_bound_1st_subset} and \eqref{eq:fano_general_sum_bound_2nd_subset},
and after rearranging some terms, we obtain the following  bound on the sum-rate of users in cell $i_{j}$:
\begin{align}
\nonumber
n\sum_{s_{i_j} \in \langle l_{i_j} \rangle} \big( R_{i_{j}}^{[s_{i_j}]}  - \epsilon \big)  &  \leq 
h\big(Y_{i_{j}}^{[l_{i_j}] n} | S_{i_{j}}^{n}  \big)  - h\big(Z_{i_{j+1}}^{n}   \big) 
+ h\big( S_{i_{j}}^{n}  \big)  -  h\big( S_{i_{j-1}}^{n}  \big)   \\
\label{eq:fano_general_sum_bound_all_subsets}
& + \sum_{s \in \langle l_{i_{j}}'' \rangle  }    \bigg[  h \big(Y_{i_{j}}^{[q_{i_{j}}(s)] n} | W_{i_{j}}''^{(s+1)}   \big) - 
h \big(Y_{i_{j}}^{[q_{i_{j}}(s+1)] n} | W_{i_{j}}''^{(s+1)}    \big) \bigg].
\end{align}
Now we elaborate on the manner in which \eqref{eq:fano_general_sum_bound_all_subsets} is obtained.
First, consider the case where $\mathcal{L}_{i_{j}}''$  is empty. 
Here \eqref{eq:fano_general_sum_bound_all_subsets}  is obtained from \eqref{eq:fano_general_sum_bound_1st_subset} on its own, 
i.e. the term $h\big(Y_{i_{j}}^{[l_{i_j}] n} | W_{i_{j}}^{\mathcal{L}_{i_{j}}'}  \big) $ in 
\eqref{eq:fano_general_sum_bound_1st_subset} is equal to $h\big(S_{i_{j-1}}^{n}  \big) $ in this case, while the summation on the right-hand-side 
of \eqref{eq:fano_general_sum_bound_all_subsets}  is not present as it is taken over an empty set.
On the other hand, when $\mathcal{L}_{i_{j}}''$  is non-empty, 
the term $h \big(Y_{i_{j}}^{[q_{i_{j}}(1)] n} | W_{i_{j}}''^{(1)}    \big) $ in 
\eqref{eq:fano_general_sum_bound_2nd_subset} it equal to $h\big(S_{i_{j-1}}^{n}  \big) $ in this case,
while the term $h\big(Y_{i_{j}}^{[l_{i_j}] n} | W_{i_{j}}^{\mathcal{L}_{i_{j}}'}  \big) $ from \eqref{eq:fano_general_sum_bound_1st_subset}
now appears in the summation on the right-hand-side 
of \eqref{eq:fano_general_sum_bound_all_subsets} by taking $q_{i_{j}}(l_{i_{j}}'' + 1) = l_{i_j}$.

Having obtained a bound on the sum-rate of each cell $i_{j}$ in \eqref{eq:fano_general_sum_bound_all_subsets}, 
we construct a multi-cell cyclic bound by adding the sum-rate bounds for all cells in the sequence $(i_{1},\ldots,i_{m})$ as:
\begin{align}
\nonumber
& n\sum_{j \in \langle m \rangle} \sum_{s_{i_j} \in \langle l_{i_j} \rangle} \big( R_{i_{j}}^{[s_{i_j}]}  - \epsilon \big)     \\
\label{eq:multi_cell_bound_converse_1}
& \leq 
\sum_{j \in \langle m \rangle} \bigg[   h\big(Y_{i_{j}}^{[l_{i_j}] n} | S_{i_{j}}^{n}  \big)  - h\big(Z_{i_{j}}^{n}   \big)  
 + \sum_{s \in \langle l_{i_j}'' \rangle }    \Big[  h \big(Y_{i_{j}}^{[q_{i_{j}}(s)] n} | W_{i_{j}}''^{(s+1)}   \big) - 
h \big(Y_{i_{j}}^{[q_{i_{j}}(s+1)] n} | W_{i_{j}}''^{(s+1)}  \big) \Big]  \bigg] \\
\label{eq:multi_cell_bound_converse_2}
& \leq n \sum_{j \in \langle m \rangle}  \bigg[   \big(  \alpha_{i_j i_j}^{[l_{i_j}]} - \alpha_{i_{j+1} i_j} \big) 
+  \sum_{s \in \langle l_{i_j}'' \rangle }      \big(  \alpha_{i_j i_{j-1}}^{[ q_{i_{j}}(s) ]} - \alpha_{i_{j} i_{j-1}}^{[q_{i_{j}}(s+1)]} \big)
\bigg] \log(P) 
+ nO(1) \\
\label{eq:multi_cell_bound_converse_3}
& = n \sum_{j \in \langle m \rangle}  \bigg[   \big(  \alpha_{i_j i_j}^{[l_{i_j}]} - \alpha_{i_{j+1} i_j} \big) 
+  \sum_{s \in \langle l_{i_{j+1}}'' \rangle }      \big(  \alpha_{i_{j+1} i_{j}}^{[ q_{i_{j+1}}(s) ]} - \alpha_{i_{j+1} i_{j}}^{[q_{i_{j+1}}(s+1)]} \big)
\bigg]    \log(P) 
+ nO(1) \\
\label{eq:multi_cell_bound_converse_4}
& = n  \sum_{j \in \langle m \rangle}  \big(  \alpha_{i_j i_j}^{[l_{i_j}]} - \alpha_{i_{j+1} i_j}^{[l_{i_{j+1}}]}  \big) \log(P) 
+ nO(1) \\
\label{eq:multi_cell_bound_converse_5}
& = n  \sum_{j \in \langle m \rangle}  \big(  \alpha_{i_j i_j}^{[l_{i_j}]} - \alpha_{i_{j} i_{j-1}}^{[l_{i_{j}}]}  \big) \log(P)
+ nO(1).
\end{align}
In \eqref{eq:multi_cell_bound_converse_1}, each $h\big(Y_{i_{j}}^{[l_{i_j}] n} | S_{i_{j}}^{n}  \big)  - h\big(Z_{i_{j}}^{n}   \big)$
is bounded in a similar manner to \eqref{eq:2_user_IC_bounds_12}.
On the other hand, each $h \big(Y_{i_{j}}^{[q_{i_{j}}(s)] n} | W_{i_{j}}''^{(s)}   \big) - 
h \big(Y_{i_{j}}^{[q_{i_{j}}(s+1)] n} | W_{i_{j}}''^{(s)}  \big)$
is bounded using Lemma \ref{lemma:diff_entropies}, which applies in this case 
due to the observations in Lemma \ref{lemma:not_more_noisy_subset} and \eqref{eq:int_order_general_case}.
By combining these bounds, we obtain \eqref{eq:multi_cell_bound_converse_2}.
Proceeding from \eqref{eq:multi_cell_bound_converse_2}, \eqref{eq:multi_cell_bound_converse_3} is obtained by applying a cyclic shift to 
the cell indices in the inner summation.
On the other hand, 
 \eqref{eq:multi_cell_bound_converse_4} is obtained from \eqref{eq:multi_cell_bound_converse_3} by considering two cases for each $j$: 
the case  $l_{i_{j+1}}'' = 0$ for which we have we have $\alpha_{i_{j+1} i_j} = \alpha_{i_{j+1} i_j}^{[l_{i_{j+1}}]} $,
and the case  $l_{i_{j+1}}'' > 0 $  for which we have $\alpha_{i_{j+1} i_j} = \alpha_{i_{j+1} i_j}^{[q_{i_{j+1}} (1)]} $
and  $ \sum_{s \in \langle l_{i_{j+1}}'' \rangle }   \big(  \alpha_{i_{j+1} i_{j}}^{[ q_{i_{j+1}}(s) ]} - \alpha_{i_{j+1} i_{j}}^{[q_{i_{j+1}}(s+1)]} \big) = 
\alpha_{i_{j+1} i_j}^{[q_{i_{j+1}} (1)]}  - \alpha_{i_{j+1} i_j}^{[l_{i_{j+1}} ]}   $.
Both cases evidently lead to the same result in \eqref{eq:multi_cell_bound_converse_4} for each $j$.
Finally, the desired sum-rate bound in  \eqref{eq:multi_cell_bound_converse_5} is obtained by applying a cyclic shift, in the opposite direction this time, to the cell indices of interference strength levels.  
This concludes the proof of Theorem \ref{theorem:outer_bound} and this section. 
\section{Conclusion}
\label{sec:conclusion}
In this paper, we established a GDoF-based uplink-downlink duality of multi-cell TIN.
On the achievability side, we showed that when restricting to a single-cell transmission strategy in each cell, with power control and treating inter-cell interference as noise,  
the corresponding achievable GDoF regions (TINA regions) for the IBC and its IMAC are identical.
On the converse side, we showed that the TINA region for the IBC is optimal in the TIN regime, identified for the IMAC in \cite{Joudeh2019a}.
Therefore, multi-cell TIN is optimal for both the IBC and IMAC in the  TIN regime of \cite{Joudeh2019a}.
In deriving the outer bound, we established a new notion of redundancy order amongst users in the same cell.
We showed that in the GDoF sense, the identified TIN conditions preserve the redundancy order of users in each cell of the IBC, while other known orders due to degradedness and less noisiness  are not preserved in general under inter-cell interference.

Theoretical GDoF-based TIN results for the $K$-user IC in \cite{Geng2015} have inspired a number efficient practical power allocation and link scheduling algorithms for D2D networks \cite{Geng2016,Yi2016,Geng2018,Naderializadeh2014}.
An interesting direction for future research would be to leverage the theoretical results in this work  and \cite{Joudeh2019a} to design new scheduling and power control  algorithms for cellular networks.
Another interesting direction is to investigate the optimality of multi-cell TIN, for both the IBC and IMAC, under finite precision CSIT. 
In this case, IA gains achieved in the CTIN regime will most likely collapse, as suggested by bounds based on aligned images (AI) \cite{Davoodi2016,Davoodi2017}. 
Hence, we envisage that multi-cell TIN will be optimal for both the IMAC and IBC in the entire CTIN regime under finite precision CSIT, which is analogous to a recent counterpart result for the regular IC in \cite{Chan2019}.
\appendix
\section{Proof of Lemma \ref{lemma:power_allocation_duality}}
\label{appebdix:proof_of_duality_lemma}
\subsection{Proof of $\mathcal{D}_{\mathrm{TINA}}^{\mathrm{IBC}}(\bm{\pi},\mathbf{r}) \subseteq \mathcal{D}_{\mathrm{TINA}}^{\mathrm{IMAC}}(\bm{\pi},\bar{\mathbf{r}})$}
Consider an arbitrary GDoF tuple
$\mathbf{d} \in \mathcal{D}^{\mathrm{IBC}}_{\mathrm{TINA}}(\bm{\pi},\mathbf{r})$.
The components of $\mathbf{d}$ must satisfy
\begin{equation}
\label{eq:IBC_GDoF per user_gamma}
d_{k}^{[\pi_{k}(l_{k})]}  \leq
\left(  \alpha_{kk}^{[\pi_{k}(l_{k})]} + r_{k}^{[\pi_{k}(l_{k})]}   - \gamma_{k}^{[\pi_{k}(l_{k})]} \right)^{+},
\end{equation}
where $\gamma_{k}^{[\pi_{k}(l_{k})]} $ is defined in \eqref{eq:IBC_gamma}.
Now consider the power allocation $\bar{\mathbf{r}}$ for the dual IMAC,
where  $ \bar{r}_{k}^{[\pi_{k}(l_{k})]} = -\gamma_{k}^{[\pi_{k}(l_{k})]}$ for every $(l_{k},k) \in \mathcal{K}$.
Using  $(\bm{\pi},\bar{\mathbf{r}})$,
we achieve the set of GDoF tuples given by
$\mathcal{D}^{\mathrm{IMAC}}_{\mathrm{TINA}}(\bm{\pi},\bar{\mathbf{r}})$ over the dual IMAC,
where for every $\bar{\mathbf{d}}$ in such set, each component $\bar{d}_{k}^{[\pi_{k}(l_{k})]}$ satisfies
\begin{equation}
\label{eq:duality_IMAC_GDoF_r}
\bar{d}_{k}^{[\pi_{k}(l_{k})]}  \leq
\left(  \alpha_{kk}^{[\pi_{k}(l_{k})]} - \gamma_{k}^{[\pi_{k}(l_{k})]}
  -  \bar{\gamma}_{k}^{[\pi_{k}(l_{k})]} \right)^{+}.
\end{equation}
By comparing \eqref{eq:duality_IMAC_GDoF_r} and \eqref{eq:IBC_GDoF per user_gamma},
it is evident that $\mathcal{D}^{\mathrm{IBC}}_{\mathrm{TINA}}(\bm{\pi},\mathbf{r}) \subseteq
\mathcal{D}^{\mathrm{IMAC}}_{\mathrm{TINA}}(\bm{\pi},\bar{\mathbf{r}})$
holds if the inequality
\begin{equation}
\label{eq:duality_IMAC_GDoF_UB_1_0}
r_{k}^{[\pi_{k}(l_{k})]} \leq   -  \bar{\gamma}_{k}^{[\pi_{k}(l_{k})]}
\end{equation}
holds for all $(l_{k},k) \in \mathcal{K}$.
Therefore, we focus on showing that this is the case in what follows.

To this end, we start by equivalently expressing the inequality in \eqref{eq:duality_IMAC_GDoF_UB_1_0}  as
\begin{equation}
\label{eq:duality_IMAC_GDoF_UB_1}
r_{k}^{[\pi_{k}(l_{k})]} \leq
\min \bigl\{0, \min_{l_{k}':l_{k}' < l_{k}} \{ \gamma_{k}^{[\pi_{k}(l_{k}')]} - \alpha_{kk}^{[\pi_{k}(l_{k}')]} \}  ,\min_{(l_{j},j):j \neq k}
 \{ \gamma_{j}^{[l_{j}]} - \alpha_{jk}^{[l_{j}]} \}   \bigr\}.
\end{equation}
As $r_{k}^{[\pi_{k}(l_{k})]}  \leq 0$, we only need to show that the inequality in \eqref{eq:duality_IMAC_GDoF_UB_1}
holds for the two remaining terms inside the $\min\{0,\cdot,\cdot\}$.
We start by showing that $r_{k}^{[\pi_{k}(l_{k})]}  \leq \gamma_{k}^{[\pi_{k}(l_{k}')]} - \alpha_{kk}^{[\pi_{k}(l_{k}')]}$,
for all $l_{k}' < l_{k}$, i.e.
\begin{align}
\nonumber
\gamma_{k}^{[\pi_{k}(l_{k}')]} - \alpha_{kk}^{[\pi_{k}(l_{k}')]} & =
\max
\Bigl\{ \max_{l_{k}'':l_{k}''>l_{k}'} \{ \! r_{k}^{[\pi_{k}(l_{k}'')]} \! \} , \max_{m_{k}:m_{k} \geq l_{k}' } \bigl\{
 \bigl( \max_{(l_{j},j):j \neq k} \{ \alpha_{kj}^{[\pi_{k}(m_{k})]}
 +   r_{j}^{[l_{j}]} \} \bigr)^{+} - \alpha_{kk}^{[\pi_{k}(m_{k})]}  \bigr\}    \Bigr\} \\
 & \geq \max_{l_{k}'':l_{k}''>l_{k}'} \{  r_{k}^{[\pi_{k}(l_{k}'')]} \} \\
\label{eq:duality_IMAC_GDoF_UB_2}
\implies & \gamma_{k}^{[\pi_{k}(l_{k}')]} - \alpha_{kk}^{[\pi_{k}(l_{k}')]} \geq r_{k}^{[\pi_{k}(l_{k})]}, \ \forall l_{k} > l_{k}'.
\end{align}
Next, we show that
$r_{k}^{[\pi_{k}(l_{k})]}  \leq \gamma_{j}^{[l_{j}]}  -  \alpha_{jk}^{[l_{j}]}$, for all
$(l_{j},j) \in \mathcal{K}$ with $j \neq k$.
As
$l_{j} = \pi_{j}(l_{j}')$, for some $l_{j}' \in \langle L_{j} \rangle$, we may write
\begin{align}
\label{eq:duality_IMAC_GDoF_UB_3_1}
\gamma_{j}^{[l_{j}]}  -  \alpha_{jk}^{[l_{j}]}
&\geq  \alpha_{jj}^{[l_{j}]} +
\max_{m_{j}:m_{j} \geq l_{j}' } \bigl\{
  \bigl( \max_{(l_{i},i):i \neq j} \{ \alpha_{ji}^{[\pi_{j}(m_{j})]}
 +   r_{i}^{[l_{i}]} \} \bigr)^{+} - \alpha_{jj}^{[\pi_{j}(m_{j})]}  \bigr\}  - \alpha_{jk}^{[l_{j}]}   \\
& \geq  \alpha_{jj}^{[l_{j}]} +
 \bigl( \max_{(l_{i},i):i \neq j} \{ \alpha_{ji}^{[\pi_{j}(l_{j}')]}
 +   r_{i}^{[l_{i}]} \} \bigr)^{+}  - \alpha_{jj}^{[\pi_{j}(l_{j}')]}  - \alpha_{jk}^{[l_{j}]}   \\
\label{eq:duality_IMAC_GDoF_UB_4}
&=   \max_{(l_{i},i):i \neq j} \{ \alpha_{ji}^{[l_{j}]} + r_{i}^{[l_{i}]}   \} - \alpha_{jk}^{[l_{j}]}  \\
\label{eq:duality_IMAC_GDoF_UB_5}
& \geq r_{k}^{[\pi_{k}(l_{k})]}.
\end{align}
In \eqref{eq:duality_IMAC_GDoF_UB_3_1}, we bound $\gamma_{j}^{[\pi_{j}(l_{j}')]}$ by taking the second term in its outmost
$\max\{\cdot,\cdot\}$ (see \eqref{eq:IBC_gamma}).

From \eqref{eq:duality_IMAC_GDoF_UB_2} and \eqref{eq:duality_IMAC_GDoF_UB_5}, we conclude that the inequality in \eqref{eq:duality_IMAC_GDoF_UB_1} holds for all $(l_{k},k) \in \mathcal{K}$, and therefore $\mathcal{D}^{\mathrm{IBC}}_{\mathrm{TINA}}(\bm{\pi},\mathbf{r}) \subseteq \mathcal{D}^{\mathrm{IMAC}}_{\mathrm{TINA}}(\bm{\pi},\bar{\mathbf{r}})$. This completes this part of the proof.
\subsection{Proof of $\mathcal{D}_{\mathrm{TINA}}^{\mathrm{IMAC}}(\bm{\pi},\bar{\mathbf{r}}) \subseteq
\mathcal{D}_{\mathrm{TINA}}^{\mathrm{IBC}}(\bm{\pi},\mathbf{r}) $}
To facilitate this part of the proof, we start by imposing a simplifying restriction.
In particular, we modify the IMAC TIN scheme in Section \ref{subsec:TIN_IMAC}
by restricting the power control policy such that
\begin{equation}
\label{eq:duality_condition_IMAC_order}
\alpha_{kk}^{[\pi_{k}(l_{k}'')]} + \bar{r}_{k}^{[\pi_{k}(l_{k}'')]} \geq   \alpha_{kk}^{[\pi_{k}(l_{k}')]} + \bar{r}_{k}^{[\pi_{k}(l_{k}')]} , \
\forall l_{k}'' > l_{k}', \ k \in \langle K \rangle.
\end{equation}
While restricting the TIN scheme should, by definition, lead to a possibly smaller TINA region compared to
$\mathcal{D}_{\mathrm{TINA}}^{\mathrm{IMAC}}$,
the restriction in \eqref{eq:duality_condition_IMAC_order} turns out to  be harmless.
Intuitively, for any $l_{k}'' > l_{k}'$, achieving a non-zero GDoF $d_{k}^{[\pi_{k}(l_{k}'')]} > 0$
naturally requires the signal of UE-$\big(\pi_{k}(l_{k}''),k\big)$ to be received
at a higher power level compared to the signal of  UE-$\big(\pi_{k}(l_{k}'),k\big)$,
specifically as the former preceded the latter in the succussive decoding order.
A formal proof showing that \eqref{eq:duality_condition_IMAC_order} has no influence is given at the end of this appendix.
We proceed while assuming, without loss of generality, that the conditions in \eqref{eq:duality_condition_IMAC_order}
holds for all considered IMAC power allocations.

Now let us consider an arbitrary GDoF tuple $\bar{\mathbf{d}} \in \mathcal{D}^{\mathrm{IMAC}}_{\mathrm{TINA}}(\bm{\pi},\bar{\mathbf{r}})$
for some feasible $(\bm{\pi},\bar{\mathbf{r}})$.
Recall that the components of $\bar{\mathbf{d}}$ must satisfy
\begin{equation}
\label{eq:IMAC_GDoF per user_gamma}
\bar{d}_{k}^{[\pi_{k}(l_{k})]}  \leq
\left(  \alpha_{kk}^{[\pi_{k}(l_{k})]} + \bar{r}_{k}^{[\pi_{k}(l_{k})]}
  -  \bar{\gamma}_{k}^{[\pi_{k}(l_{k})]} \right)^{+}
\end{equation}
where $\bar{\gamma}_{k}^{[\pi_{k}(l_{k})]}$ is defined in \eqref{eq:IMAC_gamma}.
Now consider the IBC power allocation $\mathbf{r}$, where $r_{k}^{[\pi_{k}(l_{k})]} =  - \bar{\gamma}_{k}^{[\pi_{k}(l_{k})]}$
for every $(l_{k},k) \in \mathcal{K}$.
Using $(\bm{\pi},\mathbf{r})$, we achieve the set of GDoF tuples given by $\mathcal{D}^{\mathrm{IBC}}_{\mathrm{TINA}}(\bm{\pi},\mathbf{r})$,
over the IBC,
where for every $\mathbf{d}$ is this set, each component $d_{k}^{[\pi_{k}(l_{k})]}$ must satisfy
\begin{equation}
\label{eq:duality_IBC_GDoF_beta}
d_{k}^{[\pi_{k}(l_{k})]}  \leq
\left(  \alpha_{kk}^{[\pi_{k}(l_{k})]} - \bar{\gamma}_{k}^{[\pi_{k}(l_{k})]}   - \gamma_{k}^{[\pi_{k}(l_{k})]} \right)^{+}.
\end{equation}
By examining  \eqref{eq:IMAC_GDoF per user_gamma} and \eqref{eq:duality_IBC_GDoF_beta}, it is readily seen that
$\mathcal{D}^{\mathrm{IMAC}}(\bm{\pi},\bar{\mathbf{r}}) \subseteq \mathcal{D}^{\mathrm{IBC}}(\bm{\pi},\mathbf{r})$
holds if for all $(l_{k},k) \in \mathcal{K}$, the following inequality holds:
\begin{multline}
\label{eq:duality_IBC_GDoF_UB_2}
 \bar{r}_{k}^{[\pi_{k}(l_{k})]} \leq - \gamma_{k}^{[\pi_{k}(l_{k})]} \ \implies \\
 \alpha_{kk}^{[\pi_{k}(l_{k})]} \! + \!  \bar{r}_{k}^{[\pi_{k}(l_{k})]} \! \leq \!
\min
\Bigl\{  \min_{l_{k}'':l_{k}''>l_{k}} \{  \bar{\gamma}_{k}^{[\pi_{k}(l_{k}'')]}  \} , \! \! \min_{m_{k}:m_{k} \geq l_{k} } \! \! \bigl\{
 \alpha_{kk}^{[\pi_{k}(m_{k})]}- \bigl( \max_{(l_{j},j):j \neq k} \{ \alpha_{kj}^{[\pi_{k}(m_{k})]}
 -   \bar{\gamma}_{j}^{[l_{j}]} \} \bigr)^{+}  \bigr\}  \!  \Bigr\}
\end{multline}
Therefore, we focus on proving \eqref{eq:duality_IBC_GDoF_UB_2} throughout the remainder of this part.

We start by showing that $\alpha_{kk}^{[\pi_{k}(l_{k})]} + \bar{r}_{k}^{[\pi_{k}(l_{k})]} \leq \bar{\gamma}_{k}^{[\pi_{k}(l_{k}'')]}$,
for all $l_{k}''>l_{k}$. In particular, we have
\begin{align}
\label{eq:duality_IBC_GDoF_UB_2_0}
\bar{\gamma}_{k}^{[\pi_{k}(l_{k}'')]} & = \max \bigl\{0, \max_{l_{k}':l_{k}' < l_{k}'' } \{ \alpha_{kk}^{[\pi_{k}(l_{k}')]} + \bar{r}_{k}^{[\pi_{k}(l_{k}')]} \}  ,\max_{(l_{j},j):j \neq k}
 \{ \alpha_{jk}^{[l_{j}]} + \bar{r}_{j}^{[l_{j}]} \}   \bigr\} \\
 & \geq \max_{l_{k}':l_{k}' < l_{k}'' } \{ \alpha_{kk}^{[\pi_{k}(l_{k}')]} + \bar{r}_{k}^{[\pi_{k}(l_{k}')]} \} \\
 & \geq \alpha_{kk}^{[\pi_{k}(l_{k})]} + \bar{r}_{k}^{[\pi_{k}(l_{k})]}
\end{align}
where \eqref{eq:duality_IBC_GDoF_UB_2_0} follows from the definition of $\bar{\gamma}_{k}^{[\pi_{k}(l_{k}'')]} $ in
\eqref{eq:IMAC_gamma}.
Next, it remains to show that
\begin{equation}
\label{eq:duality_IBC_GDoF_UB_3}
\alpha_{kk}^{[\pi_{k}(l_{k})]} + \bar{r}_{k}^{[\pi_{k}(l_{k})]}  \leq
 \alpha_{kk}^{[\pi_{k}(m_{k})]}- \bigl( \max_{(l_{j},j):j \neq k} \{ \alpha_{kj}^{[\pi_{k}(m_{k})]}
 -   \bar{\gamma}_{j}^{[l_{j}]} \} \bigr)^{+}
\end{equation}
holds for all $m_{k} \geq l_{k}$.
To this end, we observe that we may express the right-hand-side of \eqref{eq:duality_IBC_GDoF_UB_3} as
\begin{multline}
\label{eq:duality_IBC_GDoF_UB_4}
\alpha_{kk}^{[\pi_{k}(m_{k})]}- \bigl( \max_{(l_{j},j):j \neq k} \{ \alpha_{kj}^{[\pi_{k}(m_{k})]}
 -   \bar{\gamma}_{j}^{[l_{j}]} \} \bigr)^{+}
 =  \alpha_{kk}^{[\pi_{k}(m_{k})]}- \max \bigl\{ 0, \max_{(l_{j},j):j \neq k} \{ \alpha_{kj}^{[\pi_{k}(m_{k})]}
 -   \bar{\gamma}_{j}^{[l_{j}]} \} \bigr\} \\
 = \min \bigl\{ \alpha_{kk}^{[\pi_{k}(m_{k})]} , \alpha_{kk}^{[\pi_{k}(m_{k})]} + \min_{(l_{j},j):j \neq k} \{ \bar{\gamma}_{j}^{[l_{j}]} -
\alpha_{kj}^{[\pi_{k}(m_{k})]} \} \bigr\}.
\end{multline}
Next, we invoke the assumption in \eqref{eq:duality_condition_IMAC_order}, from which
 $\alpha_{kk}^{[\pi_{k}(l_{k})]} + \bar{r}_{k}^{[\pi_{k}(l_{k})]}  \leq \alpha_{kk}^{[\pi_{k}(m_{k})]}$ holds
 for all $m_{k} \geq l_{k}$.
Hence, proving \eqref{eq:duality_IBC_GDoF_UB_3} reduces to showing that
$\alpha_{kk}^{[\pi_{k}(l_{k})]} + \bar{r}_{k}^{[\pi_{k}(l_{k})]} \leq \alpha_{kk}^{[\pi_{k}(m_{k})]} + \bar{\gamma}_{j}^{[l_{j}]} -
\alpha_{kj}^{[\pi_{k}(m_{k})]}$ holds for all $(l_{j},j) \in \mathcal{K}$ with $j \neq k$, as seen from \eqref{eq:duality_IBC_GDoF_UB_4}.
For this purpose, we write
\begin{align}
\label{eq:duality_IBC_GDoF_UB_5}
\alpha_{kk}^{[\pi_{k}(m_{k})]} + \bar{\gamma}_{j}^{[l_{j}]}  -  \alpha_{kj}^{[\pi_{k}(m_{k})]}   &  \geq
\alpha_{kk}^{[\pi_{k}(m_{k})]} + \max_{(l_{i},i): i \neq j} \{\alpha_{ij}^{[l_{i}]} + \bar{r}_{i}^{[l_{i}]} \} - \alpha_{kj}^{[\pi_{k}(m_{k})]}  \\
\label{eq:duality_IBC_GDoF_UB_6}
& \geq \alpha_{kk}^{[\pi_{k}(m_{k})]}  + \bar{r}_{k}^{[\pi_{k}(m_{k})]}   \\
\label{eq:duality_IBC_GDoF_UB_7}
& \geq \alpha_{kk}^{[\pi_{k}(l_{k})]}  + \bar{r}_{k}^{[\pi_{k}(l_{k})]}.
\end{align}
In \eqref{eq:duality_IBC_GDoF_UB_5}, we bound $\bar{\gamma}_{j}^{[l_{j}]}$ below by taking the third term in its outmost $\max\{0,\cdot,\cdot\}$
(see its definition in \eqref{eq:IMAC_gamma}).
On the other hand,
the inequality in \eqref{eq:duality_IBC_GDoF_UB_6} is obtained by setting
$(l_{i},i) = (\pi_{k}(m_{k}),k)$ in  \eqref{eq:duality_IBC_GDoF_UB_5},
while the inequality in \eqref{eq:duality_IBC_GDoF_UB_7} holds due to the assumption in \eqref{eq:duality_condition_IMAC_order}.

As \eqref{eq:duality_IBC_GDoF_UB_2} holds for all $(l_{k},k) \in \mathcal{K}$,
we conclude that $\mathcal{D}^{\mathrm{IMAC}}(\bm{\pi},\bar{\mathbf{r}}) \subseteq \mathcal{D}^{\mathrm{IBC}}(\bm{\pi},\mathbf{r}) $,
which completes this part of the proof. To complete the proof,
we now justify the assumption in \eqref{eq:duality_condition_IMAC_order}.
\subsection{Justification for \eqref{eq:duality_condition_IMAC_order}}
To show that the restriction in \eqref{eq:duality_condition_IMAC_order} is harmless,
consider a feasible strategy $(\bm{\pi},\bar{\mathbf{r}})$, and suppose that the contrary of \eqref{eq:duality_condition_IMAC_order}
holds for a pair of UEs in cell $k$. That is, we have
\begin{equation}
\label{eq:GDoF_duality_MAC_just_0}
\alpha_{kk}^{[\pi_{k}(l_{k}'')]} + \bar{r}_{k}^{[\pi_{k}(l_{k}'')]}  <   \alpha_{kk}^{[\pi_{k}(l_{k}')]} + \bar{r}_{k}^{[\pi_{k}(l_{k}')]}
\end{equation}
for some $l_{k}'' > l_{k}'$.
Denoting the set of GDoF tuples achieved through this fixed strategy by $\mathcal{D}_{\mathrm{TINA}}^{\mathrm{IMAC}}(\bm{\pi},\bar{\mathbf{r}})$,
we observe that for any $\bar{\mathbf{d}} \in \mathcal{D}_{\mathrm{TINA}}^{\mathrm{IMAC}}(\bm{\pi},\bar{\mathbf{r}})$,
the component $\bar{d}_{k}^{[\pi_{k}(l_{k}'')]}$ must satisfy
\begin{equation}
\label{eq:GDoF_duality_MAC_just}
\bar{d}_{k}^{[\pi_{k}(l_{k}'')]}   \leq \left( \alpha_{kk}^{[\pi_{k}(l_{k}'')]} + \bar{r}_{k}^{[\pi_{k}(l_{k}'')]} -
\big(\alpha_{kk}^{[\pi_{k}(l_{k}')]} + \bar{r}_{k}^{[\pi_{k}(l_{k}')]} \big) \right)^{+}
= 0
\end{equation}
where the inequality \eqref{eq:GDoF_duality_MAC_just} follows directly from \eqref{eq:IMAC_GDoF per user} in Section \ref{subsec:TIN_IMAC},
while the equality in \eqref{eq:GDoF_duality_MAC_just} holds due to \eqref{eq:GDoF_duality_MAC_just_0}.
Therefore, for any $\bar{\mathbf{d}} \in \mathcal{D}_{\mathrm{TINA}}^{\mathrm{IMAC}}(\bm{\pi},\bar{\mathbf{r}})$,
we must have $\bar{d}_{k}^{[\pi_{k}(l_{k}'')]} = 0$ whenever \eqref{eq:GDoF_duality_MAC_just_0} holds.

Now consider an alternative strategy $(\tilde{\bm{\pi}},\tilde{\mathbf{r}})$, which is a modification of $(\bm{\pi},\bar{\mathbf{r}})$
such that:
\begin{itemize}
\item  The decoding order of UE-$\big(\pi_{k}(l_{k}''),k\big)$ and UE-$\big(\pi_{k}(l_{k}'),k\big)$ from the original strategy
$(\bm{\pi},\bar{\mathbf{r}})$ is swapped in the modified strategy, while maintaining the decoding orders
of all other UEs. That is, we set
$\tilde{\pi}_{k}(l_{k}'') = \pi_{k}(l_{k}')$ and $\tilde{\pi}_{k}(l_{k}') = \pi_{k}(l_{k}'')$.
\item We set $\tilde{r}_{k}^{[\tilde{\pi}_{k}(l_{k}')]} = -\infty$,
while maintaining the power allocation for all remaining UEs.
\end{itemize}
Next, we show that any GDoF tuple achieved through the original strategy is also achievable through the modified strategy, i.e.
$\mathcal{D}_{\mathrm{TINA}}^{\mathrm{IMAC}}(\bm{\pi},\bar{\mathbf{r}}) \subseteq
\mathcal{D}_{\mathrm{TINA}}^{\mathrm{IMAC}}(\tilde{\bm{\pi}},\tilde{\mathbf{r}})$.

To this end, we first observe that for any GDoF tuple $\tilde{\mathbf{d}} \in
\mathcal{D}_{\mathrm{TINA}}^{\mathrm{IMAC}}(\tilde{\bm{\pi}},\tilde{\mathbf{r}})$,
the components corresponding to UEs in cell $k$ satisfy $\tilde{d}_{k}^{[\tilde{\pi}_{k}(l_{k}')]} = 0$
and
\begin{equation}
\label{eq:GDoF_duality_MAC_just_2}
\tilde{d}_{k}^{[\tilde{\pi}_{k}(l_{k})]}   \leq
\biggl( \!  \alpha_{kk}^{[\tilde{\pi}_{k}(l_{k})]}  +  \tilde{r}_{k}^{[\tilde{\pi}_{k}(l_{k})]}
  -   \max \bigl\{ \! 0, \max_{l_{k}^{\star}:l_{k}^{\star} < l_{k},l_{k}^{\star} \neq l_{k}'} \{
  \alpha_{kk}^{[\tilde{\pi}_{k}(l_{k}^{\star})]} + \tilde{r}_{k}^{[\tilde{\pi}_{k}(l_{k}^{\star})]} \}  ,\max_{(l_{j},j):j \neq k}
 \{ \alpha_{jk}^{[l_{j}]} +  \tilde{r}_{j}^{[l_{j}]} \}  \! \bigr\}   \biggr)^{+},
\end{equation}
where the latter holds for all $l_{k} \in \langle L_{k} \rangle \setminus \{l_{k}' \}$.
Recalling that $\pi_{k}(l_{k}'') = \tilde{\pi}_{k}(l_{k}')$, it is evident that the zero GDoF
achieved by UE-$\big(\tilde{\pi}_{k}(l_{k}'),k\big)$ is unchanged across the two strategies.
For all remaining UEs in cell $k$, by comparing \eqref{eq:IMAC_GDoF per user} and \eqref{eq:GDoF_duality_MAC_just_2},
it can be seen that GDoF components achieved using the original strategy $(\bm{\pi},\bar{\mathbf{r}})$
are also achievable under the modified strategy  $(\tilde{\bm{\pi}},\tilde{\mathbf{r}})$,
as such UEs experience the same inter-cell interference and less intra-cell interference under the latter strategy.
Furthermore, by extending this reasoning to UEs in cells indexed by
$i$, for all $i \in \langle K \rangle \setminus \{k\} $, we see that the corresponding GDoF components
achieved using $(\bm{\pi},\bar{\mathbf{r}})$ are also achievable using $(\tilde{\bm{\pi}},\tilde{\mathbf{r}})$,
as power allocations in cells
$i \in \langle K \rangle \setminus \{k\}$ are unaltered,
while the transmit power of cell $k$ is reduced in the modified strategy $(\tilde{\bm{\pi}},\tilde{\mathbf{r}})$.
Therefore, we have
$\mathcal{D}_{\mathrm{TINA}}^{\mathrm{IMAC}}(\bm{\pi},\bar{\mathbf{r}}) \subseteq
\mathcal{D}_{\mathrm{TINA}}^{\mathrm{IMAC}}(\tilde{\bm{\pi}},\tilde{\mathbf{r}})$.

The above argument is applied, recursively, to all pairs of UEs in all cells that violate the conditions in \eqref{eq:duality_condition_IMAC_order}.
Therefore, we end up with a strategy that satisfies the order in
\eqref{eq:duality_condition_IMAC_order} and achieve a set of GDoF tuples that contains
$\mathcal{D}_{\mathrm{TINA}}^{\mathrm{IMAC}}(\bm{\pi},\bar{\mathbf{r}})$.
This completes the proof of Lemma \ref{lemma:power_allocation_duality}.
\section{Interference Alignment in the CTIN Regime}
\label{appendix:IA}
Here we consider the $2$-cell, $3$-user network in Fig. \ref{fig:signal_levels_TIN}(a).
We show that structured codes and IA achieve GDoF gains over TIN in the sub-regime where the CTIN conditions hold but the TIN conditions do not. 
In what follows, we assume that the IC-type conditions in \eqref{eq:CTIN_cond_2_cell_3_user_2} and \eqref{eq:CTIN_cond_2_cell_3_user_3} hold.
Moreover, to ensure that we are strictly in the CTIN regime and not in the TIN regime, we assume\footnote{The case where 
$\beta_{1} - \beta_{2} =  \alpha_{1} - \alpha_{2}$ is discussed at the end of this appendix.}
\begin{align}
\label{eq:CTIN_cond_1_2_cell_3_user_example}
\beta_{1} - \beta_{2} & < \alpha_{1}  \  \text{and} \ \beta_{1} - 2\beta_{2} < \alpha_{1} - \alpha_{2} \\
\label{eq:CTIN_cond_2_2_cell_3_user_example}
\beta_{1} - \beta_{2} & >  \alpha_{1} - \alpha_{2}.
\end{align}
For ease of exposition, we further focus on the case where the interference level seen by the weaker BC user is no less than the interference level seen by the stronger BC user, 
i.e.  $\alpha_{2} \geq \beta_{2}$. 
Similar arguments can be constructed for the other case, where $\alpha_{2} < \beta_{2}$.

Using TIN and power control in an altruistic fashion, transmit powers are adjusted such that user $b$ and user $c$ receive no interference above noise levels. 
Note, however, since we have assumed $\alpha_{2} \geq \beta_{2}$, 
the lowest $\alpha_{2} - \beta_{2}$ levels at user $a$ are occupied by interference from transmitter $1$ (see Fig. \ref{fig:signal_levels_IA}(a) and (c)).
This strategy achieves the sum-GDoF in \eqref{eq:sum_GDoF_users_a_b_c_TIN}, which we rewrite here as 
\begin{equation}
\label{eq:sum_GDoF_TIN_2_cell_3_user}
d_{\mathrm{TINA}}^{\mathrm{IBC}} = (\beta_{1} - \gamma_{1}) + (\gamma_{2} - \beta_{2}).
\end{equation}
As we are in the CTIN regime, it follows from Theorem \ref{theorem:TIN_IBC}  that \eqref{eq:sum_GDoF_TIN_2_cell_3_user} is the maximum sum-GDoF
achievable using TIN. 
Next, we define the following quantity
\begin{equation}
\gamma_{\mathrm{IA}} \triangleq \min \big\{ (\alpha_{1} - \alpha_{2}) - (\beta_{1} - 2\beta_{2}),  (\beta_{1} - \beta_{2}) - (\alpha_{1} - \alpha_{2})  \big\}.
\end{equation}
From \eqref{eq:CTIN_cond_1_2_cell_3_user_example} and \eqref{eq:CTIN_cond_2_2_cell_3_user_example}, we know that $\gamma_{\mathrm{IA}}  > 0$.
In what follows, we show that under the above-described conditions, IA yields 
a strict GDoF gain of $\gamma_{\mathrm{IA}}$ over TIN, achieving a sum-GDoF of
\begin{equation}
\label{eq:sum_GDoF_IA_2_cell_3_user}
d_{\mathrm{IA}}^{\mathrm{IBC}} = d_{\mathrm{TINA}}^{\mathrm{IBC}}  + \gamma_{\mathrm{IA}}. 
\end{equation}
In showing the achievability of \eqref{eq:sum_GDoF_IA_2_cell_3_user}, we restrict ourselves to an intuitive 
exposition using  the notion of signal levels. 
Realizing such signal levels in the Gaussian setting of interest is achieved using multilevel lattice  codes (see, for example, \cite{Bresler2010,Jafar2010}).
Next, we treat each of the two cases that determine the quantity $\gamma_{\mathrm{IA}} $ separately.
We point to the illustrative examples of these two cases given in Fig. \ref{fig:signal_levels_IA}(b) and (d), which help in visualizing the following arguments. 
\begin{figure}[h]
\centering
\includegraphics[width = 1.0\textwidth]{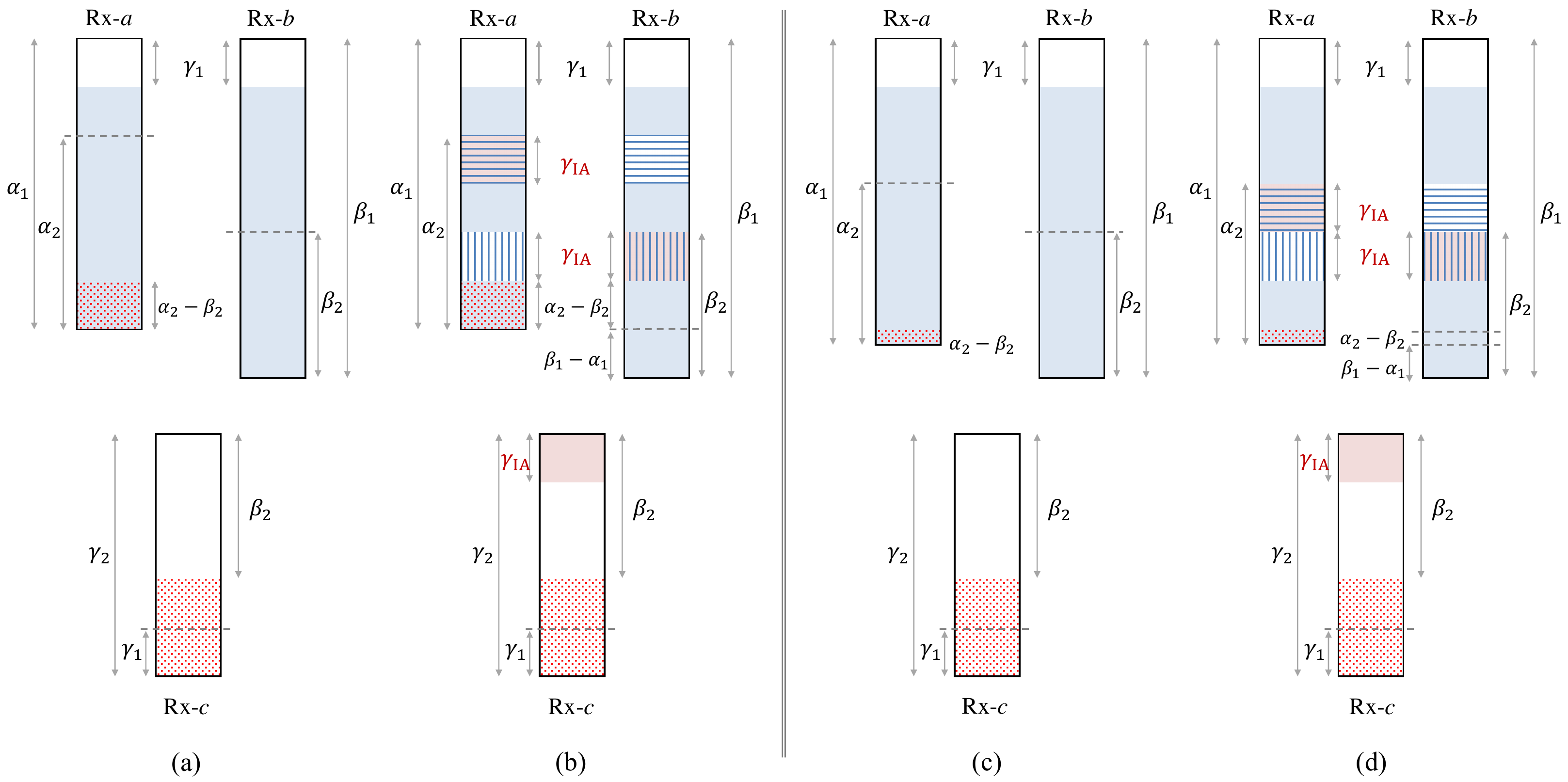}
\caption{\small Signal power levels for two instances of the 2-cell, 3-user network in the CTIN regime. Levels in white, blue and red are empty, $X_{1}$ and $X_{2}$, respectively.  Left: $(\alpha_{1} - \alpha_{2}) - (\beta_{1} - 2\beta_{2}) \leq (\beta_{1} - \beta_{2}) - (\alpha_{1} - \alpha_{2})$, (a) TIN, and (b) IA. Right:  $(\alpha_{1} - \alpha_{2}) - (\beta_{1} - 2\beta_{2}) > (\beta_{1} - \beta_{2}) - (\alpha_{1} - \alpha_{2})$, (c) TIN, and (d) IA.}
\label{fig:signal_levels_IA}
\end{figure}
\begin{enumerate}
\item $(\alpha_{1} - \alpha_{2}) - (\beta_{1} - 2\beta_{2}) \leq (\beta_{1} - \beta_{2}) - (\alpha_{1} - \alpha_{2}) $:
We start from the altruistic TIN scheme in Fig. \ref{fig:signal_levels_IA}(a).
It is useful to examine the lower $\beta_{2}$ signal levels received by user $b$.
These signal levels can be partitioned into three parts: 1) the lowest $(\beta_{1} - \alpha_{1})$ levels, which are received by user $b$ but not by user $a$,
2) the middle $(\alpha_{2} - \beta_{2})$ levels, adjacent to signal levels of user $a$ which are corrupted by interference from the TIN scheme,
and 3) the upper $\beta_{2} - (\beta_{1} - \alpha_{1}) - (\alpha_{2} - \beta_{2}) = \gamma_{\mathrm{IA}} $ levels.\footnote{It is worthwhile noting that since $(\beta_{1} - \alpha_{1})  \geq 0$, $(\alpha_{2} - \beta_{2}) \geq 0$ and $\gamma_{\mathrm{IA}} > 0$, such partition of the lowest $\beta_{2}$ levels of user $b$ exists, with the third part consisting of a strictly positive number of signal levels.} An illustration is shown in Fig. \ref{fig:signal_levels_IA}(b). 

We proceed by assuming that in addition to the lower $(\gamma_{2} - \beta_{2})$ signal levels employed by transmitter $2$, this transmitter further uses its upper 
$\gamma_{\mathrm{IA}} = (\alpha_{1} - \alpha_{2}) - (\beta_{1} - 2\beta_{2})$ levels to transmit an additional signal, denoted by $U_{2}$, to user $c$.
By doing so, user $c$ now achieves a GDoF of $(\gamma_{2} - \beta_{2}) + \gamma_{\mathrm{IA}}$.
Nevertheless, the $\gamma_{\mathrm{IA}}$  signal levels of user $b$ in the third part of the above partition are now corrupted by interference from   $U_{2}$.
We denote such $\gamma_{\mathrm{IA}}$ signal levels of $X_{1}$ by the signal $U_{1}$. 
It follows that user $b$  achieves a GDoF of $(\beta_{1} - \gamma_{1})  - \gamma_{\mathrm{IA}}$.
Next, we show that user $a$ can compensate for this loss of GDoF by decoding $U_{1}$. 
In particular, we show that while $U_{1}$ and $U_{2}$ align at user $b$, they are received  separately at   
user $a$.

To this end, we look at the lower $\alpha_{2}$ signal levels of user $a$, which can be partitioned into three parts: 1) the lowest  $(\alpha_{2} - \beta_{2})$ levels,
which are corrupted by interference from the TIN scheme, 2) the middle $\gamma_{\mathrm{IA}}$ levels, containing $U_{1}$
which user $a$ wishes to decode, and 3) the upper $\alpha_{2} - (\alpha_{2} - \beta_{2}) - \gamma_{\mathrm{IA}} = (\beta_{1} - \beta_{2}) - (\alpha_{1} - \alpha_{2})$ levels. 
It can be verified that such partition exists.
Moreover, since we are considering the case where $(\beta_{1} - \beta_{2}) - (\alpha_{1} - \alpha_{2}) \geq \gamma_{\mathrm{IA}}$,
it follows that $U_{2}$ is received entirely in the third part of the above partition, 
and hence does not overlap with the  $\gamma_{\mathrm{IA}}$ signal levels of $U_{1}$, received in the middle part of the partition.
Therefore, user $a$ can decode $U_{1}$ and achieve a GDoF of $\gamma_{\mathrm{IA}}$.
By adding the $3$ individual GDoF contributions, we achieve the sum-GDoF of $d_{\mathrm{IA}}^{\mathrm{IBC}}$ in \eqref{eq:sum_GDoF_IA_2_cell_3_user}.
\item $(\alpha_{1} - \alpha_{2}) - (\beta_{1} - 2\beta_{2}) >  (\beta_{1} - \beta_{2}) - (\alpha_{1} - \alpha_{2}) $:
Similar to the previous case, we start from the altruistic TIN scheme in Fig. \ref{fig:signal_levels_IA}(c), and then proceed to
assume that transmitter $2$ uses its upper $\gamma_{\mathrm{IA}} =  (\beta_{1} - \beta_{2}) - (\alpha_{1} - \alpha_{2})$ signal levels  to transmit 
an additional signal $U_{2}$ to user $c$.
Therefore, user $c$ achieves a GDoF of $(\gamma_{2} - \beta_{2}) + \gamma_{\mathrm{IA}}$. 

We now partition the lower $\beta_{2}$ signal levels of user $b$ into: 1) the lowest $(\beta_{1} - \alpha_{1})$ levels, which are not received by user $a$,  
2) the middle $(\alpha_{2} - \beta_{2})$ levels,
adjacent to the lower levels at user $a$ which are corrupted by interference from the TIN scheme,  and 3) 
the upper $(\alpha_{1} - \alpha_{2}) - (\beta_{1} - 2\beta_{2})$ levels.
We are interested in the upper $\gamma_{\mathrm{IA}} $ levels of the third part of this partition, 
in which user $b$ receives a part of $X_{1}$, denoted by $U_{1}$, corrupted by interference from $U_{2}$.
Due to this interference, user $b$  achieves a GDoF of $(\beta_{1} - \gamma_{1})  - \gamma_{\mathrm{IA}}$.

Considering user $a$, the received signal levels are partitioned into: 1) the lowest $(\alpha_{2} - \beta_{2})$ levels, which are corrupted by interference from the TIN scheme, 
2) the middle $(\alpha_{1} - \alpha_{2}) - (\beta_{1} - 2\beta_{2})$ levels, of which the upper $\gamma_{\mathrm{IA}} $ levels contain $U_{1}$,
and 3) the upper $(\beta_{1} - \beta_{2}) - (\alpha_{1} - \alpha_{2})$,
which contain interference from $U_{2}$.
It can be verified that the above partition exists.
Moreover, it can be seen that $U_{1}$ and $U_{2}$ are received by user $a$ through non-overlapping signal levels. 
Therefore, user $a$  achieves a GDoF of $ \gamma_{\mathrm{IA}}$ by decoding $U_{1}$.
It follows that the sum-GDoF of $d_{\mathrm{IA}}^{\mathrm{IBC}}$ in \eqref{eq:sum_GDoF_IA_2_cell_3_user} is achieved in this case as well.
\end{enumerate}
The key step in the above discussion is identifying those signal levels of $X_{1}$ and $X_{2}$ which align at user $b$, yet are received separately (with no overlap) at user $a$.
Signal levels with such property open the door for IA, which surpasses TIN in the CTIN regime as seen above. 
\begin{remark}
For the case where $\beta_{1} - \beta_{2} = \alpha_{1} - \alpha_{2}$,  we have $\gamma_{\mathrm{IA}} = 0$, and hence  $d_{\mathrm{IA}}^{\mathrm{IBC}} $
collapses to $d_{\mathrm{TINA}}^{\mathrm{IBC}}$.
In this case, the above-described scheme, with IA over signal power levels, achieves no GDoF gain over TIN.
A similar situation arises for the dual uplink 2-cell, 3-user network  in \cite{Gherekhloo2016}, where it was also observed that  IA over signal power levels fails to surpass 
TIN wherever $\beta_{1} - \beta_{2} = \alpha_{1} - \alpha_{2}$.
Alternatively, it was shown that a scheme which employs phase alignment, instead of signal level alignment,
achieves a strict GDoF gain over TIN when  $\beta_{1} - \beta_{2} = \alpha_{1} - \alpha_{2}$, except for a set of channel coefficients of measure zero  (see \cite[Rem. 11]{Gherekhloo2016}).
We envisage that, in a similar fashion, phase alignment achieves a strict improvement over TIN in the downlink case as well.
\hfill $\lozenge$
\end{remark}
\section{Proofs of Lemma \ref{lemma:more_capable} and Lemma \ref{lemma:diff_entropies}}
\label{appendix:proof_lemmas_more_capable_diff_entropies}
To gain some insights, we start our treatment of  Lemma \ref{lemma:more_capable} and Lemma \ref{lemma:diff_entropies} 
by looking through the lens of the well-known Avestimehr-Diggavi-Tse (ADT) linear deterministic model \cite{Avestimehr2011}.
This model separates  signal power levels, e.g. as the ones in Fig. \ref{fig:signal_levels_TIN} and Fig. \ref{fig:signal_levels_IA}, into 
parallel, non-interacting bit levels.
As  noted and observed through a number of previous works, the ADT  model is particularly useful for deriving and understanding TIN GDoF results,
as the TIN GDoF framework tends to be insensitive to details not captured by this model \cite{Geng2016,Geng2018,Sun2016,Gherekhloo2016}.
The proof steps that we develop next in the context of the ADT model are then 
translated to the original Gaussian model. 
\begin{figure}[h]
\centering
\includegraphics[width = 0.8\textwidth]{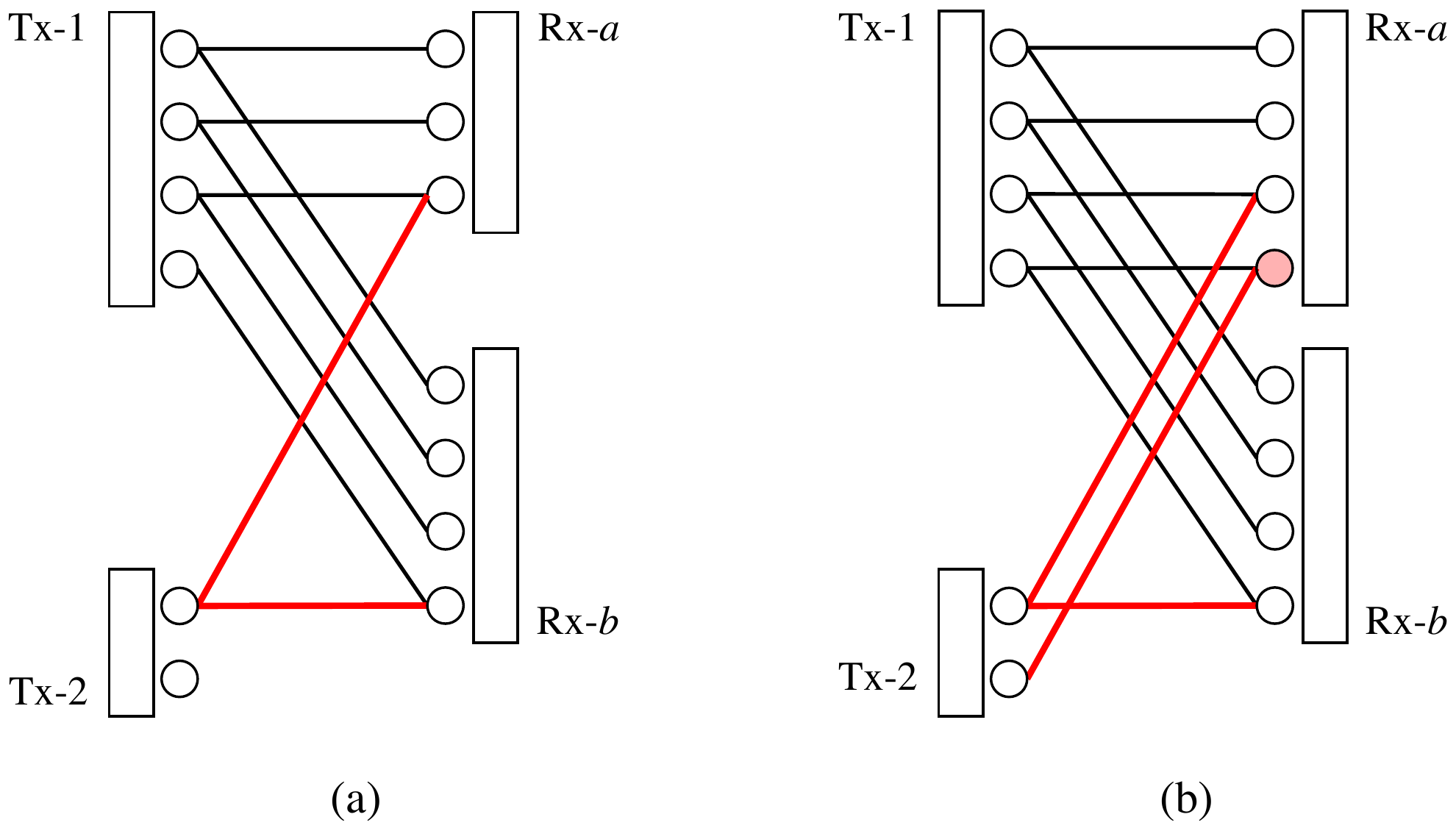}
\caption{\small
ADT linear deterministic network examples:
(a) $(m_{1},m_{2},n_{1},n_{2}) = (3,1,4,1)$, hence $n_{1} - n_{2} \geq m_{1}$, and
(b) $(m_{1},m_{2},n_{1},n_{2}) = (4,2,4,1)$,  hence $n_{1} - n_{2} < m_{1}$ and $n_{1}  - 2n_{2} \geq m_{1}  - m_{2} $. 
Note that the most-significant-bit occupies the highest bit level of each signal. In example (a),   receiver $b$ is more capable than receiver $a$.
In example (b), after removing the lowest bit level from the observation of receiver $a$ (highlighted in red), the channel reduces to the one in example (a), 
where receiver $b$ is more capable than receiver $a$.}
\label{fig:ADT_model}
\end{figure}
\subsection{ADT Linear Deterministic Model}
In the ADT deterministic model, for each $i \in \{1,2\}$, the channel strength levels $\alpha_{i}$ and $\beta_{i}$
map into $m_{i}$ and $n_{i}$, respectively, where $m_{i}$ and $n_{i}$ are non-negative integers, $m_{1} \geq m_{2}$ and $n_{1} \geq n_{2}$. 
The signal model corresponding to \eqref{eq:signal_model_y_a} and \eqref{eq:signal_model_y_b}
is therefore described by 
\begin{align}
\label{eq:ADT_model_1}
\mathbf{y}_{a}  & =  \mathbf{S}^{q - m_{1}}  \mathbf{x}_{1} \oplus \mathbf{S}^{q- m_{2}} \mathbf{x}_{2}  \\
\label{eq:ADT_model_2}
\mathbf{y}_{b}  & =   \mathbf{S}^{q- n_{1}} \mathbf{x}_{1} \oplus \mathbf{S}^{q- n_{2}} \mathbf{x}_{2}.
\end{align}
In the above, $\mathbf{y}_{a}$, $\mathbf{y}_{b}$, $\mathbf{x}_{1}$, $\mathbf{x}_{2}$ are binary column vectors of length $q$ each,
$\mathbf{S}$ is a down-shift matrix of size $q \times q$, and $q \triangleq \max\{m_{1},m_{2},n_{1},n_{2}\}$ 
(see \cite{Avestimehr2011,Bresler2008} for a detailed exposition of the ADT model).

Next, we look at the two regimes of interest, that correspond to the conditions in  Lemma \ref{lemma:more_capable} and 
Lemma \ref{lemma:diff_entropies}.
We observe that in such regimes, we have $q = n_{1}$.
For ease of exposition, we assume that $W$ and $\mathbf{x}_{1}$ are equivalent when treating Lemma \ref{lemma:more_capable},
while we ignore the conditioning on $W$ when treating Lemma \ref{lemma:diff_entropies}.
Moreover, we consider only a single use of the deterministic channel, and we 
use $\mathbf{x}(i:j)$ to denote  the vector comprising entries of  $\mathbf{x}$ which are indexed by $\langle i :  j \rangle$.
\begin{enumerate}
\item $n_{1} - n_{2} \geq m_{1}$: 
Suppose that both receiver $a$ and receiver $b$ wish to decode $\mathbf{x}_{1}$, while $\mathbf{x}_{2}$ is seen as interference.
Receiver $a$ is enhanced by providing the  interfering signal $\mathbf{x}_{2}$, leading to:
\begin{align}
I \big( \mathbf{x}_{1} ; \mathbf{y}_{a} \big) & \leq  I \big( \mathbf{x}_{1} ; \mathbf{y}_{a} |  \mathbf{x}_{2} \big) \\ 
& = H \big( \mathbf{y}_{a} |  \mathbf{x}_{2} \big) \\
& =  H \big( \mathbf{x}_{1}(1:m_{1}) \big).
\end{align}
On the other hand, since $n_{1} \geq m_{1} + n_{2}$, the observation of receiver $b$ may be expressed as
\begin{equation}
\mathbf{y}_{b} = \left[ 
\begin{array}{c}
\mathbf{x}_{1}(1:m_{1})  \\
\mathbf{y}_{b}(m_{1}+1 : n_{1})
\end{array}
\right].
\end{equation}
As the upper $m_{1}$ bit levels of $\mathbf{y}_{b} $ are received free of  interference, 
all information about $\mathbf{x}_{1}$ contained in $\mathbf{y}_{a} $  can be retrieved from $\mathbf{y}_{b} $.
Formally, we have the following:
\begin{align}
I \big( \mathbf{x}_{1} ; \mathbf{y}_{b} \big) & =  I \big( \mathbf{x}_{1} ;  \mathbf{x}_{1}(1:m_{1}) \big)
+ I \big( \mathbf{x}_{1} ;   \mathbf{y}_{b}(m_{1}+1 : n_{1}) | \mathbf{x}_{1}(1:m_{1})  \big)  \\ 
& \geq   H \big( \mathbf{x}_{1}(1:m_{1}) \big) \\
\label{eq:more_capable_determanistic}
& \geq I \big( \mathbf{x}_{1} ; \mathbf{y}_{a} \big).
\end{align}
Therefore, receiver $b$ is \emph{more capable} than receiver $a$ in this regime (see Fig. \ref{fig:ADT_model}(a)).
\item $n_{1} - 2n_{2} \geq m_{1} - m_{2} $ and $n_{2} \leq m_{2}$:
Here we have $m_{1} \geq m_{2} - n_{2} \geq 0$.
Therefore, $\mathbf{y}_{a}$ may be split into upper $m_{1} - (m_{2} - n_{2})$ bit levels and lower
$(m_{2} - n_{2})$ bit levels, that is:
\begin{equation}
\mathbf{y}_{a} = \left[ 
\begin{array}{c}
\mathbf{y}_{a}(1:m_{1}-m_{2}+n_{2}) \\
\mathbf{y}_{a}(m_{1}-m_{2}+n_{2}+1 : m_{1})
\end{array}
\right] =
\left[ 
\begin{array}{c}
\mathbf{y}_{a}' \\
\mathbf{y}_{a}''
\end{array}
\right].
\end{equation}
The difference of entropies $H\big( \mathbf{y}_{a} \big)  - H\big( \mathbf{y}_{b} \big) $
is hence bounded as 
\begin{align}
H\big( \mathbf{y}_{a} \big)  - H\big( \mathbf{y}_{b} \big) & =  
H\big(  \mathbf{y}_{a}',\mathbf{y}_{a}'' \big)   - H\big( \mathbf{y}_{b} \big) \\
& \leq  
H\big( \mathbf{y}_{a}'' \big)  + H\big( \mathbf{y}_{a}' \big)   - H\big( \mathbf{y}_{b}\big)  \\
\label{eq:proof_diff_lemma_det_inequality_0}
& \leq  (m_{2} - n_{2}) + H\big( \mathbf{y}_{a}' \big)   - H\big( \mathbf{y}_{b} \big).
\end{align}
The problem of interest reduces to bounding $H\big( \mathbf{y}_{a}' \big)   - H\big( \mathbf{y}_{b} \big)$,
where $\mathbf{y}_{a}' $  is a new observation comprising the upper $m_{1} - (m_{2} - n_{2})$ bit levels of $\mathbf{y}_{a} $.
We may express  $\mathbf{y}_{a}' $ and $\mathbf{y}_{b}$ as:
\begin{equation}
\mathbf{y}_{a}' = \left[ 
\begin{array}{c}
\mathbf{x}_{1}(1:m_{1} - m_{2})  \\
\mathbf{x}_{1}(m_{1} - m_{2} + 1 : m_{1} - m_{2} + n_{2}) \oplus  
\mathbf{x}_{2}( 1 : n_{2})
\end{array}
\right]
\end{equation}
\begin{equation}
\mathbf{y}_{b} = \left[ 
\begin{array}{c}
\mathbf{x}_{1}(1:n_{1} - n_{2})  \\
\mathbf{x}_{1}(n_{1} - n_{2} + 1 : n_{1})  \oplus \mathbf{x}_{2}(1 : n_{2}) 
\end{array}
\right].
\end{equation}
For this new channel, suppose that $\mathbf{x}_{1}$ is the desired signal and $\mathbf{x}_{2}$ is the interfering signal.
It follows that $\mathbf{y}_{a}' $ sees signal and interference levels of $m_{1}' = m_{1} - m_{2} + n_{2}$ and $m_{2}' = n_{2}$, respectively, 
while $\mathbf{y}_{b} $ sees signal and interference levels of $n_{1}$ and $n_{2}$, respectively.
Since $n_{1} \geq m_{1} - m_{2} + 2n_{2} \Leftrightarrow n_{1} - n_{2} \geq m_{1}'$,  it follows from \eqref{eq:more_capable_determanistic}
in the previous part that receiver  $b$ is more capable than the new receiver $a'$ (with observation $\mathbf{y}_{a}'$),
that is:
\begin{equation}
I\big(\mathbf{x}_{1} ; \mathbf{y}_{a}' \big)   \leq
I\big( \mathbf{x}_{1} ; \mathbf{y}_{b} \big).
\end{equation}
This is used to bound $H\big( \mathbf{y}_{a}' \big)   - H\big( \mathbf{y}_{b} \big)$ as follows:
\begin{align}
H\big( \mathbf{y}_{a}' \big)   - H\big( \mathbf{y}_{b} \big) &  = I\big(\mathbf{x}_{1},\mathbf{x}_{2} ; \mathbf{y}_{a}' \big)   - 
I\big( \mathbf{x}_{1},\mathbf{x}_{2} ; \mathbf{y}_{b} \big) \\
\label{eq:proof_diff_lemma_det_inequality}
 & =  I\big(\mathbf{x}_{1} ; \mathbf{y}_{a}' \big)   - 
I\big( \mathbf{x}_{1} ; \mathbf{y}_{b} \big) \\
\label{eq:proof_diff_lemma_det_inequality_1}
 & \leq  0
\end{align}
where in the above, \eqref{eq:proof_diff_lemma_det_inequality} holds since 
$I\big(\mathbf{x}_{2} ; \mathbf{y}_{a}' | \mathbf{x}_{1}  \big)  = I\big(\mathbf{x}_{2} ; \mathbf{y}_{b} |  \mathbf{x}_{1} \big)  = 
H\big(\mathbf{x}_{2}(1:n_{2}) \big)$.
From \eqref{eq:proof_diff_lemma_det_inequality_0} and \eqref{eq:proof_diff_lemma_det_inequality_1}, 
we obtain the desired bound $H\big( \mathbf{y}_{a}\big)   - H\big( \mathbf{y}_{b} \big) \leq (m_{2} - n_{2}) $.

To summarize the above, after removing the lower $(m_{2} - n_{2})$ bit levels from the observation $\mathbf{y}_{a}$, which contribute at most 
$(m_{2} - n_{2})$ bits to the difference of entropies  $H\big( \mathbf{y}_{a}\big)   - H\big( \mathbf{y}_{b} \big)$, 
we are left with a new channel with outputs $\mathbf{y}_{a}'$ and $\mathbf{y}_{b}$.
In this new channel, receiver $b$ is more capable than the new receiver $a'$ with respect to the desired signal $\mathbf{x}_{1}$, 
and both receiver see the same level of interference from the interfering signal $\mathbf{x}_{2}$. 
Therefore, this new channel cannot contribute positively to the difference 
of entropies (see Fig. \ref{fig:ADT_model}(b)). 
\end{enumerate}
Next, we translate the insights gained from the deterministic model to the original Gaussian model.
\subsection{Proof of Lemma \ref{lemma:more_capable}}
Reverting t back to the original Gaussian setting of interest,
we start by observing that the term  $I(W ; Y_{a}^{n} ) $ is bounded above as
\begin{align}
I(W ; Y_{a}^{n} ) & \leq I(W ; Y_{a}^{n} |  X_{2}^{n} )  \\
\label{eq:I_diff_proof_UB}
& =  I(W ; a_{1}X_{1}^{n} + Z_{a}^{n} ),
\end{align}
which follows from the independence of $W$ and $X_{2}^{n}$.
Next, we bound $I(W ; Y_{b}^{n} )$ below as
\begin{align}
\label{eq:I_diff_proof_LB_0}
I(X_{1}^{n} ; Y_{b}^{n} ) & = I\Big(W  ; \frac{b_{1}}{b_{2}}X_{1}^{n} + X_{2}^{n} + \frac{1}{b_{2}}Z_{b}^{n} \Big) \\
\label{eq:I_diff_proof_LB_1}
& \geq I\Big( W  ; \frac{b_{1}}{b_{2}}X_{1}^{n} + X_{2}^{n} + Z_{b}^{n} \Big) \\
\label{eq:I_diff_proof_LB_2}
& \geq I\Big( W ; \frac{b_{1}}{b_{2}}X_{1}^{n} \! + \! Z_{b}^{n}  \Big) \! - \! I( X_{2}^{n} ; X_{2}^{n} \! + \!  Z_{b}^{n} 
) \\
\label{eq:I_diff_proof_LB_3}
& \geq I (W ; a_{1}X_{1}^{n} \! + \! Z_{b}^{n} ) \!-\! I(X_{2}^{n} ; X_{2}^{n} \! + \! Z_{b}^{n} ).
\end{align}
In the above, the inequality \eqref{eq:I_diff_proof_LB_1} holds due to $|b_{2}|^{2} \geq 1$ and $Z_{b}(t) \sim\mathcal{N}_{\mathbb{C}}(0,1)$, 
i.e. the output signal in  \eqref{eq:I_diff_proof_LB_1} is a (stochastically) degraded version of the output signal in \eqref{eq:I_diff_proof_LB_0}. 
Using a similar argument,  the inequality \eqref{eq:I_diff_proof_LB_3} holds due to $\frac{ |b_{1}|^{2} }{ |b_{2}|^{2}  } \geq |a_{1}|^{2}$.
It remains to justify the inequality in
\eqref{eq:I_diff_proof_LB_2}, which is obtained from  the chain rule as follows
\begin{align}
I\Big(W ,X_{2}^{n} ; \frac{b_{1}}{b_{2}}X_{1}^{n} + X_{2}^{n} + Z_{b}^{n} \Big) & = 
I\Big(W  ; \frac{b_{1}}{b_{2}}X_{1}^{n} + X_{2}^{n} + Z_{b}^{n} \Big)  +
I\Big(X_{2}^{n} ;  \frac{b_{1}}{b_{2}}X_{1}^{n} + X_{2}^{n} + Z_{b}^{n}  | W \Big)   \\
&= I\Big( W ; \frac{b_{1}}{b_{2}}X_{1}^{n} \! + \! Z_{b}^{n} \Big)   +
I\Big(X_{2}^{n} ; \frac{b_{1}}{b_{2}}X_{1}^{n} \! +  \! X_{2}^{n} \! +  \! Z_{b}^{n} \Big)  \\
& \geq I\Big( W ; \frac{b_{1}}{b_{2}}X_{1}^{n} + Z_{b}^{n} \Big) \\
\implies I\Big(W  ; \frac{b_{1}}{b_{2}}X_{1}^{n} + X_{2}^{n} + Z_{b}^{n} \Big) 
& \geq   I\Big( W ; \frac{b_{1}}{b_{2}}X_{1}^{n} + Z_{b}^{n} \Big) -  I\Big(X_{2}^{n} ;  \frac{b_{1}}{b_{2}}X_{1}^{n} + X_{2}^{n} + Z_{b}^{n}  | W \Big)  \\
\label{eq:I_diff_proof_LB_4}
& \geq   I\Big( W ; \frac{b_{1}}{b_{2}}X_{1}^{n} + Z_{b}^{n} \Big) -  I\Big(X_{2}^{n} ;  \frac{b_{1}}{b_{2}}X_{1}^{n} + X_{2}^{n} + Z_{b}^{n}  | W , X_{1}^{n}\Big)  \\
& =  I\Big( W ; \frac{b_{1}}{b_{2}}X_{1}^{n} + Z_{b}^{n} \Big) -  I\Big(X_{2}^{n} ;  X_{2}^{n} + Z_{b}^{n}  \Big) .
\end{align}
Note that the inequality in \eqref{eq:I_diff_proof_LB_4} holds due to the independence of $X_{2}^{n}$ and $X_{1}^{n}$. 

Equipped with \eqref{eq:I_diff_proof_UB} and \eqref{eq:I_diff_proof_LB_3}, we directly obtain
\begin{align}
I(W ; Y_{a}^{n} ) - I( W  ; Y_{b}^{n} )  & \leq I(X_{2}^{n} ; X_{2}^{n} + Z_{b}^{n} )  \\
\label{eq:lemma_I_diff_no_W}
& \leq n,
\end{align}
where \eqref{eq:lemma_I_diff_no_W} follows from the capacity of the Gaussian channels under an average power constraint.  
\subsection{Proof of Lemma \ref{lemma:diff_entropies}}
\label{appendix:subsection_proof_lemma_diff}
For the proof of this lemma, we start by omitting the conditioning on $W$  in \eqref{eq:diff_entropies} for brevity, 
and then incorporate it at the end. 
We define a degraded version of $Y_{a}$ given by
\begin{align}
Y_{a}'(t) & = g\big[ a_{1}X_{1}(t) +  a_{2}X_{2}(t) \big] + Z_{a}'(t) \\
           & = a_{1}'X_{1}(t) +  a_{2}'X_{2}(t)  + Z_{a}'(t)
\end{align}
where $g = \sqrt{P^{-(\alpha_{2} - \beta_{2})}}$ and $ Z_{a}'(t) \sim \mathcal{N}_{\mathbb{C}}(0,1)$.
Due to $\alpha_{2} \geq \beta_{2} \geq 0$, we have $|g|^{2} \leq 1$, 
from which the degradedness of $Y_{a}'$ with respect to $Y_{a}$ follows.
Striking an analogy with the deterministic model discussed earlier, 
$Y_{a}'$  can be seen as the upper $\alpha_{1} - (\alpha_{2} - \beta_{2})$ signal levels of $Y_{a}$.

For brevity, we define $X^{n} \triangleq \big(X_{1}^{n},X_{2}^{n}\big)$, which we use in the following sequence of inequalities:
\begin{align}
h\big(Y_{a}^{ n}\big)  - h \big(Y_{b}^{n} \big) & = I\big(X^{n};Y_{a}^{ n}\big)  - I \big(X^{n};Y_{b}^{n} \big) \\
\label{eq:diff_ha_hb_1}
& = I\big(X^{n};Y_{a}^{n},Y_{a}'^{n}\big)  - I \big(X^{n};Y_{b}^{n} \big) \\
& = I\big(X^{n};Y_{a}^{n}|Y_{a}'^{n}\big) + I\big(X^{n};Y_{a}'^{n}\big)  - I \big(X^{n};Y_{b}^{n} \big) \\
\label{eq:diff_ha_hb_2}
& = h\big(Y_{a}^{n}|Y_{a}'^{n}\big) - h\big(Z_{a}^{n}\big) + h\big(Y_{a}'^{n}\big)  - h \big(Y_{b}^{n} \big).
\end{align}
As remarked above, the vector input
$X$ and the scalar outputs $Y_{a}$ and $Y_{a}'$ form a (stochastically) degraded Gaussian BC,
from which we have $I\big(X^{n};Y_{a}'^{n} | Y_{a}^{n}\big) = 0$, and therefore \eqref{eq:diff_ha_hb_1} holds.
Next, we separately bound the two differences of differential entropies in \eqref{eq:diff_ha_hb_2}.
\begin{itemize}
\item
Starting with the first difference $h\big(Y_{a}^{n}|Y_{a}'^{n}\big) - h\big(Z_{a}^{n}\big)$, we have
\begin{align}
\label{eq:diff_ha_hb_1st_0}
h\big(Y_{a}^{n}|Y_{a}'^{n}\big) - h\big(Z_{a}^{n}\big)  & \leq \sum_{t = 1}^{n} \left[ h\big(Y_{a}(t)|Y_{a}'(t)\big) - h\big(Z_{a}(t)\big) \right] \\
\label{eq:diff_ha_hb_1st_1}
& \leq n \left[ h\big(Y_{a}^{\mathrm{G}}|Y_{a}'^{\mathrm{G}}\big) - h\big(Z_{a}\big) \right] \\
\label{eq:diff_ha_hb_1st_2}
& = n \log \left( \sigma^{2}_{Y_{a}^{\mathrm{G}}|Y_{a}'^{\mathrm{G}}} \right)
\end{align}
where $Y_{a}^{\mathrm{G}}$ and $Y_{a}'^{\mathrm{G}}$ denote the outputs $Y_{a}$ and $Y_{a}'$, respectively,
when the inputs are drawn from Gaussian distributions as $X_{i} = X_{i}^{\mathrm{G}} \sim \mathcal{N}_{\mathbb{C}}(0,1)$, for all $i \in \{1,2\}$.
The inequality in \eqref{eq:diff_ha_hb_1st_1} follows from \cite[Lem. 1]{Annapureddy2009}
and the i.i.d. noise assumption, where $Z_{a} \sim Z_{a}(t)$.
This Gaussianity of signals leads to \eqref{eq:diff_ha_hb_1st_2}, where
the conditional variance $\sigma^{2}_{Y_{a}^{\mathrm{G}}|Y_{a}'^{\mathrm{G}}} $  is defined as:
\begin{equation}
\label{eq:cond_var}
\sigma^{2}_{Y_{a}^{\mathrm{G}}|Y_{a}'^{\mathrm{G}}} \triangleq
\E \big[ |Y_{a}^{\mathrm{G}}|^{2} \big] - \frac{\E \big[ Y_{a}^{\mathrm{G}}Y_{a}'^{\mathrm{G}\ast} \big]
\E \big[ Y_{a}'^{\mathrm{G}}Y_{a}^{\mathrm{G}\ast} \big]}{\E \big[ |Y_{a}'^{\mathrm{G}}|^{2} \big]}.
\end{equation}
Next, we wish to calculate the terms constituting \eqref{eq:cond_var}
and bound the variance.
To this end, and for convenience, we express $Y_{a}^{\mathrm{G}}$ and $Y_{a}'^{\mathrm{G}}$ compactly as:
\begin{align}
Y_{a}^{\mathrm{G}} \triangleq \mathbf{a}^{\Hrm} X^{\mathrm{G}} + Z_{a}
\ \ \text{and} \ \
Y_{a}'^{\mathrm{G}} \triangleq g\mathbf{a}^{\Hrm} X^{\mathrm{G}} + Z_{a}'
\end{align}
where $\mathbf{a}$ and $X^{\mathrm{G}}$ are both column vectors defined as:
\begin{equation}
\mathbf{a} \triangleq 
\left[ 
\begin{array}{c}
a_{1}^{\ast}  \\
a_{2}^{\ast}
\end{array}
\right] 
\quad \text{and} \quad
X^{\mathrm{G}} \triangleq 
\left[ 
\begin{array}{c}
X_{1}^{\mathrm{G}} \\
X_{2}^{\mathrm{G}}
\end{array}
\right].
\end{equation}
It therefore follows that
\begin{align}
\sigma^{2}_{Y_{a}^{\mathrm{G}}|Y_{a}'^{\mathrm{G}}} & = 
1 + \| \mathbf{a} \|^{2} - \frac{|g|^{2} \| \mathbf{a} \|^{2} \| \mathbf{a} \|^{2}}{1 + |g|^{2} \| \mathbf{a} \|^{2}} \\
& = \frac{1 + |g|^{2} \| \mathbf{a} \|^{2} +  \| \mathbf{a} \|^{2}}{1 + |g|^{2} \| \mathbf{a} \|^{2}} \\
& \leq \frac{3\| \mathbf{a} \|^{2}}{1 + |g|^{2} \| \mathbf{a} \|^{2}}  \\
& \leq 3 P^{\alpha_{2} - \beta_{2}}.
\end{align}
Combining the above, we obtain
\begin{equation}
\label{eq:diff_ha_hb_1st}
h\big(Y_{a}^{n}|Y_{a}'^{n}\big) - h\big(Z_{a}^{n}\big) \leq n(\alpha_{2} - \beta_{2})\log(P) + n\log(3).
\end{equation}
\item Moving on to the second difference $h\big(Y_{a}'^{n}\big) - h \big(Y_{b}^{n} \big)$,
we first recall that the channel coefficients in $Y_{a}'$ and $Y_{b}$ have the following gains: $|a_{1}'|^{2} = P^{\alpha_{1} - (\alpha_{2} - \beta_{2})}$,
$|a_{2}'|^{2} = P^{\beta_{2}}$, $|b_{1}|^{2} = P^{\beta_{1}}$ and $|b_{2}|^{2} = P^{\beta_{2}}$.
Therefore, these coefficients satisfy
\begin{align}
1 \leq |a_{1}'|^{2} \leq \frac{|b_{1}|^{2}}{|b_{2}|^{2}}   & \impliedby  0 \leq \alpha_{1} - \alpha_{2} + \beta_{2} \leq \beta_{1} - \beta_{2} \\
1 \leq |a_{2}'|^{2} = |b_{2}|^{2}   & \impliedby  0 \leq \beta_{2}.
\end{align}
Building upon the insights gained from the deterministic model, 
by taking $X_{1}$ to be a desired signal and $X_{2}$ to be an interfering signal, 
a receiver observing $Y_{b}$  is less noisy (hence more capable) than a receiver observing  $Y_{a}'$, up to a constant gap. 
Hence, from Lemma \ref{lemma:more_capable}, we have 
\begin{equation}
\label{eq:more_capable_proof_diff}
I\big( X_{1}^{n}; Y_{a}'^{n}\big)  - I\big(X_{1}^{n}; Y_{b}^{n} \big)  \leq n. 
\end{equation}
Moreover, since both receivers see the same level of interference, we have 
\begin{equation}
\label{eq:same_int_proof_diff}
I\big( X_{2}^{n}; Y_{a}'^{n} |  X_{1}^{n} \big)  - I\big(X_{2}^{n}; Y_{b}^{n} | X_{1}^{n} \big)  = I\big( X_{2}^{n}; a_{2}' X_{2}^{n} + Z_{a}'^{n} \big)  - 
I\big(X_{2}^{n} ; b_{2} X_{2}^{n} + Z_{b}^{n} \big) = 0.
\end{equation}
Combining  \eqref{eq:more_capable_proof_diff} and \eqref{eq:same_int_proof_diff}, we obtain
\begin{align}
h\big(Y_{a}'^{n}\big)  - h\big(Y_{b}^{n} \big)  & = I\big( X^{n}; Y_{a}'^{n}\big)  - I\big(X^{n}; Y_{b}^{n} \big)  \\
& = I\big( X_{1}^{n}; Y_{a}'^{n}\big)  - I\big(X_{1}^{n}; Y_{b}^{n} \big)  + I\big( X_{2}^{n}; Y_{a}'^{n} |  X_{1}^{n} \big)  - I\big(X_{2}^{n}; Y_{b}^{n} | X_{1}^{n} \big)  \\
\label{eq:diff_ha_hb_2nd}
& \leq n.
\end{align}
\end{itemize}
From \eqref{eq:diff_ha_hb_2}, \eqref{eq:diff_ha_hb_1st} and \eqref{eq:diff_ha_hb_2nd}, we obtain the bound
\begin{equation}
\label{eq:diff_entropies_no_W}
h\big(Y_{a}^{ n}\big)  - h \big(Y_{b}^{n} \big) \leq
n(\alpha_{2} - \beta_{2})\log(P) + n\log(6).
\end{equation}
The only remaining part is to incorporate the conditioning on $W$ into  \eqref{eq:diff_entropies_no_W}.
For this purpose, we highlight the dependency of the outputs on $X_{1}^{n}$ as $Y_{a}^{n}(X_{1}^{n})$
and $Y_{b}^{n}(X_{1}^{n})$. 
We proceed as 
\begin{align}
h\big(Y_{a}^{ n} (X_{1}^{n}) | W \big)  - h \big(Y_{b}^{n} (X_{1}^{n})  | W \big)  & = 
\int_{w} \Big[ h\big(Y_{a}^{ n} (X_{1}^{n})  | W = w \big)  - h \big(Y_{b}^{n} (X_{1}^{n}) | W = w \big)  \Big] dF(w) \\
\label{eq:diff_entropies_with_W_proof}
& = 
\int_{w} 
\Big[ h\big(Y_{a}^{ n} (X_{1w}^{n}) \big)  - h \big(Y_{b}^{n} (X_{1w}^{n}) \big)  \Big] dF(w) \\
\label{eq:diff_entropies_with_W_proof_2}
& \leq \Big[ n(\alpha_{2} - \beta_{2})\log(P) + n\log(6) \Big] \cdot \int_{w} dF(w) \\ 
\label{eq:diff_entropies_with_W_proof_3}
& =  n(\alpha_{2} - \beta_{2})\log(P) + n\log(6),
\end{align}
where  $X_{1w}^{n} \sim X_{1}^{n}|\{W = w\}$, i.e. $X_{1w}^{n}$ is drawn from the same distribution of $X_{1}^{n}$ given $W = w$,
and \eqref{eq:diff_entropies_with_W_proof}  follows since $W$ may only change the outputs through $X_{1}^{n}$ 
(see the Markov chain in \eqref{eq:Markov_chain_lemmas}). 
Finally, we observe that or every $w$, the difference of entropies in \eqref{eq:diff_entropies_with_W_proof} is bounded above as in \eqref{eq:diff_entropies_no_W}.
Therefore, the bound in \eqref{eq:diff_entropies_with_W_proof_2} (and \eqref{eq:diff_entropies_with_W_proof_3}) holds, which completes the proof. 
\bibliographystyle{IEEEtran}
\bibliography{References}
\end{document}